\documentclass[12pt]{article}
\usepackage[T1]{fontenc}
\usepackage{array,longtable}
\usepackage{xcolor}
\usepackage{fancyhdr}
\usepackage{graphicx}
\usepackage{float}
\usepackage{algorithmic}
\usepackage[american]{babel}
\usepackage{csquotes}
\usepackage{enumitem} % for enumerate
\usepackage{subcaption} % for subfigure
\usepackage{multirow, booktabs}
\usepackage{setspace} % singlespacing
\usepackage[backend=biber,citestyle=authoryear,sortcites=true,natbib,sorting=nyt,style=apa]{biblatex}
\DeclareLanguageMapping{american}{american-apa}
\bibliography{ds-bibtex.bib}
\usepackage[
  top=2.5cm, bottom=2.5cm,
  left=2.6cm, right=2.6cm,
  headsep=3mm,
  headheight=15pt]{geometry}
\usepackage{newpxtext}
\usepackage{algorithm}
\usepackage{algorithmic}

\usepackage{pifont}% http://ctan.org/pkg/pifont
\newcommand{\cmark}{\ding{51}}%
\newcommand{\xmark}{\ding{55}}%
\usepackage{amsthm, amsmath, amssymb}
\usepackage{ifthen}
\DeclareMathOperator*{\argmax}{arg\,max}
\DeclareMathOperator*{\argmin}{arg\,min}
\DeclareMathOperator*{\tr}{tr}

\DeclareMathOperator*{\Cov}{Cov}
\DeclareMathOperator*{\Var}{Var}

\DeclareMathOperator*{\sgn}{sign}

\newcommand\IR{\mathrm{I\!R}}

\newcommand\cC{{\mathcal C}}

\newcommand\bA{{\mathbf A}}

\newcommand\bI{{\mathbf I}}

\newcommand\bfI{{\mathbf I}}

\newcommand\bfL{{\mathbf L}}

\newcommand\bfX{{\mathbf X}}

\newcommand\bfy{{\mathbf y}}

\newcommand\bbE{{\mathbb E}}

\newcommand\bbW{{\mathbb W}}

\newcommand{\one}{\boldsymbol{1}}

\ifthenelse{\isundefined{\definition}}{\newtheorem{definition}{Definition}}{}
\ifthenelse{\isundefined{\fact}}{\newtheorem{fact}{Fact}}{}
\ifthenelse{\isundefined{\theorem}}{\newtheorem{theorem}{Theorem}}{}
\ifthenelse{\isundefined{\proposition}}{\newtheorem{proposition}{Proposition}}{}
\ifthenelse{\isundefined{\corollary}}{\newtheorem{corollary}{Corollary}}{}
\ifthenelse{\isundefined{\lemma}}{\newtheorem{lemma}{Lemma}}{}

\newtheorem{remark}{Remark}
\newtheorem{assumption}{Assumption}

\newcommand{\FDR}{{\mathrm{FDR}}}
\newcommand{\FDP}{{\mathrm{FDP}}}
\newtheorem{property}{Property}
\usepackage{hyperref}
\hypersetup{
    colorlinks=true,
    linkcolor=blue,
    filecolor=magenta,      
    urlcolor=cyan,
    }

\newcommand\pois{{\mathrm{Poisson}}}
\newcommand\bern{{\mathrm{Bernoulli}}}
\newcommand{\logfc}{{\mathrm{logFC}}}
\newcommand{\logFC}{{\mathrm{logFC}}}
\newcommand{\nDE}{{\mathrm{nDE}}}

\newcommand{\supp}{\href{supp.pdf}{Appendix}}

\usepackage[affil-it]{authblk}
\title{False Discovery Rate Control via Data Splitting for Testing-after-Clustering}
\author{Lijun Wang%
}
\author{Yingxin Lin%
}
\author{Hongyu Zhao%
    \\
    \texttt{\{lijun.wang,yingxin.lin,hongyu.zhao\}@yale.edu}
}

\affil{Department of Biostatistics, Yale University, New Haven, Connecticut, USA}
\newcommand{\isarxiv}{1}
\begin{document}
\maketitle
\pagestyle{fancy}
\lhead{}
\rhead{}
\chead{}

\maketitle

\begin{abstract}

    Testing for differences in features between clusters in various applications often leads to inflated false positives when practitioners use the same dataset to identify clusters and then test features, an issue commonly known as ``double dipping''.
    % For example, in typical single-cell RNA-seq data analysis, a clustering algorithm is first conducted to find putative cell types as clusters, and then a statistical test is used to identify the differentially expressed (DE) genes between the cell clusters. 
    To address this challenge, inspired by data-splitting strategies for controlling the false discovery rate (FDR) in regressions \parencite{daiFalseDiscoveryRate2023}, we present a novel method that applies data-splitting to control FDR while maintaining high power in unsupervised clustering. We first divide the dataset into two halves, then apply the conventional testing-after-clustering procedure to each half separately and combine the resulting test statistics to form a new statistic for each feature. The new statistic can help control the FDR due to its property of having a sampling distribution that is symmetric around zero for any null feature. To further enhance stability and power, we suggest multiple data splitting, which involves repeatedly splitting the data and combining results.
    Our proposed data-splitting methods are mathematically proven to asymptotically control FDR in Gaussian settings. Through extensive simulations and analyses of single-cell RNA sequencing (scRNA-seq) datasets, we demonstrate that the data-splitting methods are easy to implement, adaptable to existing single-cell data analysis pipelines, and often outperform other approaches when dealing with weak signals and high correlations among features.

\end{abstract}

\if0\isarxiv
\spacingset{1.6} % DON'T change the spacing! (JASA template)
\fi

\section{Introduction}

% \subsection{Backgroud: Double-Dipping Issue}
Researchers nowadays often collect large amounts of data with numerous features, and a key challenge is to identify which features behave differently across distinct groups. When the groups are not predefined, a common approach is first to apply clustering to divide the data into several clusters, followed by hypothesis testing to detect differences in feature means between the groups. However, this can lead to double-dipping when the same data used for both clustering and testing. In single-cell data analysis, for example, the double-dipping arises when testing whether a gene is differentially expressed (DE) across clusters (e.g., cell types) after using the same data to define those clusters, leading to false-positive DE genes even when the cell clusters are spurious. This issue may also arise in using single-cell data to infer pseudotime trajectory during continuous biological processes, such as cell differentiation or immune responses. In this context, the double-dipping issue occurs when pseudotime is first estimated for each cell, representing its relative position along the trajectory based on the gene expression pattern and then a DE test is performed along the pseudotime to identify genes that change along the trajectory.

Recently, several attempts have been made to address the double-dipping issue in single-cell data analysis. \textcite{neufeldInferenceLatentVariable2024}'s CountSplit method splits the scRNA-seq count matrix into two count matrices (training matrix and test matrix) of the same dimensions (cells by genes) by data thinning, which is also equivalent to data fission \parencite{leinerDataFissionSplitting2023} in the Poisson case. CountSplit estimates cell clusters (or pseudotime) by applying a clustering algorithm to the training matrix, and it subsequently identifies DE genes by applying a DE test to the test matrix given the cell clusters (or pseudotime). 

% Essentially, CountSplit is developed from the selective inference framework \parencite{taylorStatisticalLearningSelective2015}. 

The second attempt for the double-dipping issue is the selective inference framework \parencite{taylorStatisticalLearningSelective2015}. One usually needs to calculate the selective $p$-values, conditioning on the clustering results. The selective $p$-values are usually hard to compute in practice. One typical solution is to modify the selective $p$-value by conditioning on extra information for computational traceability, which leads to power loss. And also, for calculating such $p$-value, one needs to specify the clustering method and data distributions. 
% For example, CountSplit is designed for the Poisson data.
% Those approaches control the selective Type I error for a test of the difference in means between a pair of clusters.
\textcite{gaoSelectiveInferenceHierarchical2022} studied the agglomeration hierarchical clustering with Gaussian assumption, and \textcite{chenSelectiveInferenceKmeans2022} extended to $k$-means clustering under the Gaussian setting. However, these two methods only considered the test of difference in the mean vector instead of tests for each single feature, which is of more interest in practice (e.g., single-cell community). \textcite{chenTestingDifferenceMeans2023} proposed CADET for testing the difference in means in a single feature between a pair of clusters obtained using hierarchical or $k$-means clustering under the Gaussian setting.

Another attempt for addressing the double-dipping issue in single-cell is inspired by the Knockoff methods \parencite{barberControllingFalseDiscovery2015,candesPanningGoldModelX2018a}, represented by \textcite{songClusterDEPostclusteringDifferential2023}'s ClusterDE. While Knockoff is originally designed for regression setting to control the FDR by generating negative control data, ClusterDE adapts this to the unsupervised setting by generating real-data-based synthetic null data with only one cluster, as a counterfactual in contrast to the real data, for evaluating the whole procedure of clustering followed by a DE test.

FDR control has been well studied in the regression setting. The traditional Benjamini-Hochberg (BH) procedure \parencite{benjaminiControllingFalseDiscovery1995} is widely used in many fields, but it might fail when features are highly correlated, and it requires $p$-values, which are challenging to construct in high dimensions. On the other hand, the Knockoff methods can account for the correlations between features, but they require nearly exact knowledge of the joint distribution of all features, potentially limiting its applicability in high dimensions. Recently, the data splitting (DS) procedure in linear regressions \parencite{daiFalseDiscoveryRate2023} and generalized linear regressions \parencite{daiScaleFreeApproachFalse2023} is another powerful but simple approach, which requires neither $p$-values nor the joint distribution of features.

% On the other hand, the double-dipping issue can also be viewed as a selective inference problem 
However, the methods developed based on regression models cannot be directly applied to the testing-after-clustering problems. One major difference between regressions and clustering is that clustering is an unsupervised task, thus we do not have a response variable as in regression models. As a result, we cannot define a relevant feature by checking the association between the response and the feature. Instead, the relevant features need to be defined through the association between features and the underlying latent variable, which needs to be estimated from the data itself. Then the double-dipping issue arises if we simply first perform a clustering on the whole dataset and then conduct testing on the same dataset again. Specifically, the double-dipping issue comes from false positives when there is no (or unclear) cluster structure. In such cases, the initial clustering step may return a superficial cluster, leading the subsequent testing step to produce many false positives. 
Also, most regression models will assume that the observations are independent and identically distributed (i.i.d.), but the clustering implicitly implies that the observations are not i.i.d. if there exists a cluster structure.

In this paper, we extend the DS and the associated multiple DS (MDS) approaches in regressions to the testing-after-clustering problems. 
% The double-dipping issue comes from false positives when there is no (or unclear) cluster structure. In such cases, the initial clustering step may return a superficial cluster, leading the subsequent testing step to produce many false positives. 
We propose a new mirror statistic to address the label-switching issue specific to clustering. An adaption of inclusion rate is proposed for multiple data splitting in the clustering setting to address potentially unstable splits. 
By applying the DS procedure, where the testing-after-clustering process is performed on each half of the data separately, this double-dipping issue can be mitigated. Specifically, if no clear cluster structure exists, the results from the two independent halves are likely to differ. Conversely, when a clear structure is present, the results from different halves tend to be consistent. By further employing the MDS approach, we can summarize the signal strength based on the consistency of results, thus effectively addressing the double-dipping issue. We provide the theoretical characterization for the power and FDR under the Gaussian models for the whole testing and clustering procedure.
% We demonstrate the applications of DS and MDS in the DE gene test across discrete cell types and along continuous pseudotime trajectory in the single-cell data analysis.

The rest of this article is organized as follows. Section~\ref{sec:mirror_reg_cluster} transfers the mirror statistics definition in \textcite{daiFalseDiscoveryRate2023}'s DS for regressions to the clustering setting. Section~\ref{sec:single_ds} describes the details of constructing a single data splitting for the testing-after-clustering, and Section~\ref{sec:mds} discusses MDS for more stable performance. Section~\ref{sec:gaussian} characterizes the DS procedure for testing-after-clustering in the Gaussian case. 
Section~\ref{sec:sim} demonstrates that MDS can achieve the best or near-best power in most cases while controlling the FDR through extensive simulations based on the ideal Gaussian and Poisson settings in the discrete and continuous settings (Sections~\ref{sec:sim_discrete} and \ref{sec:sim_traj}) or synthetic scRNA-seq data (Section~\ref{sec:sim_synthetic}). Section~\ref{sec:real} applies MDS to real scRNA-seq data for DE analysis within homogeneous cell populations (Section~\ref{sec:real_homo}) and heterogeneous cell populations (Section~\ref{sec:real_hete}), respectively. Section~\ref{sec:discussion} concludes with some remarks and potential directions.

\section{Data Splitting for FDR Control}\label{sec:ds}

% \subsection{Related work: FDR Control via Mirror Statistics in Regressions}
\subsection{Mirror Statistics: From regressions to clustering}\label{sec:mirror_reg_cluster}

Suppose a set of features $(X_1,\ldots, X_p)$ follows a $p$-dimensional distribution. Let $n$ independent observations of these features form the \emph{design matrix} $\bfX = (\bfX_1,\ldots,\bfX_p)$, where $\bfX_j = (X_{1j},\ldots,X_{nj})^\top$ is the vector containing $n$ independent realizations of feature $X_j$. In regressions, for each set of the observation $(X_{i1},\ldots, X_{ip})$, there is an associated response variable $y_i$ for $i = 1,\ldots, n$. Assume that the response variable $y$ depends only on a subset of features with the corresponding index set denoted as $S_1$. Let $p_1 = \vert S_1\vert$ and $p_0 = p-p_1$. We call $X_j$ relevant (non-null) if $j\in S_1$; otherwise, we call it a null feature. Denote the index set of the null features as $S_0$. 
The goal is to identify as many relevant features as possible with the FDR under control. Denote the selected features as $\hat S$, then we can define the FDR:
$$
\FDR = \bbE[\FDP], \quad \text{with }\FDP = \frac{\vert S_0\cap \hat S\vert}{\vert\hat S\vert \vee 1}\,.
$$

Both Knockoff-based and DS frameworks in regressions construct a mirror statistic $M_j$ for each feature $X_j$ with the following two properties:

\begin{property}[Symmetry]\label{p1}
    For a null feature, $M_j$ is symmetric around zero.
\end{property}
\begin{property}[Signal]\label{p2}
    For a relevant feature, $M_j$ is relatively large.
\end{property}

These two properties suggest an approximate upper bound on the number of false positives, where $\hat{S}$ is constructed by collecting all the $X_j$ whose corresponding statistic $M_j$ is larger than $t$:
$$
\FDP(t) = \frac{\#\{j:j\in S_0, M_j> t\}}{\#\{j: M_j > t\}\vee 1}\lesssim \frac{\#\{j:M_j < -t\}}{\#\{j:M_j > t\}\vee 1}\,,
$$
then we can use the rightmost term for FDR control.

However, in unsupervised tasks, we do not have a response variable $y$. Instead, we assume that there exists a latent variable $L$. We wish to know which features are associated with $L$. Similar to regressions, we call a relevant feature if it is associated with $L$; otherwise, we call it a null feature. Again, denote the set of relevant and null features as $S_1$ and $S_0$, respectively. Specifically, in clustering with 2 classes, $L_i\in \{1,2\}$ represents the cluster label of sample $i$; and in trajectory inference, $L_i$ is continuous pseudotime. Note that pseudotime can be viewed as ``\emph{continuous}'' cluster labels because it can be derived by connecting cluster centers after clustering gene expression data. Similarly, we define the selection set as $\hat S$, and then we can have the same definition for FDR. And the mirror statistics can be seamlessly defined in clustering settings since $M_j$ does not require a response variable, and is therefore not limited to the regression setting. 

% Instead, suppose there is a latent variable $L$. Specifically, suppose $L_i\in \{1, 2\}$ be the group membership of the $i$-th observation.

% Given a designated FDR level $q\in (0, 1)$, set the cutoff as
% $$
% \tau_q = \inf\{t> 0: \widehat{\FDP}(t) \le q\}
% $$
% and then select the features $\hat S = \{j: M_j > \tau_q\}$.

\subsection{Single Data Splitting}\label{sec:single_ds}

% Let $\bfX\in \IR^{n\times p}$ be the gene expression matrix of $p$ genes measured on $n$ cells (or generally, $p$ features measured on $n$ samples). In unsupervised tasks, we do not have a response variable, but we can still define relevant gene set $S_1$ and null gene set $S_0$, or more specifically, DE genes and non-DE genes with respect to a \emph{latent} clustering label (or pseudotime).

To conduct a data splitting procedure for clustering followed by a hypothesis testing, we divide the samples into two parts with indexes $I_1, I_2$, i.e., $\bfX^{(k)}$ constructed by the rows $I_k$ of $\bfX$, where $k=1$ or 2.
Let $\cC^{(1)}, \cC^{(2)}$ be two clustering methods on two parts of the data $\bfX^{(1)}, \bfX^{(2)}$, respectively. Then for $k=1,2$, denote $\bfL^{(k)} \triangleq \cC^{(k)}(\bfX^{(k)})$ as the clustering labels for the partial data $\bfX^{(k)}$. For each part of data $\bfX^{(k)}$, we perform an association test $T$ between each feature $j$, i.e., the column $\bfX^{(k)}_j$, and the clustering labels $\bfL^{(k)}$.
% To perform an association test between each feature $j$ and the clustering label $L^{(k)}$.
Denote the test statistic of $T$ as
\begin{equation}\label{eq:djk}
d_j^{(k)} = T(\bfX_j^{(k)}, \bfL^{(k)})\,,\quad j = 1,\ldots,p; k=1,2\,.    
\end{equation}
\begin{remark}
Since the data splitting framework is quite general, we do not assume parametric form on $\bfX$. On the other hand, a (semi)-parametric form can help better illustrate the procedure. One particular semi-parametric form can be 
    $$
    \bbE [\bfX_{ij}] = g(\beta_{0j} + \beta_{1j} L_i)\,,i=1,\ldots,n, j=1,\ldots,p\,,
    $$
    where $g$ is an unknown linking function. The association test $T$ is equivalent to test $H_0:\beta_{1j} = 0$ for each feature $j$. 
    
    This semi-parametric form indicates that $L_i$ are not necessarily discrete clustering labels, and instead can be continuous. Indeed, the ``clustering'' method $\cC^{(k)}$ can be more general, such as the first principal component of $\bfX^{(k)}$ when estimating the linear pseudotime \parencite{saelensComparisonSinglecellTrajectory2019}. For brevity, we will mainly focus on clustering with two classes, but we also demonstrate the extension to the continuous pseudotime in Section~\ref{sec:sim_traj}. It is of interest to extend to more general and complex $L_i$ (even beyond one dimension).
\end{remark}

\begin{remark}
Under the null hypothesis, the test statistic $d_j^{(k)}$ needs to be symmetric around zero. Specifically, for clustering with two classes, $d_j^{(k)}$ can be the two-sample $t$-test statistic, where its sign indicates which class dominates. Basically, $d_j^{(k)}$ measures the signal strength, and the sign indicates the class dominance, thus one can also take the signed $p$-value, multiplying the $p$-value with the sign of the mean difference, as the $d_j^{(k)}$.
\end{remark}

% To perform a DE test on gene $j$, we take the usual two-sample t-test statistic as the signal measurement, i.e.,
% \begin{equation}\label{eq:diff_d}
% d_j^{(k)} = T(\{\bfX_{ij}\mid \cC^{(k)}(\bfX^{(k)})=1\}_{i\in I_k}, \{\bfX_{ij}\mid \cC^{(k)}(\bfX^{(k)})=2\}_{i\in I_k})\,, k=1,2.
% \end{equation}
% Alternatively, one might take the $p$-value from the t-test or the mean difference to construct $d_j^{(1)}$ and $d_j^{(2)}$. Note that $p$-value is not symmetric about zero, but we can multiply the mean difference onto the $p$-value to obtain a \emph{signed} $p$-value. 

%% for DE test along the trajectory
%%For the DE test along pseudotime, we take the signed $p$-value by multiplying $\log p$-value with the mean difference before and after the middle point.
 
To combine the signals from two halves, the data splitting for regressions defines the following mirror statistic \parencite{daiFalseDiscoveryRate2023}:
\begin{equation}
M_j = \sgn(d_j^{(1)}d_j^{(2)})f(\vert d_j^{(1)}\vert,\vert  d_j^{(2)}\vert)\,,    \label{eq:orig_mirror}
\end{equation}
where function $f(u, v)$ is non-negative, symmetric about $u$ and $v$, and monotonically increasing in both $u$ and $v$. There are several choices of $f(u, v)$, such as $u+v$, $uv$ and $\min(u, v)$. We take 
$f(u, v) = u+v$, which has been shown to be optimal under certain conditions \parencite{kePowerKnockoffImpact2024}. 
% $$
% M_j = \sgn(d_j\check d_j)(\vert d_j\vert + \vert \check d_j\vert)
% $$
In other words, the mirror statistic is designed to satisfy Properties~\ref{p1} and ~\ref{p2}.

However, different from the regression setting, there is a potential label-switching issue in the clustering setting. Specifically, for clustering with two classes, the cluster labels from two parts of data might be reversed, i.e., cluster 1 of the first part might correspond to cluster 2 of the second part. 
For example, suppose gene $j$ is a relevant feature, and it is more expressed in cluster 1 than in cluster 2 based on the first part of the data. Due to the label-switching issue, however, it is more expressed in cluster 2 than in cluster 1 using the second part of the data. As a result, the signs of $d_j^{(1)}$ and $d_j^{(2)}$ are likely to differ,  making $d_j^{(1)} d_j^{(2)}$ negative but of large magnitude, which violates Property~\ref{p2}. Since all features share the same cluster labels within each part of data, then for all relevant features $S_1$, the label-switching will cause $d_j^{(1)}d_j^{(2)}, j\in S_1$ to be negative but large in magnitude; while for null features $S_0$, the label-switching does not make much difference because $d_j^{(1)}d_j^{(2)},j\in S_0$ will be randomly positive and negative with small magnitudes. Consequently, $\sum_{j=1}^p d_j^{(1)}d_j^{(2)}$ tends to be negative in the presence of label-switching. In contrast, if there is no label-switching, $\sum_{j=1}^p d_j^{(1)}d_j^{(2)}$ tend to be positive. Therefore, $\sgn(\sum_{j=1}^p d_j^{(1)}d_j^{(2)})$ serves as an indicator of label-switching, and hence we can correct the sign of \eqref{eq:orig_mirror} by multiplying $\sgn(\sum_{j=1}^p d_j^{(1)}d_j^{(2)})$ with the test statistic above. Thus, to address the label-switching issue, we propose the following mirror statistic for the clustering setting:
\begin{equation}\label{eq:new_mirror}
M_j = \sgn(d^{(1)}{}^\top d^{(2)})\sgn(d_j^{(1)} d_j^{(2)})f(\vert d_j^{(1)}\vert,\vert d_j^{(2)}\vert)\,,    
\end{equation}
where
$$
d^{(1)} = (d_1^{(1)},\ldots, d_p^{(1)})\,,\qquad d^{(2)} = (d_1^{(2)},\ldots, d_p^{(2)})\,.
$$
Proposition~\ref{prop:label_switch} formulates the label-switching issue in the Gaussian setting, and shows that we can correct the sign using $\sgn(\sum_{j=1}^p d_j^{(1)}d_j^{(2)})$ with a high probability. 

\begin{proposition}\label{prop:label_switch}
    If $d_j^{(1)}\sim N(\delta_j, \sigma^2)$, where $\delta_j \neq 0, j\in S_1$ and $\delta_j = 0, j\in S_0$,  and $d_j^{(2)} \sim N(-\delta_j, \sigma^2)$. Assume $d_j^{(k)},j=1,\ldots,p; k=1,2$ are independent.
    If $\sum_{j\in S_1}\delta_j^2 > c_1\sigma^2p^{1/2+\varepsilon}$, where $c_1 > 0$ is a constant and $\varepsilon > 0$, then $\sum_{j=1}^p d_j^{(1)}d_j^{(2)} < 0$ holds with a probability of at least
    $$
        1 - 2\exp\left(
    -\min\left\{
    \frac{\sum_{j\in S_1}\delta_j^2}{4\sigma^2},
    \frac{(\sum_{j\in S_1}\delta_j^2)^2 }{8p\sigma^4}
    \right\}
    \right)\,.
    $$
\end{proposition}

Now the proposed FDR control procedure for the testing-after-clustering task is summarized in Algorithm~\ref{alg:fdr_ds}.
\begin{algorithm}[H]
    \caption{FDR via DS for Testing-after-Clustering}
    \label{alg:fdr_ds}
    \begin{algorithmic}[1]
    \REQUIRE Data $\bfX$, a nominal FDR level $q\in (0, 1)$, and an association test $T$.
    \STATE Split the data into two parts $\bfX^{(1)}$ and $\bfX^{(2)}$.
    \STATE Conduct the testing-after-clustering procedure with test $T$ on each part of the data, and obtain the signal measurements $\{d_j^{(1)}\}_{j=1}^p$ and $\{d_j^{(2)}\}_{j=1}^p$ following \eqref{eq:djk}. The two clustering procedures can be potentially different.
    \STATE Calculate the mirror statistics $\{M_j\}_{j=1}^p$ following \eqref{eq:new_mirror}.
    \STATE Calculate the cutoff $\tau_q$ as:
    $$
    \tau_q = \min\left\{t > 0: \frac{\#\{j:M_j < -t\}}{\#\{j:M_j>t\}\vee 1 }\le q \right\}\,.
    $$
    \RETURN The features $\{j: M_j > \tau_q\}$.
    \end{algorithmic}
\end{algorithm}

Note that Proposition 2.1 of \textcite{daiFalseDiscoveryRate2023} only depends on the assumptions of the mirror statistics. Thus the conclusion still holds for the clustering setting with proper assumptions on the mirror statistics. 

\begin{assumption}[Symmetry]
    For $j\in S_0$, the sampling distribution of at least one of $\hat d_j^{(1)}$ and $\hat d_j^{(2)}$ is symmetric around zero.
\end{assumption}

\begin{assumption}[Weak dependence]
    The mirror statistics $M_j$'s are continuous random variables, and there exist constant $c > 0$ and $\alpha\in (0,2)$ such that
    $$
    var(\sum_{j\in S_0}1(M_j > t)) \le cp_0^\alpha, \forall t\in \IR, \text{ where } p_0 = \vert S_0\vert\,.
    $$
\end{assumption}

\begin{proposition}[\cite{daiFalseDiscoveryRate2023}]\label{prop:fdr_general}
    Suppose $var(M_j)$ is uniformly upper bounded and also lower bounded away from zero. For any nominal FDR level $q\in (0, 1)$, assume that there exists a constant $\tau_q > 0$ such that $P(FDP(t_q)\le q)\rightarrow 1$ as $p\rightarrow\infty$. Then, under Assumptions 1 and 2, the DS procedure satisfies
$$
FDP(t_q) \le q+o_p(1) \quad \limsup_{p\rightarrow \infty} FDR(\tau_q) \le q
$$
\end{proposition}

Figure~\ref{fig:demo_mirror} demonstrates the distribution of $\{M_j\}_{j=1}^p$ with or without cluster structure. Without cluster structure, the mirror statistics are symmetric about zero since all features are null features. With cluster structure, the mirror statistics of DE genes tend to be larger and away from null features, where the null features still exhibit a symmetric distribution about zero. Then we can properly take the cutoff to control the FDR, as shown by the red vertical line.
\begin{figure}[H]
    \centering
    \begin{subfigure}{0.5\textwidth}
        \includegraphics[width=\textwidth]{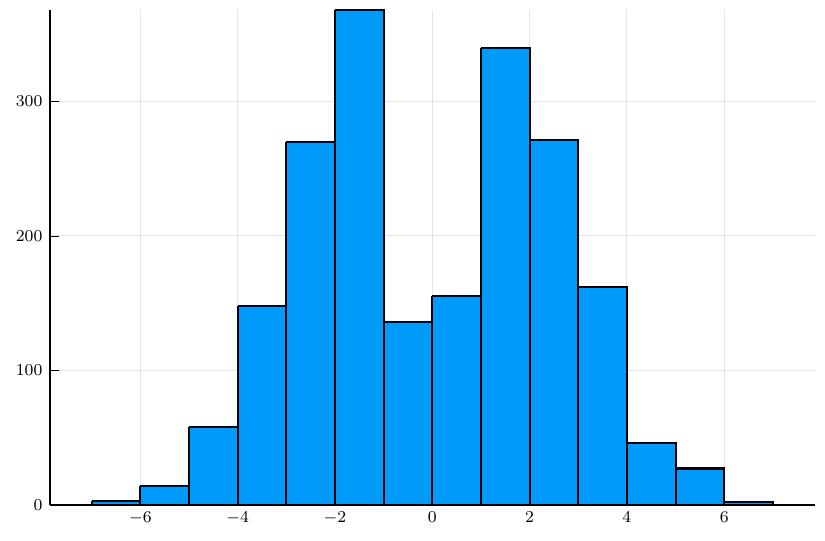}
        \caption{}
    \end{subfigure}%
    \begin{subfigure}{0.5\textwidth}
        \includegraphics[width=\textwidth]{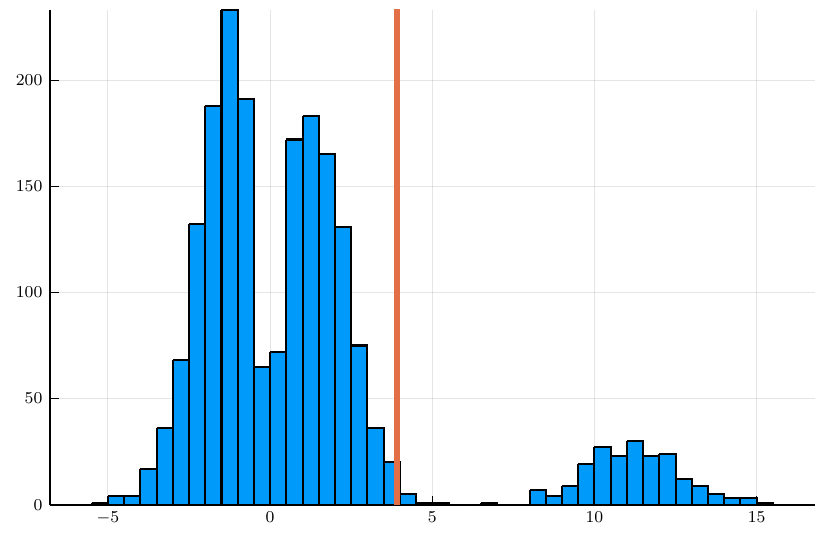}
        \caption{}
    \end{subfigure}
    \caption{Demo of mirror statistic when (\emph{a}) no cluster structure and (\emph{b}) presence of cluster structure.}
    \label{fig:demo_mirror}
\end{figure}

\subsection{Multiple Data Splitting}\label{sec:mds}

In the regression setting, there are two main concerns about a single DS \parencite{daiFalseDiscoveryRate2023}. First, splitting the data inflates the variances of the estimated regression coefficients, thus, DS can potentially suffer from a power loss in comparison with competing methods that properly use the full data. Second, the selection result of DS may not be stable and can vary substantially across different sample splits.

To address these two concerns, \textcite{daiFalseDiscoveryRate2023} proposed the multiple data splitting (MDS) procedure. Given $(\bfX, \bfy)$, suppose we independently repeat DS $m$ times with random sample splits. Each time the set of the selected features is denoted as $\hat S^{(k)}$ for $k\in \{1,\ldots, m\}$. For each feature $X_j$, define the associated inclusion rate $I_j$ and its estimate $\hat I_j$ as
\begin{equation}\label{eq:mds_avg}
I_j = \bbE\left[\frac{1(j\in \hat S)}{\vert\hat S\vert\vee 1}\mid \bfX, \bfy\right],\quad \hat I_j = \frac 1m\sum_{k=1}^m \frac{1(j\in \hat S^{(k)})}{\vert\hat S^{(k)}\vert \vee 1}\,,    
\end{equation}
in which the expectation is taken with respect to the randomness in data splitting. Intuitively, if a feature is selected frequently in the repeated data splitting, it is more likely to be a relevant feature. In other words, the inclusion rates reflect the importance of features. The cutoff of the inclusion rate is chosen as follows:
\begin{algorithm}[H]
    \caption{Multiple Data Splitting}
    \label{alg:fdr_mds}
    \begin{algorithmic}[1]
    \REQUIRE Selected features $\{\hat S^{(k)}\}_{k=1}^m$ from multiple data splitting procedures.
    \STATE Calculate the inclusion rates $\{\hat I_j\}_{j=1}^p$.
    \STATE Sort the estimated inclusion rates: $0\le \hat I_{(1)}\le \hat I_{(2)}\le\cdots\le \hat I_{(p)}$.
    \STATE Find the largest $\ell \in \{1,\ldots, p\}$ such that $\hat I_{(1)}+\hat I_{(2)}+\cdots + \hat I_{(\ell)}\le q$.
    \RETURN The features $\{j:\hat I_j > \hat I_{(\ell)}\}$.
    \end{algorithmic}
\end{algorithm}
% Note that this rate is not an estimate of the probability of being selected, but rather an importance measurement of each feature relative to the DS procedure.

These power loss and unstable concerns also exist in the clustering setting. Furthermore, there is one more concern in the clustering setting. Note that the samples are identically distributed from a joint distribution in the regression setting, but the samples are \emph{not} identically distributed, and it might lead to an unbalanced split. For example, in a very extreme case, the samples from the same class might be put into one split. In other words, a single DS might not be reliable.

One drawback of $\hat I_j$ is that it might be sensitive to the size of selection set $\vert \hat S^{(k)}\vert$ for one split. In the clustering setting, an unbalanced data split can lead to different sizes of selected features.
Alternatively, we consider 
\begin{equation}\label{eq:mds_weightavg}
\tilde I_j = \frac{\sum_{k=1}^m 1(j\in \hat S^{(k)})}{\sum_{k=1}^m \vert \hat S^{(k)}\vert \vee 1}\,,    
\end{equation}
which is robust to the size of a selected feature set.

%% TODO: argue the proposed I_j is better
% make a connection with the comparison of two estimator in importance sampling.
\begin{remark}
    The difference of $\hat I_j$ and $\tilde I_j$ can be inspired by two different estimators in the importance sampling literature. Specifically, let $X\sim f$. If $f$ is difficult to simulate from, one can instead generate $Y_1,\ldots, Y_m$ i.i.d. from $g$, then for any function $h$, one can approximate the expectation $\bbE h(X)$ with
    $$
    \frac{1}{m}\sum_{i=1}^mw_ih(Y_i)\, \qquad\text{or}\qquad \frac{\sum_{j=i}^m w_ih(Y_i)}{\sum_{j=1}^m w_j}\,,
    $$
    where $w_i = f(Y_i)/g(Y_i)$. Practically, the second estimator is more often used and is superior to the first one \parencite{casellaRaoBlackwellisationSamplingSchemes1996}.

    %% TODO: more formal comparisons??
    % For example, in some extreme cases, we might have fewer items in $\vert\hat S^{(k)}\vert$, which causes the selected item to be much more important. and on the other hand, if $\vert \hat S^{(k)}\vert =0$, then this split is not counted, and hence $\hat I_j$ tends to be smaller. and it becomes easier to exclude relevant features. 
\end{remark}
To quantify how unbalanced the data splitting could be, Proposition~\ref{prop:unbalanced_split} examines the probability distribution of the proportion of the minority class. It implies that the probability of obtaining an (extremely) unbalanced split is very small, particularly when the sample size $n$ is large. Therefore, it is not a major concern to account for the unbalanced splits during the data splitting procedure.
\begin{proposition}\label{prop:unbalanced_split}
    Suppose there are $n$ samples ($n$ is even) from two classes. Randomly split $n$ samples into two halves. Let $W$ be the proportion of the minority class of the first half, then $W\in [0, 1/2]$ and
    $$
    \Pr(W \le w) \le \exp(-(\alpha-w)^2n) + \exp(-(1-\alpha-w)^2n)\,,
    $$
    where $\alpha$ is the proportion of the first class. Particularly, if $\alpha=1/2$ and $w  = 1/2 - n^{-\gamma}, \gamma \in (0, 1/2)$, we have
    $$
    \Pr\left(\frac 12 - W \le n^{-\gamma}\right) \le 2\exp\left(-n^{1-\gamma}\right)\,.
    $$    
\end{proposition}

\section{Testing-after-Clustering under Gaussian Model}\label{sec:gaussian}
In this section, we explain why the DS procedure works both in the absence and presence of cluster structure for the testing-after-clustering problem under the Gaussian setting.

For cluster analysis, the goal is to assign close points to the same cluster, then a natural loss function is \parencite{hastieElementsStatisticalLearning2009}
\begin{equation}\label{eq:wc}
W(\cC) = \frac{1}{2}\sum_{k=1}^K \sum_{\cC(i) = k}\sum_{\cC(i')=k} d(x_i, x_{i'})\,,
\end{equation}
where $K$ is the number of clusters, $\cC(\cdot)$ is the cluster assignment and $d(\cdot, \cdot)$ is the dissimilarity measure. The k-means algorithm is one of the most popular iterative clustering methods, which minimizes $W(\cC)$ by taking $d$ as the squared Euclidean distance. If there are only two clusters $K=2$, we define a set $C = \{i: \cC(i)=1\}$ and then $-C \triangleq \{i: \cC(i) \neq 1\}$, then the loss function \eqref{eq:wc} with squared Euclidean distance can be rewritten as
$$
W(C) = \frac{1}{2}\left[ \sum_{i,i'\in C} \Vert x_i-x_{i'}\Vert^2 + \sum_{i,i'\in -C}\Vert x_i- x_{i'}\Vert^2\right]\,.
$$
Now suppose $x, x'$ are samples from a distribution, then we can define a loss function for random variables. Specifically, let $X, X'$ be two i.i.d. random variables, an expected version of the loss function can be defined as
\begin{equation}\label{eq:ws}
\bbW(C) = \frac{1}{2}\left[ \bbE \Vert X_C-X'_C\Vert^2 + \bbE\Vert X_{-C} - X'_{-C}\Vert^2\right]\,,
\end{equation}
where $X_C\triangleq X1(X\in C)$. Note that $C$ can be represented as $C = \{x: c(x) > 0\}$, where $c(\cdot)$ is a function to split the data space. The function $c(\cdot)$ can be quite complicated, such as nonlinear and discontinuous, depending on the clustering algorithms. For theoretical illustration, here we focus on the set $C$ formulated by a hyperplane $C = \{a^\top X > b: \Vert a\Vert^2 = 1\}$. 

% First of all, we can have a simplified loss function.
% \begin{lemma}
% \begin{equation}
% \argmax_S \bbW(S) = \argmax_S \bbE X_S^\top \bbE X_{-S}
% \end{equation}    
% \end{lemma}

\subsection{No Cluster Structure}
With the loss function \eqref{eq:ws}, we first study the clustering behavior when the data does not exhibit a cluster structure in the Gaussian settings, 

\begin{proposition}\label{prop:ip_cluster}
Let $X, X'$ be i.i.d. $N(0, \Sigma)$. Let $C^\star = \argmin_C \bbW(C)$.
Consider $C = \{a^\top X > 0: \Vert a\Vert^2 = 1\}$, then
\begin{itemize}
    \item if $\Sigma = \bfI_p$, the optimal hyperplane $a^\star{}^\top X > 0$ for the optimal cluster assignment $C^\star$ is not unique, i.e., $a^\star\in\{a: \Vert a\Vert^2=1\}$.
    \item if $\Sigma \neq \bfI_p$, the optimal hyperplane
    % $\sum_{j=1}^p a_j^\star X_j > 0$ for the optimal cluster assignment $S^\star$ 
    is unique, and $a^\star$ is the first eigenvector, i.e., the hyperplane is perpendicular to the direction of the first eigenvector of $\Sigma$.
\end{itemize}
    
\end{proposition}

Proposition~\ref{prop:ip_cluster} implies that when the data come from $N(0, \bfI)$, we can obtain different cluster assignments if we vary the random seed or the initialization. Thus, the selections of features $\{\hat S^{(m)}\}_{m=1}^M$ will not be consistent, meaning that there are not many common features selected among $M$ selection sets. On the other hand, if the data come from $N(0, \Sigma)$ where $\Sigma \neq \bfI$, different data splits will return the same cluster assignments, and hence the selections $\{\hat S^{(m)}\}_{m=1}^M$ might be consistent, but note that the original data $N(0, \Sigma)$ does not exhibit a cluster structure. 

To avoid false positives when $\Sigma\neq \bI$, we propose to obtain the cluster assignment from $\Sigma^{-1/2}\bfX$. Note that when there exists a cluster structure, $\Sigma^{-1/2}\bfX$ does not change the cluster structure. For example, suppose $X_1\sim N(\mu_1, \Sigma)$ and $X_2\sim N(\mu_2, \Sigma)$, where $\mu_1 \neq \mu_2$. After left-multiplying $\Sigma^{-1/2}$, which is also referred to as \emph{whitening}, we have $\Sigma^{-1/2}X_1\sim N(\Sigma^{-1/2}\mu_1, \bI)$ and $\Sigma^{-1/2}X_2\sim N(\Sigma^{-1/2}\mu_2, \bI)$. Since $\Sigma^{-1/2}$ is positive definite, we still have $\Sigma^{-1/2}(\mu_2 - \mu_1) \neq 0$. 
% \begin{remark}
% One can only use $\Sigma^{-1/2}X$ to obtain the cluster assignment and still use the original $X$ to perform hypothesis testing once we obtain the cluster assignments. If there is a cluster structure, 
% \end{remark}

% \begin{remark}
%     If $\Sigma$ is unknown, we need to estimate it. When there is no cluster structure, we can estimate $\Sigma$ using all of the data. On the other hand, if there is a cluster structure, it is problematic to estimate it. We have a positive definite $\Sigma^{-1/2}$.
% \end{remark}

\begin{remark}
    The optimal hyperplane in Proposition~\ref{prop:ip_cluster} is derived based on the expected loss \eqref{eq:ws}. In the practical finite-sample case, even when $\Sigma \neq \bfI$, the resulting clustering can be unstable, leading to inconsistent selections, then we can still avoid false positives, as shown in the correlated simulations in Section~\ref{sec:sim}.
\end{remark}

\subsection{With Cluster Structure}
Now we consider the optimal hyperplane if there exists a cluster structure with two clusters of equal proportions. 
\begin{proposition}\label{prop:2normal}
Suppose $X$ is drawn with equal probability from one of the two distributions $N(\mu_1,\Sigma)$ and $N(\mu_2, \Sigma)$. Let $X'$ be an independent copy of $X$. Consider $C = \{x: a^\top (x -  \frac{\mu_1+\mu_2}{2}) > 0\}$, then the optimal hyperplane direction is given by $a^\star = \Sigma^{-1}(\mu_2 - \mu_1)\,.$
% the optimal $S^\star = \argmin_S \bbW(S)$ is constructed by the hyperplane
%     \begin{equation}\label{eq:optimal_hyperplane}
%     \left(x-\frac{\mu_1+\mu_2}{2}\right)^\top\Sigma^{-1} (\mu_2 - \mu_1) > 0\,.         
%     \end{equation}    
\end{proposition}
\begin{remark}
    The optimal hyperplane coincides with Fisher's discriminant rule for classification with two classes (e.g., see \textcite{andersonIntroductionMultivariateStatistical2003}). Suppose $X$ is drawn with equal probability from one of the two distributions $N(\mu_1,\Sigma)$ (class 1) and $N(\mu_2, \Sigma)$ (class 2), the Fisher's rule is given by
    $$
    \hat G = \begin{cases}
        1 & \Pr(X\mid G=2)\pi_2 < \Pr(X\mid G=1)\pi_1\\
        2 & \Pr(X\mid G=2)\pi_2 > \Pr(X\mid G=1)\pi_1
    \end{cases}\,,
    $$
    where $\pi_1 = \pi_2 = 0.5$ are the prior probabilities. 

    Note that for the clustering task, we do not have prior information about the class proportions, so we focus on the equal proportion and the simple $\pi_1=\pi_2 = 0.5$ is a natural choice.
\end{remark}

% A natural way is first to estimate $\hat L$ and then for each $j$, regress $X_j$ along $\hat L$, and test if $j$ is a relevant feature (DE gene).
Based on the derived hyperplane for classification, Proposition~\ref{prop:power} presents the power for testing after clustering by taking the clustering error into account.

\begin{proposition}\label{prop:power}
Suppose $X_1,\ldots, X_m \sim N(\xi, \Sigma)$ and $Y_1,\ldots, Y_n \sim N(\eta,\Sigma)$. For any point $Z$, let $G(Z)=1$ if it comes from $X$; otherwise $G(Z)=2$. We assume that $\Sigma$ is known. 
Take the optimal hyperplane in Proposition~\ref{prop:2normal}, 
% Adopt the Fisher's linear discriminant rule to cluster $\{X_1,\ldots, X_m, Y_1,\ldots, Y_n\}$ into two clusters:
$$
\hat G(Z) = \begin{cases}
    1, & \left(Z - \frac{\xi+\eta}{2}\right)^\top\Sigma^{-1}\delta
    % + \log \frac{1-\pi}{\pi} 
    < 0\,,\\
    2, & \left(Z - \frac{\xi+\eta}{2}\right)^\top\Sigma^{-1}\delta
    % + \log \frac{1-\pi}{\pi}
    > 0\,,
\end{cases}
$$
where $\delta = \eta - \xi$. Suppose $m\le n$ and $m/n \rightarrow \kappa$ is a constant.
\begin{enumerate}[label=(\roman*)]
    \item Let $\Delta = \sqrt{\delta^\top\Sigma^{-1}\delta}$. The mis-clustering error is given by
    $$
    p_e\triangleq
\Pr(\hat G = 2\mid G=1) = \Pr(\hat G = 1\mid G=2) = \Phi\left(
    -\frac{\Delta}{2}
    \right);
    $$
% \begin{align*}
% \Pr(\hat G = 2\mid G=1) &= \Phi\left(
%     -\frac{\Delta}{2} +\frac{\log\frac{1-\pi}{\pi}}{\Delta}
%     \right);
% \Pr(\hat G = 1\mid G=2) = \Phi\left(
%     -\frac{\Delta}{2} -\frac{\log\frac{1-\pi}{\pi}}{\Delta}
%     \right);
% \end{align*}
    \item For the $j$-th relevant feature $\delta_j\neq 0$, the power of $z$-test with significant level $\alpha$ is given by
    \begin{align*}
    \beta  = %\bbE \beta(k) = 
    \sum_{k=0}^m \beta(k)\binom{m}{k}p_{e}^k (1-p_{e})^{m-k}\,,
    % \sum_{k=0}^m \beta(k)\binom{m}{k}\Phi^k\left(-\frac\Delta 2\right) \Phi^{m-k}\left(\frac\Delta 2\right)\,.
\end{align*}
with
% $$
% \beta(k) = \Phi\left(-k\frac{\delta}{\sigma}r^{1/2} +\frac{\delta}{\sigma}r^{-1/2}-c\right) + \Phi\left(k\frac{\delta}{\sigma}r^{1/2} -\frac{\delta}{\sigma}r^{-1/2}-c\right)\,.
% $$
% $$
% \beta(k) = \Phi\left(\frac{\delta_j}{\sigma_j}\frac{1-kr}{\sqrt{r}} -z_{1-\alpha/2}\right) + \Phi\left(-\frac{\delta_j}{\sigma_j}\frac{1-kr}{\sqrt r}-z_{1-\alpha/2}\right)\,,
% $$
$$
\beta(k) = \Phi\left(b_j(1-kr) -z_{1-\alpha/2}\right) + \Phi\left(-b_j(1-kr)-z_{1-\alpha/2}\right)\,,
$$
where $b_j = \frac{\delta_j}{\sigma_j\sqrt r}, r = \frac{m+n}{mn}, \sigma_j = \sqrt{\Sigma_{jj}}$ and $z_{1-\alpha/2}$ is the $1-\alpha/2$ quantile of $N(0, 1)$.

\item With a high probability at least $1-\frac{2}{m^4}$, we have
$$
\beta =
% \Phi\left(-\delta_j \left[1-\frac{mp_{e}(m+n)}{mn}\right] - z_{1-\alpha/2}\right) + \Phi\left(\delta_j \left[1-\frac{mp_{e}(m+n)}{mn}\right] - z_{1-\alpha/2}\right) + O(m^{-1})\,.
% \Phi\left(\frac{\delta_j}{\sigma_j} \frac{1-mp_{e}r }{\sqrt{r} } - z_{1-\alpha/2}\right) + \Phi\left(-\frac{\delta_j}{\sigma_j} \frac{1-mp_{e}r }{\sqrt{r} } - z_{1-\alpha/2}\right) + O(m^{-1/2}e^{-m})\,.
\Phi\left(b_j(1-\rho) - z_{1-\alpha/2}\right) + \Phi\left(-b_j(1-\rho) - z_{1-\alpha/2}\right) + O(m^{-2})\,,
$$
where $\rho \triangleq rmp_e = (1+\frac{m}{n})p_e\rightarrow (1+\kappa)p_e$.
\item Compared to the case when there is no clustering error, with a high probability at least $1-\frac{2}{m^4}$, we have
$$
\phi\left(b_j-z_{1-\alpha/2}\right)\rho b_j + O\left(m^{-2}\right)\le \beta(0) - \beta \le \phi((1-\rho)b_j-z_{1-\alpha/2})\rho b_j+ O\left(m^{-2}\right)\,.
$$
% It follows that $\beta$ increases along with the signal strength $\Vert\delta_j/\sigma_j\Vert^2$
\end{enumerate}
\end{proposition}

Proposition~\ref{prop:power} implies that the power is mainly affected by the signal strength $\delta_j$ and the noise level $\sigma_j$: stronger signal strength leads to higher power and lower noise results in higher power, which is quite intuitive. Compared to the oracle case that there is no mis-clustering error, the power loss is also decreasing along the sample size, and the dominant term $\phi\left(b_j-z_{1-\alpha/2}\right)\rho b_j$ is linear in terms of the mis-clustering error $p_e$. 
% We also note that the dominant term also decreases when the sample size goes to infinity, and then the power loss goes to zero.

Note that for the proposed DS procedure, we simply apply the testing-after-clustering procedure for each half, so the power results in Proposition~\ref{prop:power} are applicable for each half. For the FDR of the DS procedure, Proposition~\ref{prop:fdr_normal} shows that the FDR control can be (asymptotically) guaranteed.
\begin{proposition}\label{prop:fdr_normal}
    Assume $X_{i}\sim N(\mu L_i, \Sigma)$, where $\mu_j = 0, j\in S_0$, and $\mu_j=\delta_j, j\in S_1$. Let the cluster assignment $L_i \sim \text{Bernoulli} (\rho), \rho \in (0, 1)$. Randomly split the data into two parts with equal sizes. For each part, cluster the data into two clusters, denoted as $I_\ell^{(k)}, \ell=1,2;k=1,2$ for the cluster $\ell$ in the $k$-th part. Conditioning on $I^{(k)}_\ell$, the test statistic
    \begin{equation}\label{eq:zstat}
    Z_j^{(k)} = \frac{\frac{1}{\vert I_1^{(k)}\vert}\sum_{i\in I_1^{(k)}} {X_{ij}} - \frac{1}{\vert I_2^{(k)}\vert}\sum_{i\in I_2^{(k)}} {X_{ij} }}{\sqrt{\Sigma_{jj}}\sqrt{\frac{1}{\vert I_1^{(k)}\vert} + \frac{1}{\vert I_2^{(k)}\vert}}}\sim_{H_0} N(0, 1)\,, k=1,2\,.            
    \end{equation}
    Consider the mirror statistic,
    $$
    M_j = \sgn(Z^{(1)}{}^\top Z^{(2)}) \sgn(Z_j^{(1)}Z_j^{(2)})f(\vert Z_j^{(1)}\vert, \vert Z_j^{(2)}\vert)\,,
    $$
    where $Z^{(k)} = (Z^{(k}_1, \ldots, Z^{(k)}_p)$.
    Under the following assumptions:
    \begin{itemize}
        \item (Regularity condition) $1/c < \lambda_{\min}(\Sigma)\le \lambda_{\max}(\Sigma) < c$ for some $c > 0$.
        \item (Existence of $\tau_q$) For any nominal FDR level $q\in (0,1)$, there exists a constant $t_q > 0$ such that $P(\FDP(t_q)\le q)\rightarrow 1$ as $p\rightarrow\infty$. 
    \end{itemize}
    Then, the DS procedure satisfies
    $$
    \FDP(\tau_q) \le q+o_p(1)\quad\text{and}\quad \limsup_{p\rightarrow\infty}\FDR(\tau_q)\le q\,.
    $$
\end{proposition}

\begin{remark}
    Note that $I_\ell^{(k)}$ depends on $\bfX$, so it is hard to directly characterize the distribution of $Z_j^{(k)}$. Instead, we consider the conditional distribution of $Z_j^{(k)}$ given $I_{\ell}^{(k)}$. Under the null $H_0$, the resulting conditional distribution is always Gaussian, which does not involve $I_\ell^{(k)}$. 
    And the FDR can be decomposed as $\FDR = \bbE[\bbE[\FDP \mid I_\ell^{k}]]$, so we can obtain Proposition~\ref{prop:fdr_normal}. In other words, the FDR control in Proposition~\ref{prop:fdr_normal} only involves the properties of null features, whose distribution is invariant to the clustering labels.
\end{remark}

\begin{remark}
For the regularity condition on the covariance matrix, consider two special covariance structures:
\begin{itemize}
    \item $\Sigma_{ij} =\rho^{\vert j-i\vert}$. The eigenvalues are bounded, $\frac{1-\rho}{1+\rho}\le \lambda_{\min} < \lambda_{\max} \le \frac{1+\rho}{1-\rho}$ \parencite{trenchAsymptoticDistributionSpectra1999}, and hence it satisfies the assumption. 
    \item $\Sigma = \rho \one\one^T + (1-\rho)\bfI$. The eigenvalues are $\lambda_{\max}= (p-1)\rho+1$ and $\lambda_{\min} = 1-\rho$, so it does not satisfy the assumption. In that case, the DS procedure cannot be guaranteed well since the total correlation among null features is too large.
\end{itemize}    
\end{remark}

% In practice, we plug the estimated clusters $\hat I_i^{(k)}$ into the test statistic \eqref{eq:zstat}. When the signal of the cluster structure is strong, we can expect that $\hat I_i^{(k)}$ is almost the same as the true $I_i^{(k)}$. However, when the signal is relatively weak, $\hat I_i^{(k)}$ would be much different from $I_i^{(k)}$, i.e., the clustering accuracy is small.

Although Proposition~\ref{prop:fdr_normal} implies that the FDR can be controlled when the clustering label is mis-specified, the power will not be high when the clustering accuracy is low, as implied in the power loss in Proposition~\ref{prop:power}. 
% Proposition~\ref{prop:power} characterizes the power by integrating the clustering and testing steps.
% Besides the FDR, we can also characterize the power of the test-after-clustering. 

% \section{Specialization for Poisson Model}

% %TODO: This section aims to establish some theoretical results for the Poisson model.
% %% But it is not clear how to impose the dependence assumption.

% Assume $X_{ij}\sim Pois(\exp(L_i\beta_j))$. If $j\in S_0$, $\beta_j=0$. Suppose $L$ is well-estimated, one estimate is the first PC $u$ of $\log \bfX$,
% $$
% (I-\frac 1n 11^T)\log \bfX \approx duv^T\,.
% $$
% Take the estimate $\hat\beta_j$ from the Poisson GLM for $X_j\sim  L$.

% \begin{align*}
% \hat\beta_j &= \argmax_{\beta} \sum_{i=1}^n \log (p(x_{ij}))    \\
% &=\argmax_{\beta}\sum_{i=1}^n [x_{ij}L_i\beta_j - \exp(L_i\beta_j) - \log x_{ij}!]
% \end{align*}
% The MLE estimate is the solution to the following equation,
% $$
% \sum_{i=1}^n L_i(x_{ij}-\exp(L_i\beta_j)) = 0\,.
% $$

% For large $n$, the MLE estimate $\hat\beta_j$ is approximately distributed as $N(\beta, \Omega)$, thus it is (approximately) symmetric around zero if $\beta = 0$. So Assumption 1 is satisfied.

% To obtain Proposition~\ref{prop:fdr_general}, we need to further verify Assumption 2.

% % Intuitively, the essential step is to bound the correlation between $\hat\beta_i$ and $\hat\beta_j$. Asymptotically, $(\hat\beta_i, \hat\beta_j)$ is a bivariate Normal distribution.

\section{Simulations}\label{sec:sim}

In this section, we investigate the performance of the proposed approaches and other competitors for various testing-after-clustering tasks. 
Table~\ref{tab:app_method} summarizes the applicability of different methods on different single-cell data analysis tasks. Our proposed DS and MDS can handle all listed tasks, while the others are limited due to their specific designs.

% \begin{table}[H]
%     \caption{Applicability of Different Approaches}
%     \label{tab:app_method}
%     \centering
%     \begin{tabular}{cccccc}
%     \hline
%     DE test & Distribution & CADET & ClusterDE & CountSplit & MDS\\
%     \hline
%     \multirow{2}{*}{across discrete cell types}
%      & Poisson & \xmark & \cmark & \cmark & \cmark\\
%     & non-Poisson & \cmark (Gaussian) & \cmark & \xmark & \cmark\\
%     \hline
%     \multirow{2}{*}{along pseudotime trajectory} &  Poisson &\xmark & \xmark& \cmark & \cmark\\
%     & non-Poisson &\xmark &\xmark& \xmark & \cmark\\
%     \hline
%     \end{tabular}
% \end{table}

\begin{table}[H]
    \caption{Applicability of Different Approaches}
    \label{tab:app_method}
    \centering
    \begin{tabular}{lccccc}
    \hline
    Type of DE test & \multicolumn{2}{c}{across discrete cell types} & \multicolumn{2}{c}{along pseudotime trajectory}\\
    \cmidrule(lr){2-3}
    \cmidrule(lr){4-5}
    Distribution & Poisson & non-Poisson & Poisson & non-Poisson\\
    \hline
    % Naive double-dipping & \cmark &\cmark & \cmark &\cmark\\
    CADET \parencite{chenTestingDifferenceMeans2023} & \xmark & \cmark (Gaussian) & \xmark & \xmark \\
    ClusterDE \parencite{songClusterDEPostclusteringDifferential2023} & \cmark & \cmark & \xmark & \xmark\\
    CountSplit \parencite{neufeldInferenceLatentVariable2024} & \cmark & \xmark & \cmark & \xmark \\
    DS and MDS (our methods) & \cmark &\cmark & \cmark &\cmark\\
    \hline
    \end{tabular}
\end{table}

\subsection{DE across discrete cell types}\label{sec:sim_discrete}

Consider two data-generating models. 
The first one is the Gaussian model
\begin{equation*}
\bfX_{i} \sim N(\mu L_i + \varepsilon_i, \Sigma)\,, i=1,\ldots,n\,,
\end{equation*}
where 
$$
\mu = [\mu_{1},\ldots, \mu_{p}]^T, \quad \mu_j=\begin{cases}
    \delta & 1\le j\le p_1\\
        0 & p_1+1\le j \le p
\end{cases}\,,
$$
and $\Sigma_{ij} = \rho^{\vert j-i\vert}, \rho \in [0, 1)$.
% consider two covariance structures:
% \begin{itemize}
%     \item Toeplitz structure: $\Sigma_{ij} = \rho^{\vert j-i\vert}, \rho \in [0, 1)$
%     \item Fixed correlation: $\Sigma = \rho \one\one^\top + (1-\rho)\bI, \rho\in [0, 1)$.
% \end{itemize}

The second one is the Poisson model, which is adapted from \textcite{neufeldInferenceLatentVariable2024},
\begin{equation}
    \bfX_{ij} \sim \pois(\Lambda_{ij}), \quad \log(\Lambda_{ij}) = \beta_0 + L_i\beta_{1j} + \varepsilon_i,\quad \beta_{1j} = \begin{cases}
    \delta & 1\le j\le p_1\\
    0 & p_1+1 \le j\le p
\end{cases},
\label{eq:pois_x}
\end{equation}
where $\beta_0 = \log 3$ is a fixed constant to ensure the mean is not too close to zero.
Different from the independence assumption in \textcite{neufeldInferenceLatentVariable2024}, we incorporate the correlation between different features by the Gaussian copula. 
% Copulas are functions that copula (or join) multivariate distributions to their one-dimensional marginal distribution functions. Specifically, for any random vector $Z\sim \IR^d$, there exist a copula $C$ such that the CDF
% $$
% F(z) =  C(F_1(z_1),\ldots, F_d(z_d))\,.
% $$
We take the Gaussian copula
$$
C(u) = \Phi_\Sigma(\Phi^{-1}(u_1), \ldots, \Phi^{-1}(u_p))\,, 
$$
where $\Phi_\Sigma$ is the CDF of a multivariate normal distribution with covariance matrix $\Sigma$ and $\Phi$ is the CDF of a standard normal distribution.
We choose the covariance $\Sigma$ of the Gaussian copula as $\Sigma_{ij}=\rho^{\vert j-i\vert}, \rho\in [0, 1)$. We first generate $(u_1, \ldots, u_p)$ from the joint distribution specified by $C(u)$, then take $\bfX_{ij} \sim F_{ij}^{-1}(u_j)$, where $F_{ij}$ is the CDF of the marginal distribution $\pois(\Lambda_{ij})$.

In both models, $\varepsilon_i \sim N(0, \sigma_\varepsilon^2\bI)$ is the Gaussian noise, $L_i \sim\bern(0.5)$ indicates the group membership, $p_1$ is the number of relevant features, and $\delta$ quantifies the signal strength.
% For the Poisson model, for simplicity, first take the intercept $\log_{1j} = \log(3)$. (If necessary, randomly choose either $\log(3)$ or $\log(25)$ with equal probability as in the setting of \textcite{}.)
Consider $n=1000, p=2000$ (except for Figure~\ref{fig:normal-cadet-sigma0.1} for \textcite{chenTestingDifferenceMeans2023}'s CADET due to its extensive computations), $p_1/p = 0.1$.
% Let $\omega = p_1/p$ be the proportion of relevant features. Consider different levels of $\omega \in \{0.01, 0.05, 0.1\}$. 
For each data-generating model, we investigate different signal strength levels and report the FDR and power. %, precision, recall, and F1 score.

\subsubsection{Gaussian Setting}

% Since \textcite{neufeldInferenceLatentVariable2024} cannot handle Gaussian data, here we only compare the naive double-dipping method with our proposed DS procedure.
We first compare two different ways for calculating the inclusion rate: the simple average in Equation~\eqref{eq:mds_avg} and the weighted average in Equation~\eqref{eq:mds_weightavg}. It turns out that these two versions behave quite similarly in most situations, but the weighted average form is slightly better when the signal is weak and the correlation is high (see the \supp{} for more details). Thus we focus on the MDS ($M=10$) using the weighted average inclusion rate in the following experiments.

Next, we compare our proposed method with \textcite{chenTestingDifferenceMeans2023}'s CADET. We find that CADET is computationally extensive. For $n=1000, p = 2000$, CADET takes around 70 minutes to complete a single experiment whereas MDS ($M=10$) only takes 13 seconds on a standard laptop (13th Gen Intel Core i7-1360P). This makes it less practical to benchmark CADET's performance using 100 experiments. We then consider a smaller setting $n = 500, p = 1000$, where CADET takes around 7 minutes per experiment. Moreover, CADET only handles the Gaussian setting and requires an estimated covariance matrix. To avoid potential errors from covariance matrix estimation, we use the true covariance matrix directly. Using the covariance matrix, we can also apply the whitening procedure for MDS. Figure~\ref{fig:normal-cadet-sigma0.1} shows the FDR and power versus the signal strength for CADET, the naive double-dipping method, and the proposed MDS (both with whitening and without whitening). 
% The patterns for the naive method and MDS are consistent with those in Figure~\ref{fig:normal-sigma0.1}. 
MDS with whitening can achieve a better FDR control when the signal strength $\delta=0.6$ in the high correlation $\rho = 0.9$ scenario. For CADET, although it can always control FDR, it is overly conservative since the power is significantly lower than others in all scenarios. The probable reason is that the proposed selective $p$-value imposes more constraints for computational ease, sacrificing power.

\begin{figure}[H]
    \centering
    \begin{subfigure}{0.33\textwidth}
        \centering
        %%% although M50 in the filename, but actually only M=10 is used
        \includegraphics[width=\textwidth]{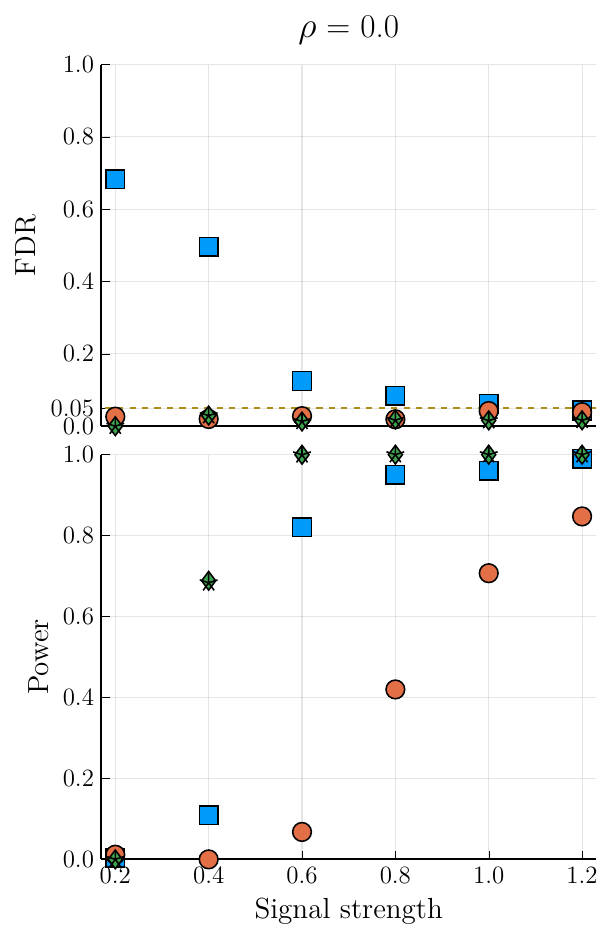}
    \end{subfigure}%
    \begin{subfigure}{0.33\textwidth}
        \centering
        \includegraphics[width=\textwidth]{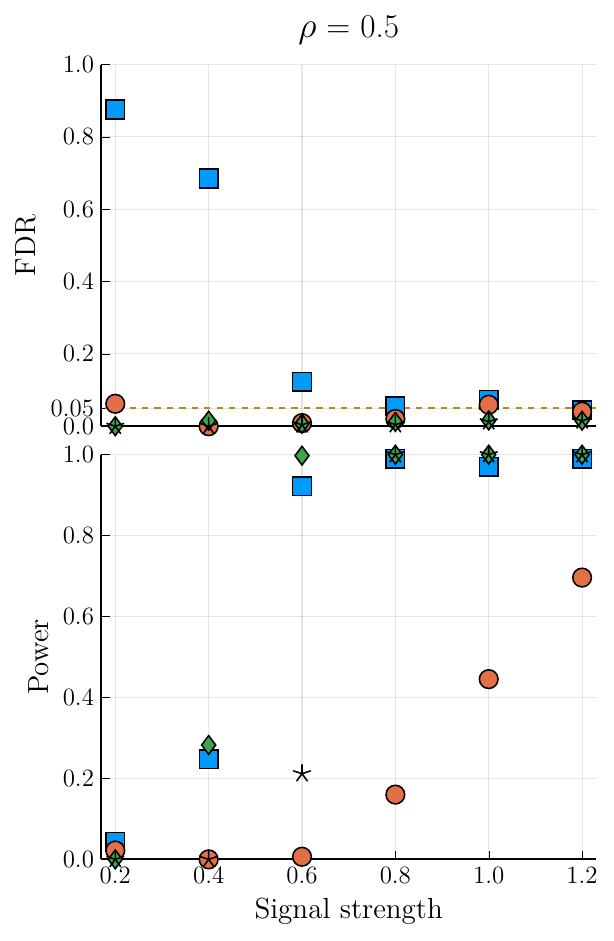}
    \end{subfigure}%
    \begin{subfigure}{0.33\textwidth}
        \centering
        \includegraphics[width=\textwidth]{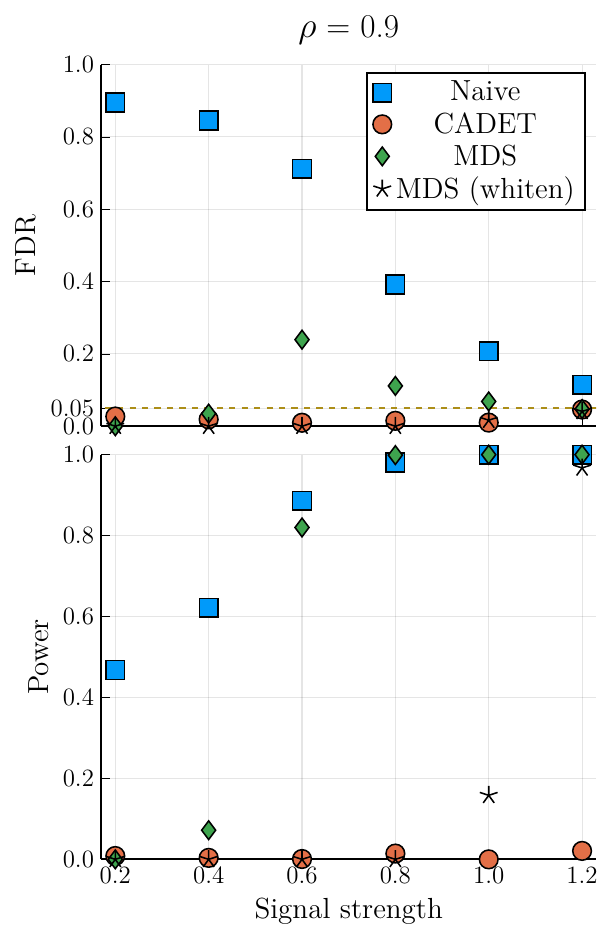}
    \end{subfigure}
    \caption{Average FDR and average power 
    % (with one standard deviation indicated by the error bar) 
    versus the signal strength of among 100 experiments under the Gaussian setting with $n=500$ samples, $p=1000$ features, $p_1=100$ relevant features and noise level $\sigma_\varepsilon = 0.1$. }
    \label{fig:normal-cadet-sigma0.1}
\end{figure}

When we increase the noise level to $\sigma_\varepsilon = 0.5$ for Figure \ref{fig:normal-cadet-sigma0.1}, the patterns remain consistent. MDS continues to effectively control the FDR and maintains high power, while the naive method fails to control FDR, and CADET loses power. Further details can be found in the \supp{}.

\subsubsection{Poisson Setting}

Next, we explore the simulation in the Poisson setting. Figure~\ref{fig:pois-sigma0.1} presents the FDR and power versus the signal strength under different correlation settings when the noise level $\sigma_\varepsilon = 0.1$.
When the correlation is strong ($\rho = 0.9$), both CountSplit and the naive double-dipping method cannot control the FDR. In contrast, our proposed MDS can achieve good power while controlling FDR. When the noise level is increased to $\sigma_\varepsilon = 0.5$ (see the \supp{}),  both CountSplit and the naive method inflate the FDR, while our proposed MDS is robust to the noise level, and can still control the FDR while maintaining high powers.

% \begin{figure}[H]
%     \centering
%     \begin{subfigure}{0.33\textwidth}
%         \centering
%         \includegraphics[width=\textwidth]{figs/2024-05-06T14_45_22-04_00-poisson-M10-nrep10-n1000-p2000-prop0.1-nsig5-0.2-0.6-rho0-bs2000-qm0.5-fdr-and-power.pdf}
%     \end{subfigure}%
%     \begin{subfigure}{0.33\textwidth}
%         \centering
%         \includegraphics[width=\textwidth]{figs/2024-05-06T14_45_22-04_00-poisson-M10-nrep10-n1000-p2000-prop0.1-nsig5-0.2-0.6-rho0.5-bs2000-qm0.5-fdr-and-power.pdf}
%     \end{subfigure}%
%     \begin{subfigure}{0.33\textwidth}
%         \centering
%         \includegraphics[width=\textwidth]{figs/2024-05-06T14_45_22-04_00-poisson-M10-nrep10-n1000-p2000-prop0.1-nsig5-0.2-0.6-rho0.9-bs2000-qm0.5-fdr-and-power.pdf}
%     \end{subfigure}
%     \caption{Poisson setting with $10\%$ relevant features with $\sigma=0$.}
%     \label{fig:pois-sigma0}
% \end{figure}

% sigma=0.1/0.2 are quite similar to sigma=0

\begin{figure}[H]
    \centering
    \begin{subfigure}{0.33\textwidth}
        \centering
        \includegraphics[width=\textwidth]{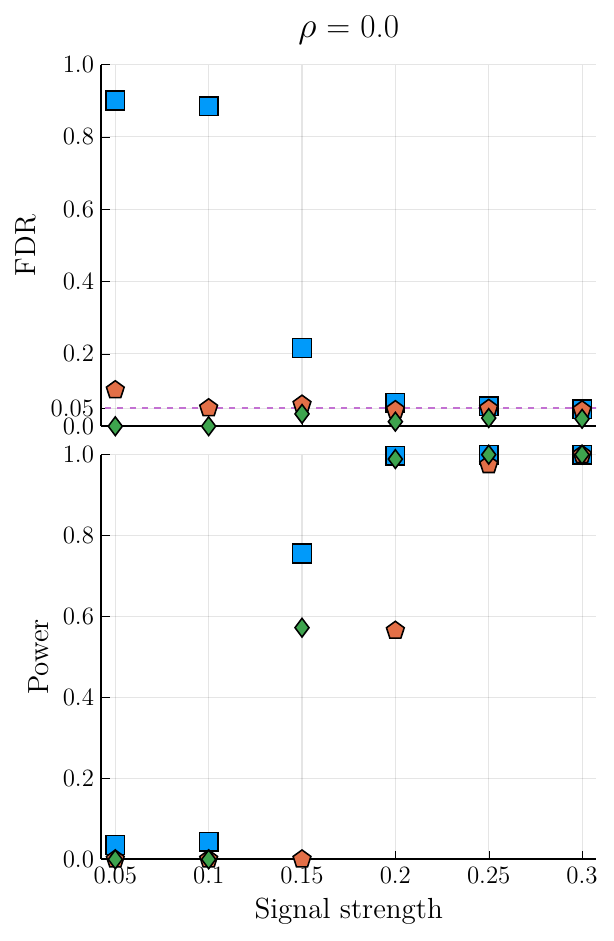}
    \end{subfigure}%
    \begin{subfigure}{0.33\textwidth}
        \centering
        \includegraphics[width=\textwidth]{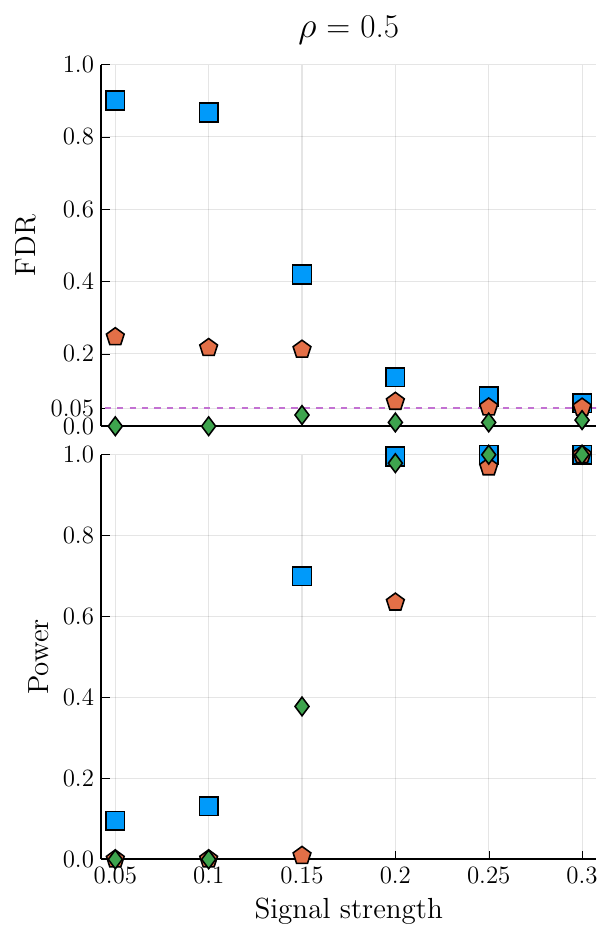}
    \end{subfigure}%
    \begin{subfigure}{0.33\textwidth}
        \centering
        \includegraphics[width=\textwidth]{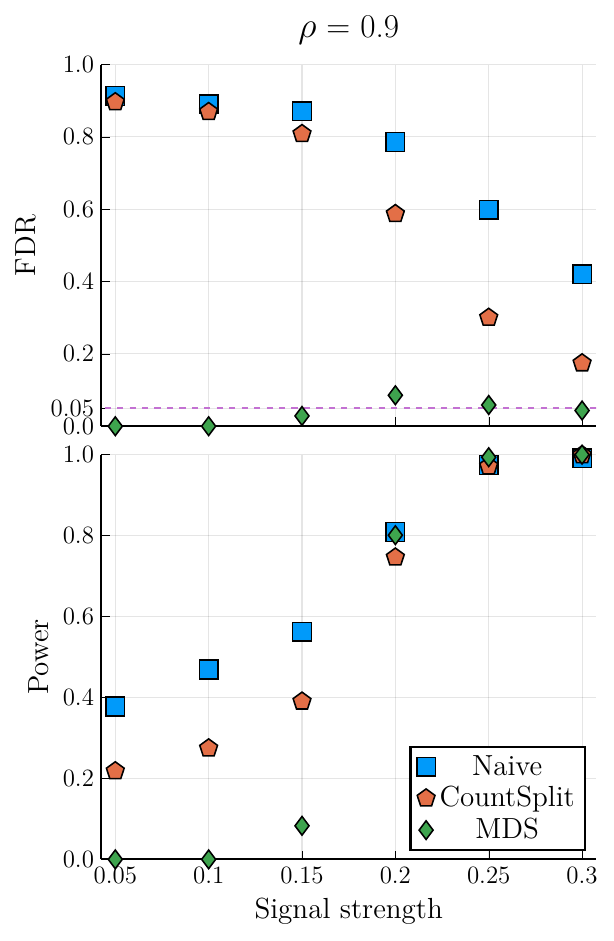}
    \end{subfigure}
    \caption{Average FDR and average power 
    % (with one standard deviation indicated by the error bar) 
    versus the signal strength among 100 experiments under the Poisson setting with $n=1000$ samples, $p=2000$ features, $p_1=200$ relevant features and noise level $\sigma_\varepsilon = 0.1$. }
    \label{fig:pois-sigma0.1}
\end{figure}

\subsection{DE along pseudotime trajectory}\label{sec:sim_traj}

To examine the performance of the proposed approaches in DE analysis along the linear trajectory, we follow the simulation setting in \textcite{neufeldInferenceLatentVariable2024}. 
\begin{equation}
L = (I_n-\frac 1n 11^\top)Z,\; Z_i \sim N(0, 1)\,.
\end{equation}
The gene expression matrix is generated from the Poisson model ~\eqref{eq:pois_x}. The trajectory pseudotime $L$ is estimated by calculating the first principal component of $X$, following the method used in \textcite{neufeldInferenceLatentVariable2024}. Besides the simplest Comp1 method, there are many trajectory inference methods for estimating the trajectory $\hat L$, and \textcite{saelensComparisonSinglecellTrajectory2019} conducted a comprehensive benchmarking study for existing trajectory inference methods, such as \textcite{jiTSCANPseudotimeReconstruction2016}'s TSCAN and \textcite{campbellUncoveringPseudotemporalTrajectories2018}'s PhenoPath. 
% \begin{itemize}
%     \item \textcite{jiTSCANPseudotimeReconstruction2016}'s TSCAN
%     \item \textcite{cannoodtSCORPIUSImprovesTrajectory2016}'s SCORPIUS
%     \item \textcite{welchMATCHERManifoldAlignment2017}'s MATCHER
%     \item \textcite{campbellUncoveringPseudotemporalTrajectories2018}'s PhenoPath
%     \item Comp1: the first principal component, which is also the method used in \textcite{neufeldInferenceLatentVariable2024}.
% \end{itemize}
%We investigated MDS ($M=10$) and CountSplit (CS) strategies on those trajectory inference methods, together with their default double-dipping approach for DE gene testing.
For brevity and without loss of generality, here we only present the performance of MDS ($M=10$), CountSplit (CS) and double-dipping approach based on the estimated pseudotime from the Comp1 method. All methods utilize the Wald test from the generalized linear model to assess the association between each gene and the pseudotime.
% , and more results based on other trajectory methods can be found in the \supp{}.

% We observed that

% \begin{itemize}
%     \item when the signal strength is large, regardless of the correlation, all methods have a good FDR control and achieve good F1 score (except for ``ouijaflow''; %this one can be removed later since some issue during running)
%     \item when the signal strength is small, both double-dipping and countsplit fail, particularly in high correlation scenarios. But DS can achieve FDR control regardless of the correlation.
%     \item a better FDR control does not necessarily imply a conservative method and a smaller F1 score. DS can still achieve good F1 scores with several trajectory inference methods, such as the ``comp1''.
% \end{itemize}

% For all pseudotime methods, calculate the p-value by the GLM with passion family by assuming that we indeed know the true distributions.

Figure~\ref{fig:pois-traj-sigma0.1} presents the FDR and power versus the signal strength of three approaches under different correlation levels. When the correlation is large $\rho = 0.9$, both the naive double-dipping approach and CountSplit fail to control FDR when the signal is weak. 
When the noise level $\sigma$ increases to 0.5, the patterns remain the same and the improvement of our proposed MDS appears to be more significant (see the \supp{}).
% Figure~\ref{fig:pois-traj-sigma0.5} shows that both the naive approach and CountSplit fail to control FDR under all investigated signal strength and correlation levels. In contrast, our proposed MDS can always achieve higher power while controlling the FDR at the nominal level $\alpha = 0.05$.

\begin{figure}[H]
    \centering
    \begin{subfigure}{0.33\textwidth}
        \centering
        \includegraphics[width=\textwidth]{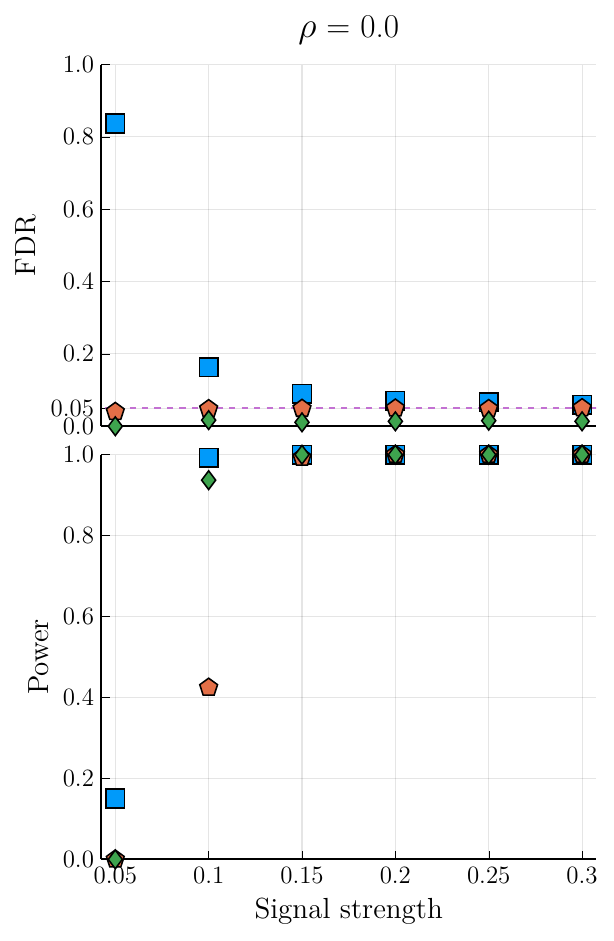}
    \end{subfigure}%
    \begin{subfigure}{0.33\textwidth}
        \centering
        \includegraphics[width=\textwidth]{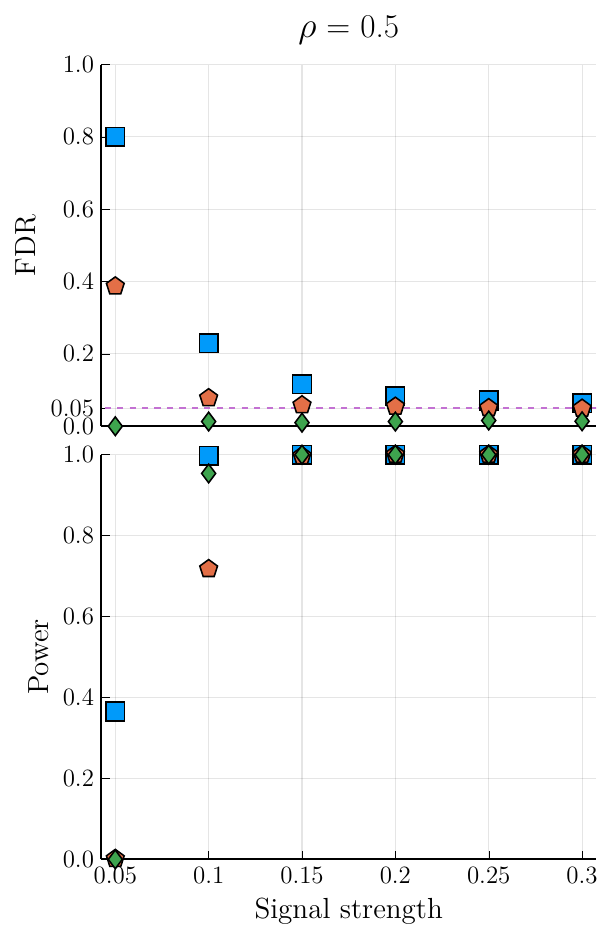}
    \end{subfigure}%
    \begin{subfigure}{0.33\textwidth}
        \centering
        \includegraphics[width=\textwidth]{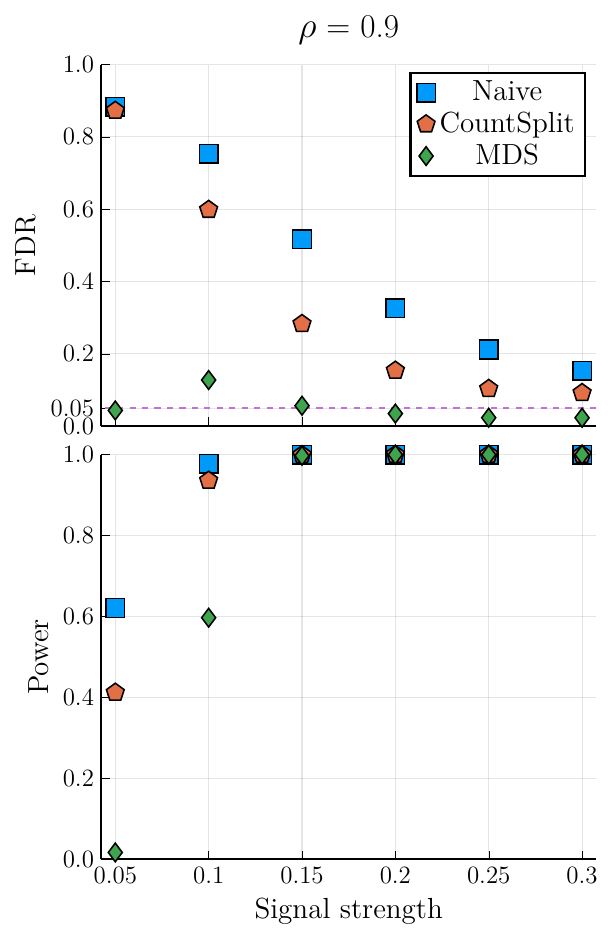}
    \end{subfigure}
    \caption{Average FDR and average power versus signal strength among 100 experiments under the linear trajectory setting with $n=1000$ samples, $p=2000$ features, $p_1=200$ relevant features and noise level $\sigma_\varepsilon=0.1$.}
    \label{fig:pois-traj-sigma0.1}
\end{figure}

\subsection{Synthetic scRNA-seq Data}\label{sec:sim_synthetic}

To benchmark the performance in more realistic scenarios, we adapted the simulation designs in \textcite{songClusterDEPostclusteringDifferential2023} to generate realistic synthetic scRNA-seq data containing true DE genes and non-DE genes, based on the model parameters learned from real scRNA-seq data.
% via the R package scDesign3 \parencite{songScDesign3GeneratesRealistic2023}.

We consider simulation settings indexed by the following three different parameters:
\begin{itemize}
    \item signal strength, which is measured by the logarithm of the fold change ($\logFC$) and ranges from 0.1 to 0.5;
    \item number of DE genes ($\nDE$), which takes 200, 400 or 800;
    \item cell type ratio between two cell types: if the ratio is $k$, then the proportion of two cell types are $\frac{1}{k+1}, \frac{k}{k+1}$, respectively. We consider three choices of $k$: 1, 2, or 4. 
\end{itemize}
Under each simulation setting, we generated 100 synthetic replicates. For each replicate, we simulated a dataset with $n = 998$ cells and $p = 9239$ genes based on the naive cytotoxic T cells in the Zhengmix4eq dataset \parencite{duoSystematicPerformanceEvaluation2020}.
Denote the mean expression for all genes as $\hat \mu$. We randomly select highly-expressed genes as DE genes, denoted by $S_1$. 
% We consider the following simulation mechanism for DE genes:
% \paragraph{One-sided up/down-regulation.}
% For $j\in S_1$, the mean expression $\mu_j^1$ for class 1 is unchanged, and only change the expression $\mu_j^2$ for class 2,
% $$
% \begin{cases}
% \mu_j^1 = \hat\mu_j\\
% \mu_j^2 = \begin{cases}
% \hat\mu_j\times 2^{\logfc}, \quad\text{if }Z_j=1\\
% \hat\mu_j\times 2^{-\logfc}, \quad\text{if }Z_j=0
% \end{cases}    
% \end{cases}
% \qquad 
% Z_j\sim \text{Bernoulli}(0.5)\,.
% $$
% \paragraph{Two-sided up-regulation.} Alternatively, we also change the mean expression of class 1, specifically,
For each $j\in S_1$, define the mean expression of gene $j$ in the $i$-th cell type as $\mu_j^i$,
$$
\begin{cases}    
\begin{cases}
\mu_j^1 =  \hat\mu_j\times 2^{\logFC}\\
\mu_j^2 =\hat\mu_j
\end{cases}
    \quad \text{if } Z_j=0\\
\begin{cases}
\mu_j^1 =  \hat\mu_j\\
\mu_j^2 = \hat\mu_j\times 2^{\logFC}
\end{cases}
\quad \text{if } Z_j=1\\
\end{cases}
\qquad Z_j\sim \text{Bernoulli}(0.5)\,,
$$
where $Z_j=1$ (or 0) indicates that gene $j$ in cell type 2 is up-regulated (or down-regulated) compared to cell type 1. We also investigate three types of hypothesis testing: t-test, Wilcox test, and Poisson test. Our proposed DS and MDS consistently achieve the best (or near-best) power while controlling FDR across all investigated scenarios and tests.
For brevity, we only present the results under different signal strengths using t-test, and other results can be found in the \supp{}.
% Compared to class 1, the expression for class 2 can also be viewed as down-regulated when we multiply $2^\logfc$ for class 2.

% \begin{remark}
%     In \textcite{songClusterDEPostclusteringDifferential2023}'s design, they randomly selected from all genes as DE genes, but practically, it might be not proper to assign genes with low expression as DE genes. And they adopted the one-sided up/down regulation mechanism.
% \end{remark}

Figure~\ref{fig:synthetic_vary_logfc} displays the actual FDR and power versus the target FDR of five different methods (DS, MDS ($M=10$), the naive double-dipping approach (DD), CountSplit (CS) and ClusterDE (CDE)) under different signal strength levels. For the naive double-dipping approach, when the signal is weak ($\logfc =0.1, 0.3$), it fails to control the FDR. In particular, when $\logfc = 0.1$, even though it can achieve a relatively higher power, the actual FDR is around 0.8. When the signal is stronger, the naive approach can maintain a higher power while controlling the FDR. On the other hand, the ClusterDE approach can control the FDR when $\logFC = 0.3, 0.5$, but it is conservative since its power curves are always smaller than others; and it cannot have a good control on FDR when $\logfc = 0.1$. In contrast, our proposed MDS procedure can always achieve a comparable power while controlling the FDR. Note that the single DS also cannot control the FDR when the signal is quite small $\logfc = 0.1$. It is also necessary to note that the mirror statistics-based approaches might be conservative when the target FDR is small (say 0.01). This is because a sufficient number of discoveries is required to achieve a nominal FDR.

\begin{figure}[H]
    \centering
    \includegraphics[width=\textwidth]{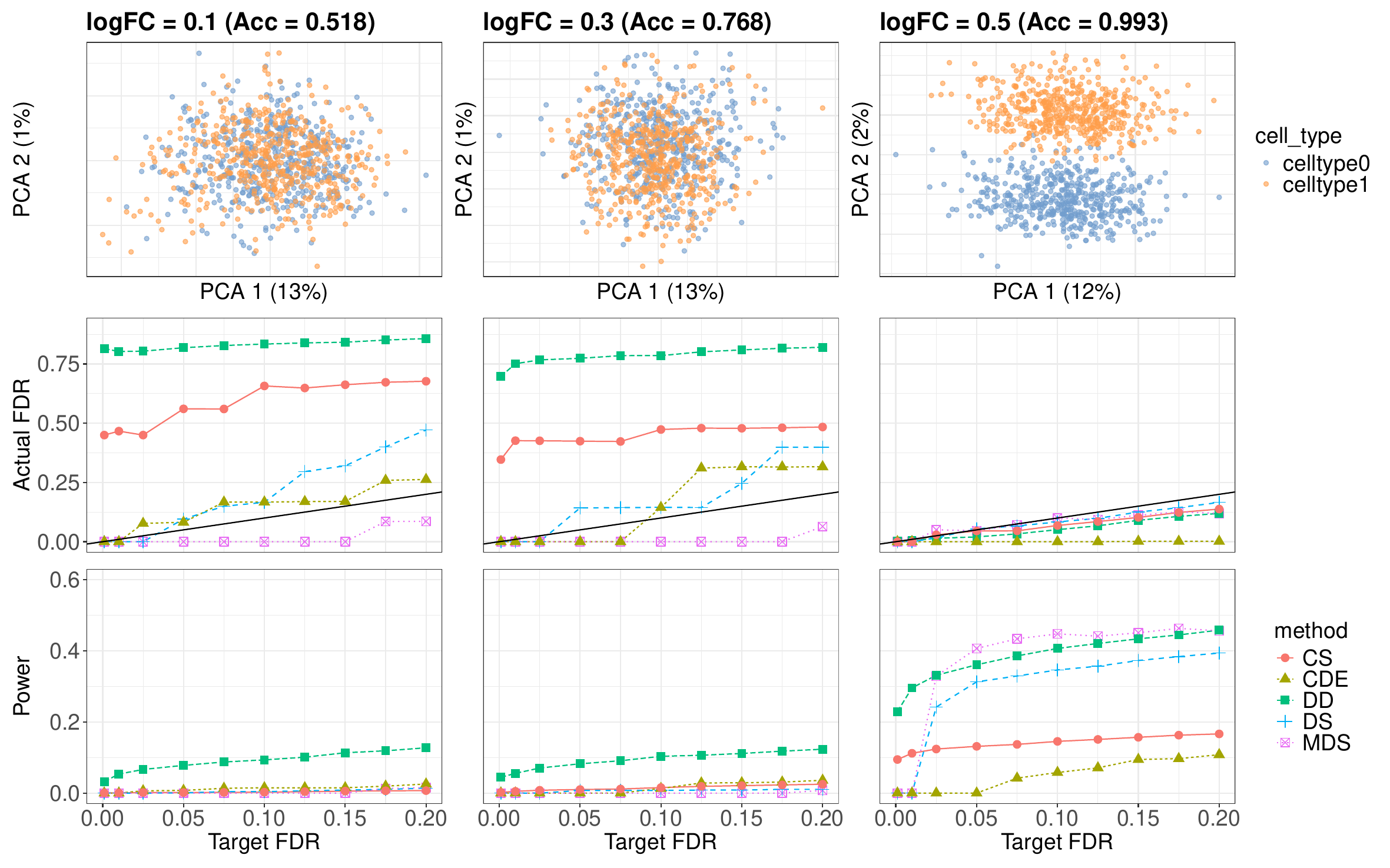}
    \caption{The actual FDR and power of five different approaches (CS: CountSplit; CDE: ClusterDE; DD: Double-dipping; DS: Data-splitting; MDS: Multiple DS) versus the target FDR under different signal strength ($\logfc$) levels when nDE = 200, cell type ratio equals 1, and using t-test. The first row displays the scatters of the first two PCs for different $\logFC$, where ``Acc'' denotes the clustering accuracy.} 
    \label{fig:synthetic_vary_logfc}
\end{figure}

%% NOTE: different covariance structure are put into the supp.tex. If necessary, we can provide it as a supp file.

%% moved nDE & cell type ratio into the appendix

\section{Real Data Application to scRNA-seq data}\label{sec:real}

In this section, we apply the proposed method to a scRNA-seq dataset from human peripheral blood mononuclear cells (PBMCs) generated by \textcite{hao2021integrated}. The dataset measures gene expression levels across 161,764 cells and 27,504 genes from eight healthy donors. A unique feature of this dataset is the simultaneous measurement of 228 surface proteins for each cell, providing additional information for more accurate cell type identification. Cell type labels in the original study were annotated by experts, combining evidence of known RNA and protein markers with unsupervised clustering results to ensure precise annotation. The cell type labels are organized into three levels of granularity: 
\begin{itemize}
    \item Level 1 represents the eight distinct broad groups of human immune cells, including CD4 T cells, CD8 T cells, Unconventional T, B cells, Natural Killer (NK) cells, Monocytes, Dendritic Cells (DC), and Other.
    \item Level 2 includes more specific subtypes of immune cells, such as ``B Memory'', ``B Naive'' and ``B Intermediate'' for B cells; ``CD8 CTL'', ``CD8 Naive'', ``CD8 Proliferating'', ``CD8 T Central Memory (TCM)'' and ``CD8 T Effector Memory (TEM)'' for CD8 T cells.
    \item Level 3 offers the highest level of granularity with 57 categories. For example, CD8 TCM are further divided into 3 subgroups, including ``CD8 TCM\_1'', ``CD8 TCM\_2'' and ``CD8 TCM\_3''. 
\end{itemize}

\subsection{DE analysis within homogeneous cell population}\label{sec:real_homo}
We first focus on the application of our method to level 3 immune cell subtypes, where each subtype can be considered as a homogeneous cell population and we expect a minimum number of DE genes detected. We consider subtypes of CD8 T cells that with a minimum of 50 cells in all eight donors and exclude subtypes that have 1,000 or more cells across all eight donors. For cell types ``CD8 TEM\_2'' and ``CD8 TEM\_4'', we observe that there are cells with low number of total counts, indicating potential low quality of cells. Therefore, we filter out the cells with less than 2500 total counts in batch 1 (donors P1$\sim$P4) and less than 4000 total counts in batch 2 (donors P5$\sim$P8). The different cutoffs for the two batches are chosen to account for the varying sequencing depths between them. This results in six subtypes for the DE analysis, with the cell counts for each donor summarized in Table~\ref{tab:cell_counts1}. Within each subtype, the genes expressed in more than 1\% of cells are kept, resulting in around 11,000 genes for the analysis. The raw gene expression data are then normalized the data by adjusting the size factor and log-transformation using Seurat \parencite{hao2021integrated}. 

\begin{table}[H]
    \centering
    \caption{Cell counts across eight donors.}%
    \begin{tabular}{lccccccccc}
    \hline
    Subtype & P1 & P2 & P3 & P4 & P5 & P6 & P7 & P8 & Total\\
    \hline
    %CD4 TCM\_1   & 1626 & 982 & 1286 & 1376 & 407 & 751 & 979 & 734 & 8141 \\
    %CD4 TCM\_3   & 710 & 446 & 721 & 806 & 602 & 1406 & 611 & 768 & 6070\\
    %CD4 TEM\_1   & 230 & 120 & 265 & 129 & 135 & 127 & 144 & 556 & 1706\\
    %CD4 TEM\_3   & 172 & 64 & 370 & 284 & 212 & 458 & 266 & 218 & 2044\\
    CD8 Naive    & 1228 & 2204 & 696 & 1762 & 1321 & 591 & 1160 & 1516 & 10478\\ 
    CD8 TCM\_1    & 125 & 94 & 98 & 149 & 63 & 191 & 83 & 126 & 929\\ 
    CD8 TCM\_2    & 72 & 51 & 68 & 579 & 95 & 101 & 95 & 261 & 1322\\ 
    CD8 TEM\_1    & 489 & 286 & 414 & 257 & 209 & 207 & 350 & 574 & 2786\\ 
    CD8 TEM\_2    & 169 & 175 & 185 & 52 & 354 & 82 & 371 & 802 & 2190\\ 
    CD8 TEM\_4    & 504 & 238 & 66 & 109 & 649 & 93 & 364 & 1074 & 3097\\ 
    \hline
    \end{tabular}
    \label{tab:cell_counts1}
\end{table}

We conduct DE analysis within each of the six subtypes and across all eight donors using five different methods: MDS, MDS (whiten), the naive double-dipping, ClusterDE, and CountSplit. This results in a total of 48 scenarios. Figure \ref{fig:onecelltype_pbmc} shows the number of DE genes identified by five methods across six subtypes of CD8 T cells for eight donors with target FDR of 5\%. We find that out of 48 scenarios, ClusterDE returns zero DE genes in 45 scenarios, followed by our proposed method, MDS (whiten) reports zero DE genes in 32 scenarios and MDS in 23 scenarios. In contrast, there are only 7 scenarios for the naive double-dipping method and 14 scenarios for CountSplit. These results indicate that under the scenario of homogeneous cell population, ClusterDE and MDS (whiten) perform effectively in avoiding false discoveries. 

We note that although the populations we analyze are assumed to be homogeneous, we observe several scenarios where MDS (whiten), CS and DD report multiple DE genes. This suggests that the population may still exhibit heterogeneity, with the presence of multiple cell states that are not identified in the original study. To further investigate this, we focus on the cell type ``CD8 TEM\_1'', where MDS (whiten) reports over a hundred of DE genes across six donors. We find that some of the reported DE genes, such as GZMA, CST7 and CCL4 (\cite{aziziSingleCellMapDiverse2018,jerby-arnonCancerCellProgram2018,yangDistinctEpigeneticFeatures2019}), are related to cytotoxicity and chemotaxis and are highly expressed in a subset of cells 
(see the \supp{}).
% Figure \ref{suppfig:cd8tem1_de}). 
This indicates these cells can be cytotoxic T cells, a subset of T cells involved in immune response to infections and cancer \parencite{kohCD8TcellSubsets2023}.

\begin{figure}[H]
    \centering
    \includegraphics[width=\textwidth]{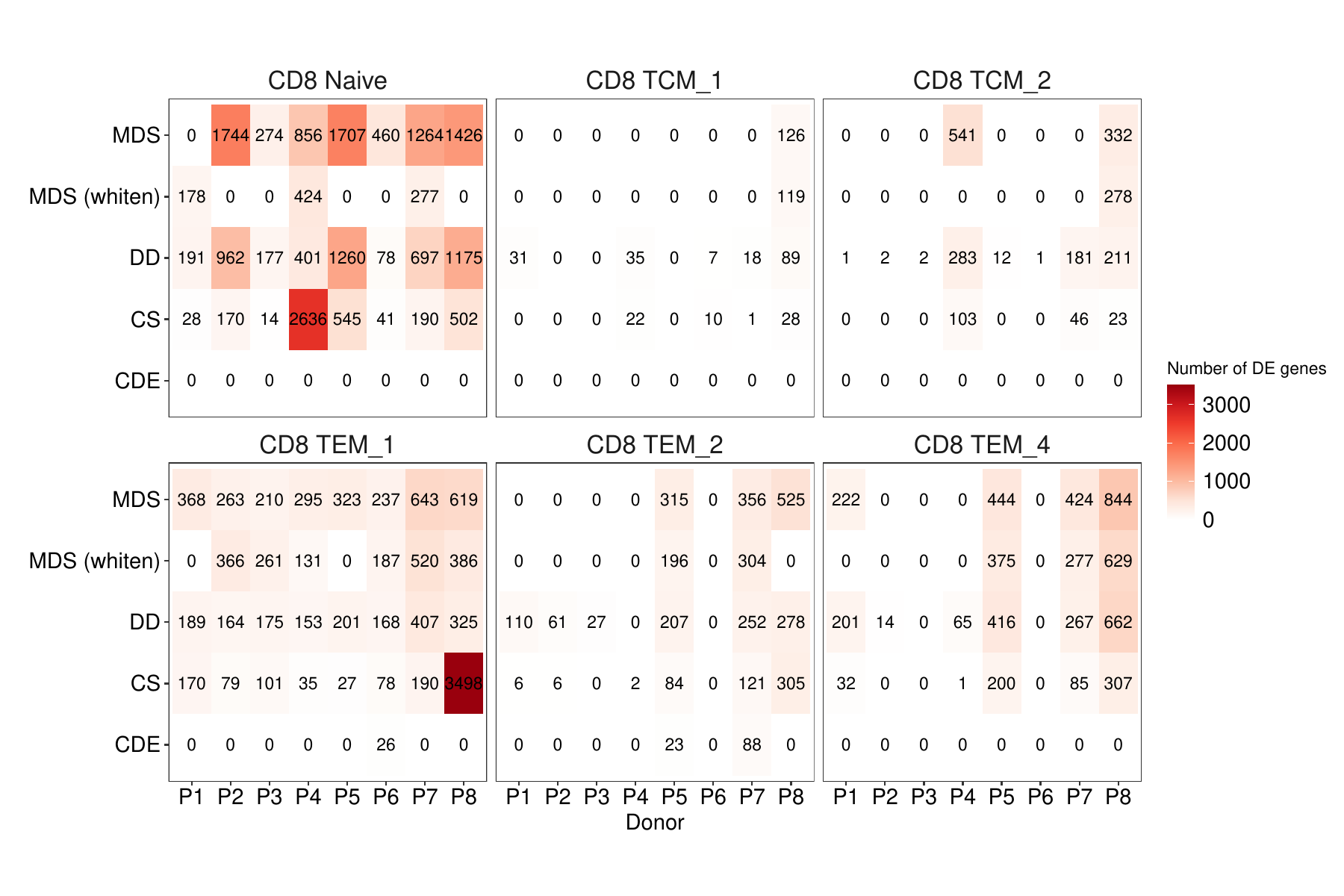}
    \caption{The number of DE genes identified using five methods for six subtypes of CD8 T cells of eight donors.} 
    \label{fig:onecelltype_pbmc}
\end{figure}

\subsection{DE analysis across heterogeneous cell population}\label{sec:real_hete}

Next, we investigate the applications in real data when the data contain two distinct cell types. Based on the level 2 cell type annotation, we focus on two scenarios: a scenario where the separation of between two cell types are less distinct: CD4 naive vs CD8 naive T cells and the another scenario where the signals to distinguish the two cell types are strong: B memory vs B naive. We perform similar data preprocessing as the case study of homogeneous cell population. To evaluate the performance of different methods in terms of their power for DE gene identification, we use the genes identified based on the labeled cell types in the DE analysis as the ``ground truth''. We then measure the number of DE genes identified by each method that overlap with the ground truth, where a high degree of consistency is expected.

Figure \ref{fig:twocelltype_T} shows the the number of genes overlapped with ground truth with varying target FDRs, for the comparison of  T cell subpopulations (CD4 naive vs. CD8 naive). We find that our proposed methods, MDS and MDS (whiten), identify very similar number of DE genes with the naive double dipping method, returning 27\% more genes overlapped than CountSplit. In contrast, in most of the scenarios, ClusterDE has very limited power in DE gene identification, especially when the signals between two cell types are weak. We also observe similar results for B cell subpopulations (B memory vs. B naive), with this case study showing stronger signals between the two cell subpopulations (see the \supp{}).

\begin{figure}[H]
    \centering
    \includegraphics[width=\textwidth]{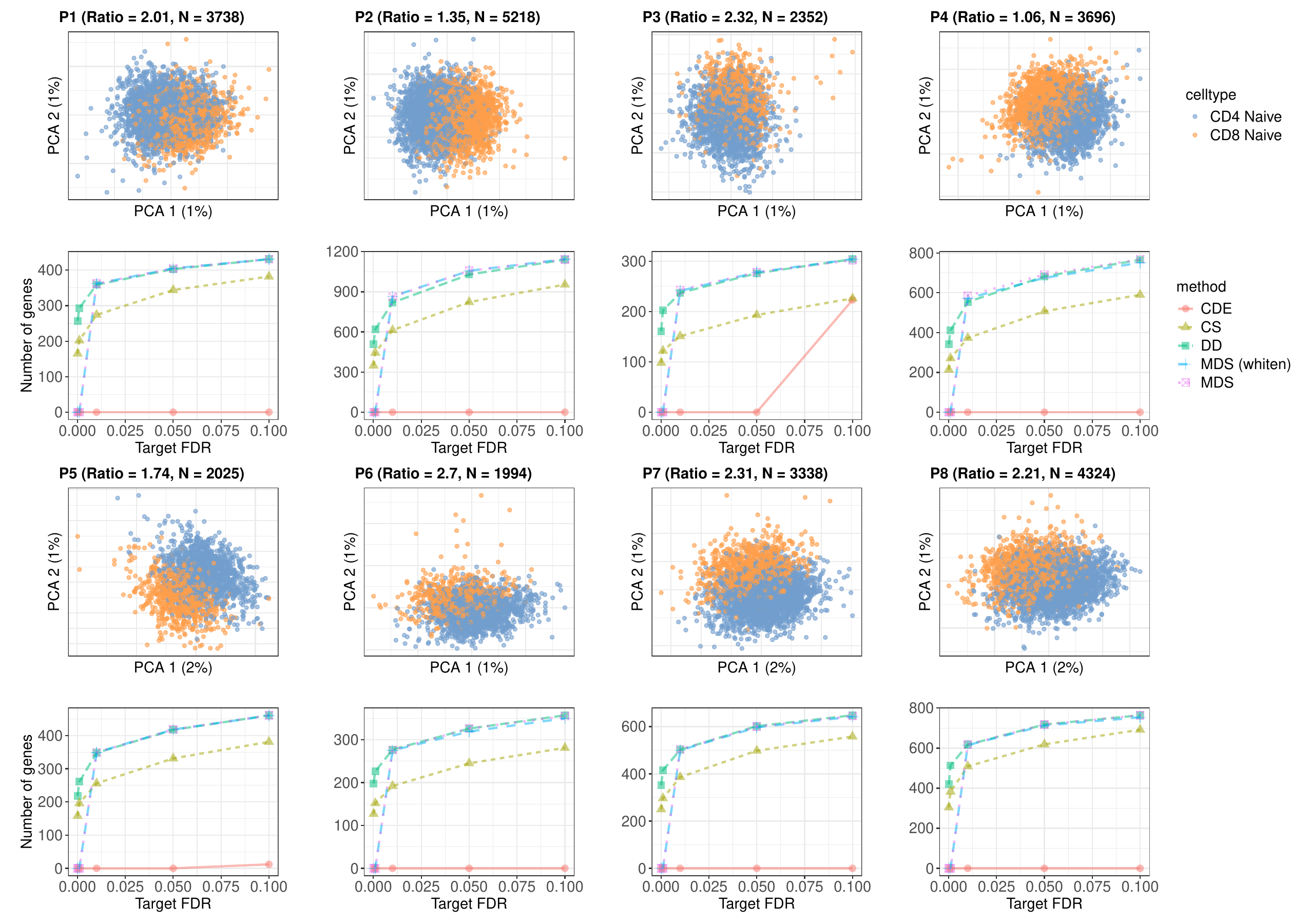}
    \caption{The number of DE genes overlapped with ground truth  using five methods across eight donors, with varying target FDRs, comparing two T cell subtypes: CD4 Naive vs. CD8 Naive. The cell type ratio and the total number of cells for each donor are indicated in the subtitle of each PCA plot.} 
    \label{fig:twocelltype_T}
\end{figure}

Together, the results from the real data application in both homogeneous and heterogeneous cell populations demonstrate that MDS achieves the best trade-off between controlling false discoveries and preserving the power to identify DE genes.

% \subsection{PBMC Data from \textcite{dingSystematicComparisonSinglecell2020}}

% Consider 10065 genes for cells from two cell types (CD14+ monocyte vs CD16+ monocyte) across 6 technologies. The cell counts are:

% \subsubsection{DE genes from the first technology}

% Take the DE genes from the 10x Chromium (v2) A\_pbmc1, then check the overlap with meaningful markers and housekeeping (nuisance) genes.

% \begin{figure}[H]
%     \centering
%     \includegraphics[width=\textwidth]{figs/ding_pbmc_2ct_overlap_1tech.pdf}
%     \caption{Overlap with meaningful markers and housekeeping (nuisance) genes for DE genes identified from 10x Chromium (v2) A\_pbmc1.}
%     \label{fig:ding_pbmc_1tech}
% \end{figure}

% \subsubsection{Common DE genes across all six technologies}

% Take the DE genes that occurred in all six technologies, then check the overlap with meaningful markers and housekeeping (nuisance) genes.
% \begin{figure}[H]
%     \centering
%     \includegraphics[width=\textwidth]{figs/ding_pbmc_2ct_overlap_withnum.pdf}
%     \caption{Overlap with meaningful markers and housekeeping (nuisance) genes for DE genes identified across all six technologies.}
%     \label{fig:ding_pbmc_6tech}
% \end{figure}

\section{Discussions}\label{sec:discussion}

We have presented a data-splitting framework for FDR control in testing-after-clustering problems to resolve the double-dipping issue by introducing a new mirror statistic for the specific label-switching issue and a weighted average inclusion rate for a more robust MDS. We also establish the theoretical guarantees for FDR control in the Gaussian settings.
Through simulations on both ideal Gaussian and Poisson models, as well as complex synthetic scRNA-seq data, we demonstrate that the proposed approaches (DS and MDS) can achieve good power while controlling the FDR, outperforming other recently proposed approaches. Using scRNA-seq data from human PBMC samples of eight donors with multi-level cell type annotations, we demonstrate that MDS and MDS (whiten) result in fewer false discoveries when analyzing homogeneous cell populations, while maintaining high power in analyses involving distinct cell types.
Both DS and MDS require no prior knowledge of the joint distributions and are easy and flexible to incorporate into existing clustering and testing frameworks. For example, in single-cell data analysis, one can directly use different normalizations, clustering and tests implemented in Seurat software for each half, and then combine the results from two halves to construct the mirror statistics.

Several directions for further developments are worth considering:
\begin{itemize}
    \item Currently, we primarily focus on datasets with two classes. While the one-vs-others strategy can be applied for multi-class settings, it would be valuable to directly address the testing and clustering in multi-class scenarios.
    \item It is important to extend this framework to samples that are not independent, such as spatial transcriptomics. Unlike classical single-cell expression data, where cells can be treated as independent, spatial transcriptomics involves spatial correlations, where nearby cells tend to be more correlated than distant ones.
    \item A key assumption of the data-splitting framework is that the correlation among null features should not be too large. However, in fields like genetics, clusters of highly correlated but null genes can occur. One possible remedy is to group these highly correlated features. More generally, the features might exhibit some group or hierarchical structures. Extending our proposed methods to accommodate such complex structures is a promising area for future research.
    \item 
    FDR control based on mirror statistics (including Knockoff-based methods) can be unstable when the number of discoveries (the denominator of FDR) is small. It also implies that a lower nominal FDR level is less reliable. In contrast, the $p$-value-based BH procedure does not suffer from this issue. It is interesting to investigate the robustness of the FDR control when there are no or quite few signals.
    \item The proposed data-splitting framework is quite general, and can be applied for the DE testing along the pseudotime. In this paper, we only demonstrate the simplest linear trajectory case, but there are many other complex trajectory patterns. Extending our methodology to these scenarios would be a valuable future direction. 
\end{itemize}

\printbibliography

\appendix
%%% put into separate supp.tex file
%%% maybe no need for arxiv version temporarily

\section{More Simulations}
% \subsubsection{Other Noise Levels}

\subsection{Two ways for inclusion rate: MDS vs MDS\_avg}

\begin{figure}[H]
    \centering
    \begin{subfigure}{0.33\textwidth}
        \centering
        %%% although M50 in the filename, but actually only M=10 is used
        \includegraphics[width=\textwidth]{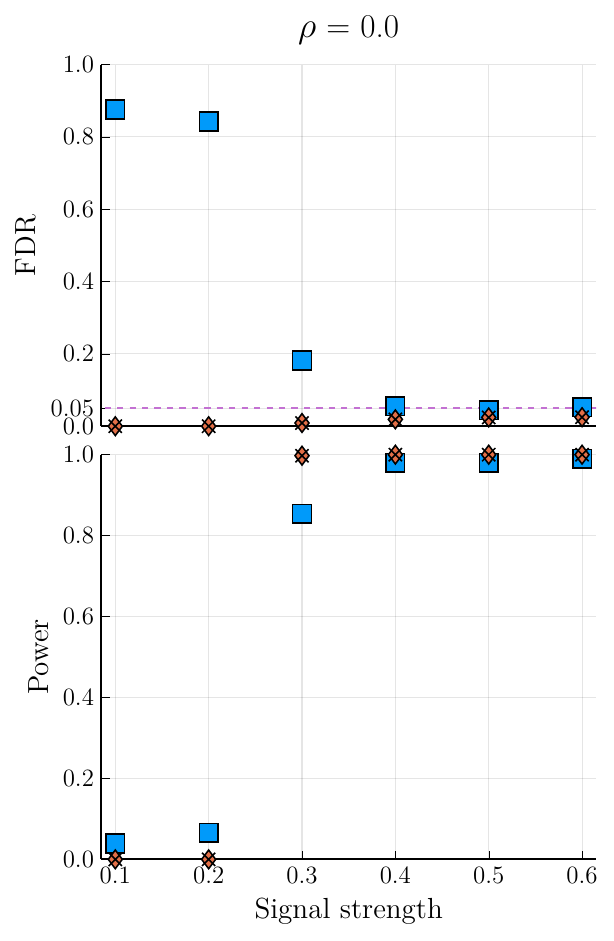}
    \end{subfigure}%
    \begin{subfigure}{0.33\textwidth}
        \centering
        \includegraphics[width=\textwidth]{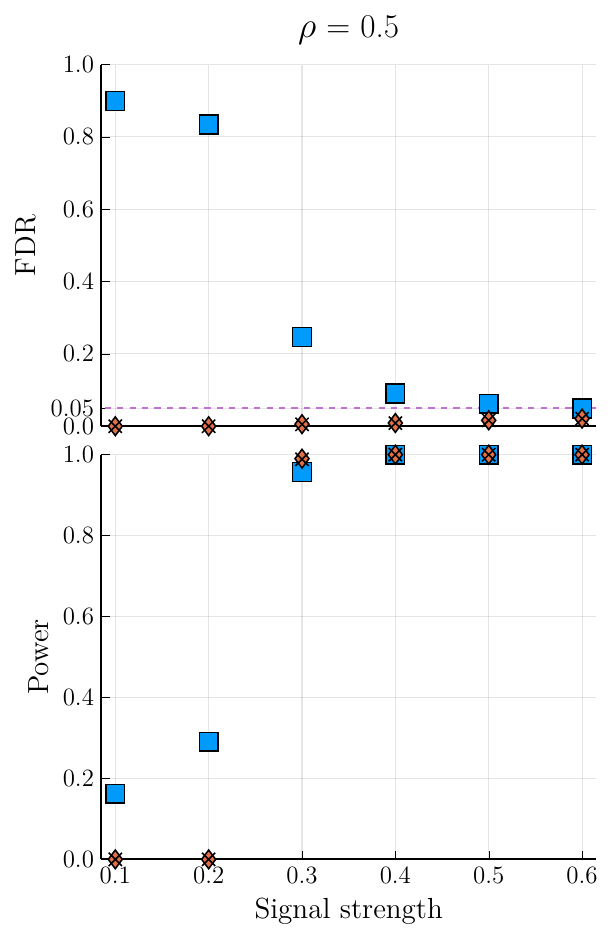}
    \end{subfigure}%
    \begin{subfigure}{0.33\textwidth}
        \centering
        \includegraphics[width=\textwidth]{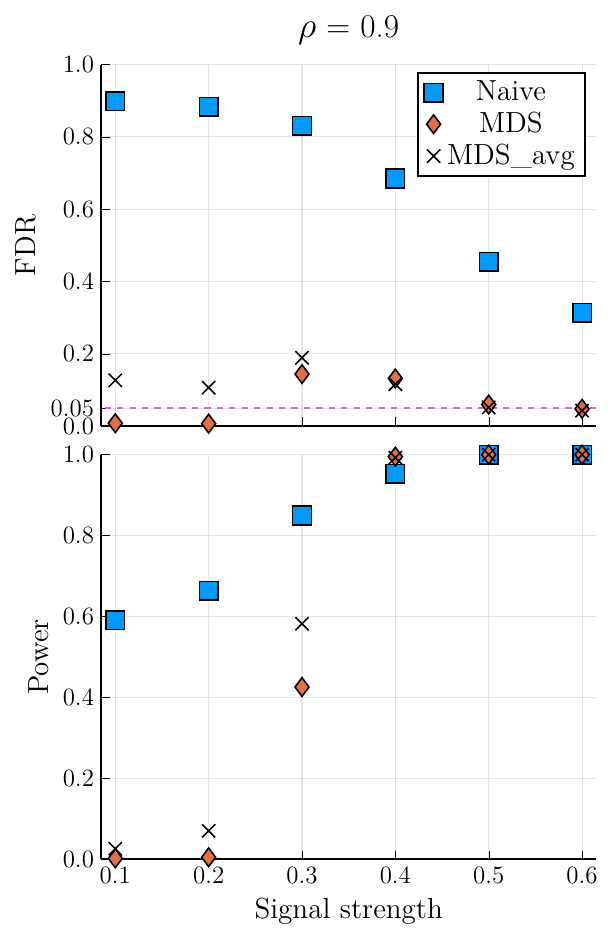}
    \end{subfigure}
    \caption{Average FDR and average power 
    % (with one standard deviation indicated by the error bar) 
    versus the signal strength among 100 experiments under the Gaussian setting with $n=1000$ samples, $p=2000$ features, $p_1=200$ relevant features and noise level $\sigma_\varepsilon = 0.1$. }
    \label{fig:normal-sigma0.1}
\end{figure}

% Here are the average FDR and power under other noise levels.
\begin{figure}[H]
    \centering
    \begin{subfigure}{0.33\textwidth}
        \centering
        %%% although M50 in the filename, but actually only M=10 is used
        \includegraphics[width=\textwidth]{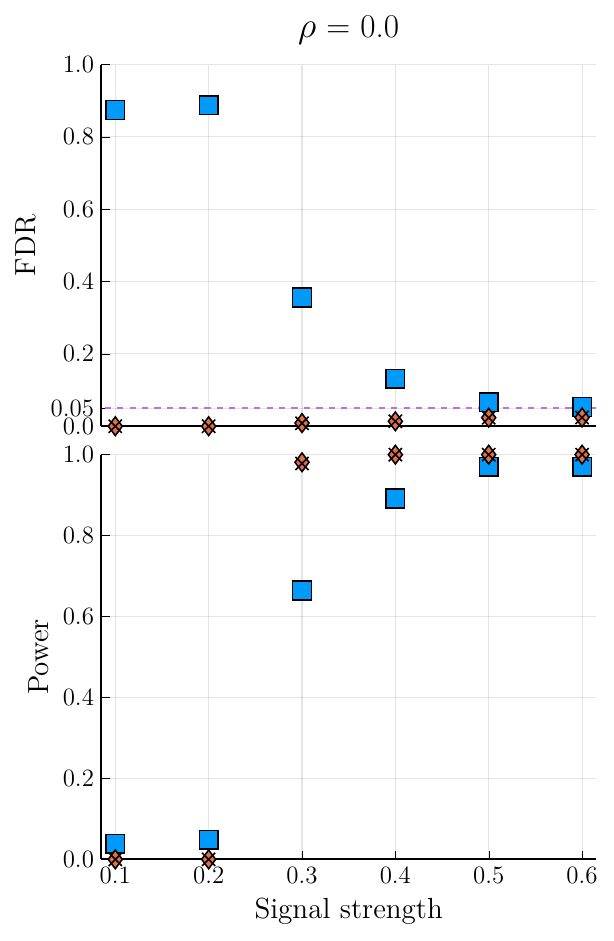}
    \end{subfigure}%
    \begin{subfigure}{0.33\textwidth}
        \centering
        \includegraphics[width=\textwidth]{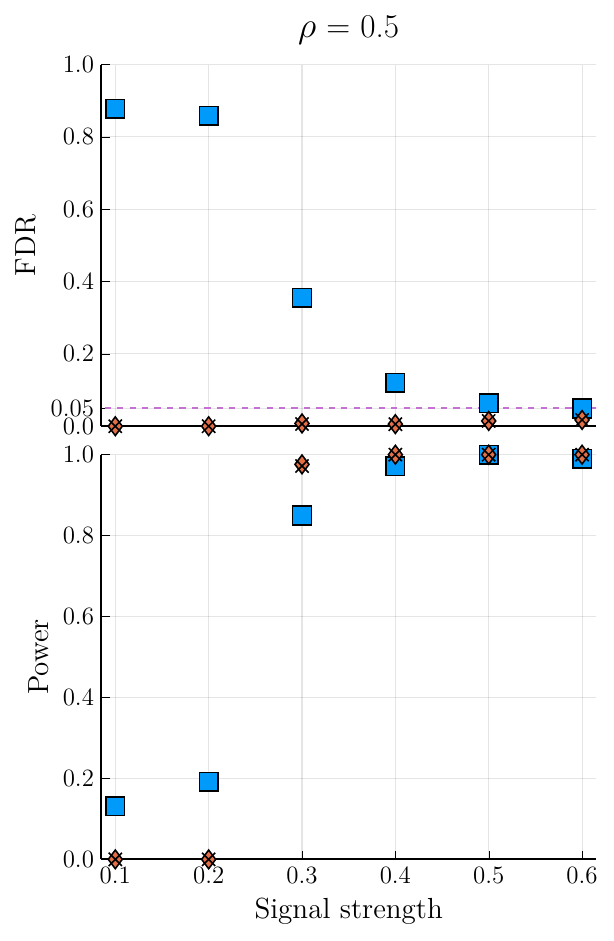}
    \end{subfigure}%
    \begin{subfigure}{0.33\textwidth}
        \centering
        \includegraphics[width=\textwidth]{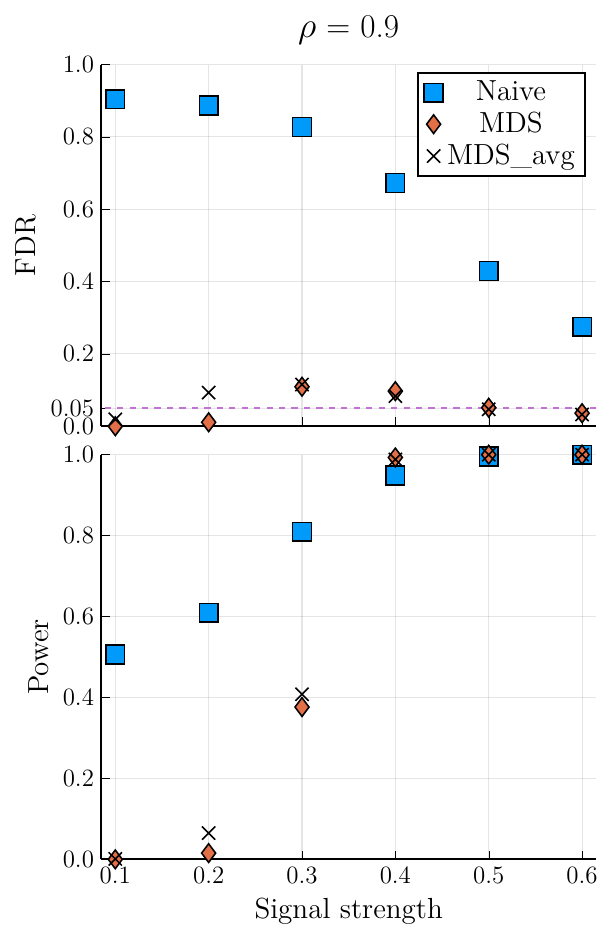}
    \end{subfigure}
    \caption{Average FDR and average power (with one standard deviation indicated by the error bar) versus the signal strength among 100 experiments under the Gaussian setting with $n=1000$ samples, $p=2000$ features, $p_1=200$ relevant features and noise level $\sigma_\varepsilon = 0.5$. }
    \label{fig:normal-sigma0.5}
\end{figure}

\subsection{Gaussian setting with higher noise level}

\begin{figure}[H]
    \centering
    \begin{subfigure}{0.33\textwidth}
        \centering
        %%% although M50 in the filename, but actually only M=10 is used
        \includegraphics[width=\textwidth]{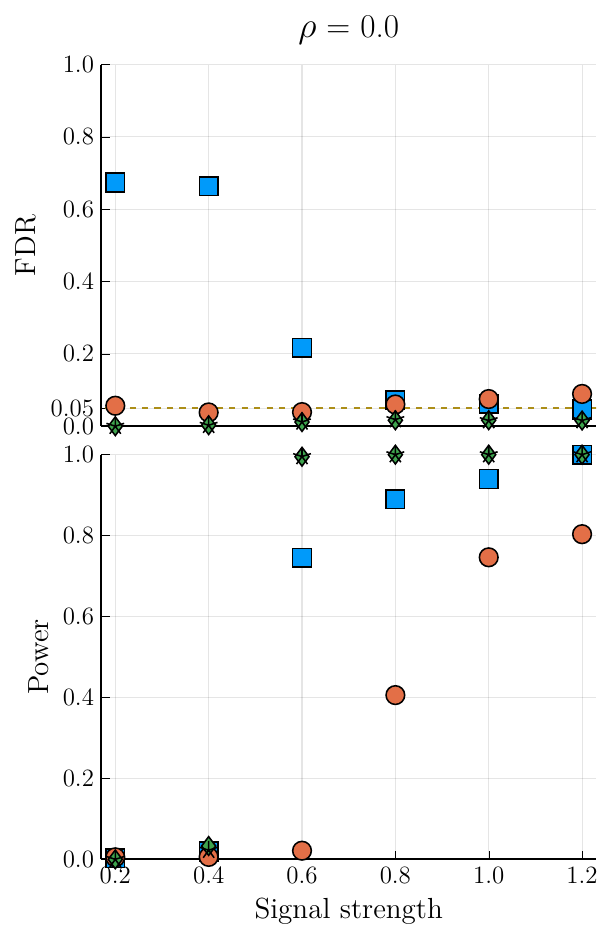}
    \end{subfigure}%
    \begin{subfigure}{0.33\textwidth}
        \centering
        \includegraphics[width=\textwidth]{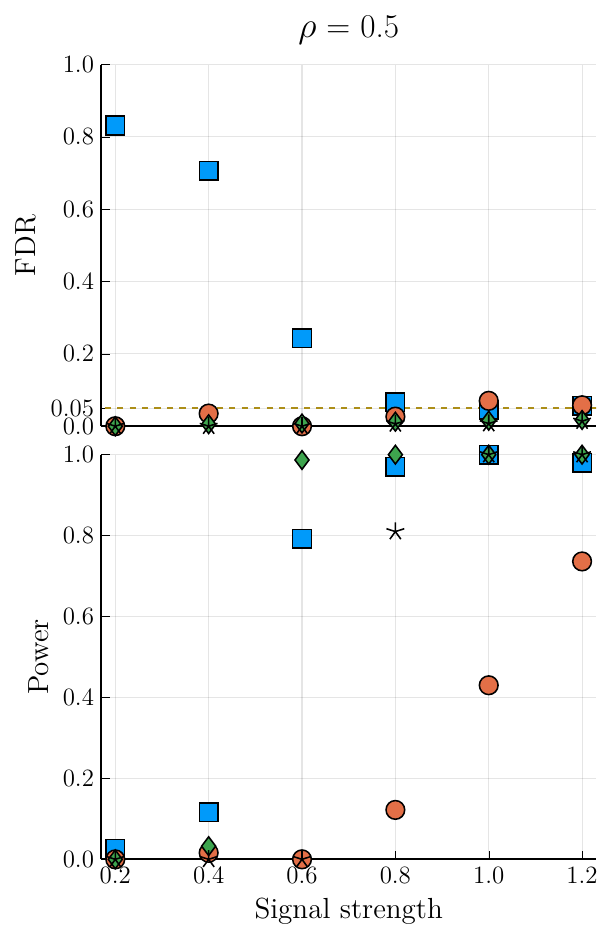}
    \end{subfigure}%
    \begin{subfigure}{0.33\textwidth}
        \centering
        \includegraphics[width=\textwidth]{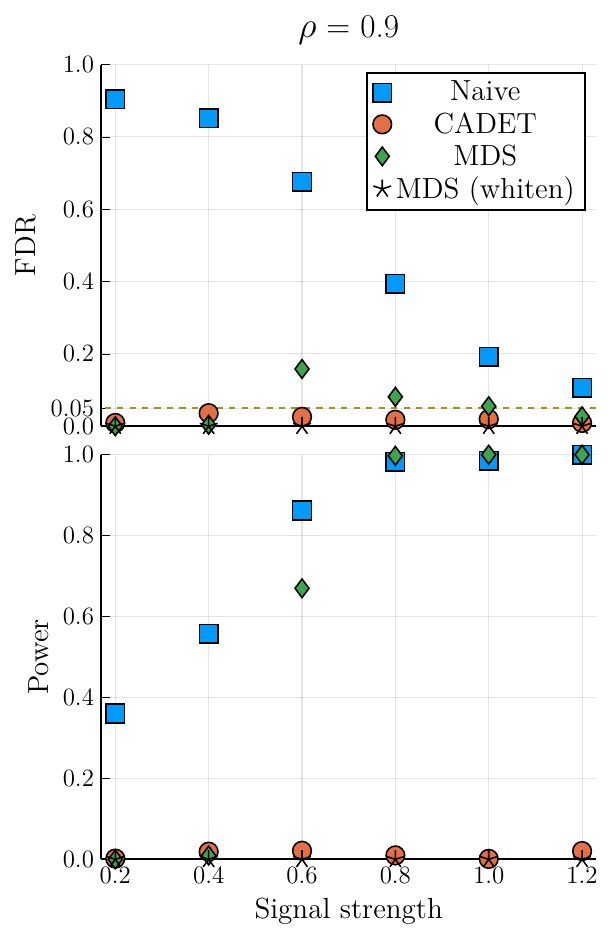}
    \end{subfigure}
    \caption{Average FDR and average power (with one standard deviation indicated by the error bar) versus the signal strength of among 100 experiments under the Gaussian setting with $n=500$ samples, $p=1000$ features, $p_1=100$ relevant features and noise level $\sigma_\varepsilon = 0.5$. }
    \label{fig:normal-cadet-sigma0.5}
\end{figure}

% \begin{figure}[H]
%     \centering
%     \begin{subfigure}{0.33\textwidth}
%         \centering
%         %%% although M50 in the filename, but actually only M=10 is used
%         \includegraphics[width=\textwidth]{figs/ds_vs_dd_vs_cadet-normal-corr_ar1-scale0-oracle_false-M10-nrep100-n500-p1000-prop0.1-nsig6-0.2-1.2-rho0-sigma0-mds2-fdr-and-power.pdf}
%     \end{subfigure}%
%     \begin{subfigure}{0.33\textwidth}
%         \centering
%         \includegraphics[width=\textwidth]{figs/ds_vs_dd_vs_cadet-normal-corr_ar1-scale0-oracle_false-M10-nrep100-n500-p1000-prop0.1-nsig6-0.2-1.2-rho0.5-sigma0-mds2-fdr-and-power.pdf}
%     \end{subfigure}%
%     \begin{subfigure}{0.33\textwidth}
%         \centering
%         \includegraphics[width=\textwidth]{figs/ds_vs_dd_vs_cadet-normal-corr_ar1-scale0-oracle_false-M10-nrep100-n500-p1000-prop0.1-nsig6-0.2-1.2-rho0.9-sigma0-mds2-fdr-and-power.pdf}
%     \end{subfigure}
%     \caption{Average FDR and average power (with one standard deviation indicated by the error bar) versus the signal strength of among 100 experiments under the Gaussian setting with $n=500$ samples, $p=1000$ features, $p_1=100$ relevant features and noise level $\sigma_\varepsilon = 0$. }
%     \label{fig:normal-cadet-sigma0}
% \end{figure}

\subsection{Poisson setting with higher noise level}

\begin{figure}[H]
    \centering
    \begin{subfigure}{0.33\textwidth}
        \centering
        \includegraphics[width=\textwidth]{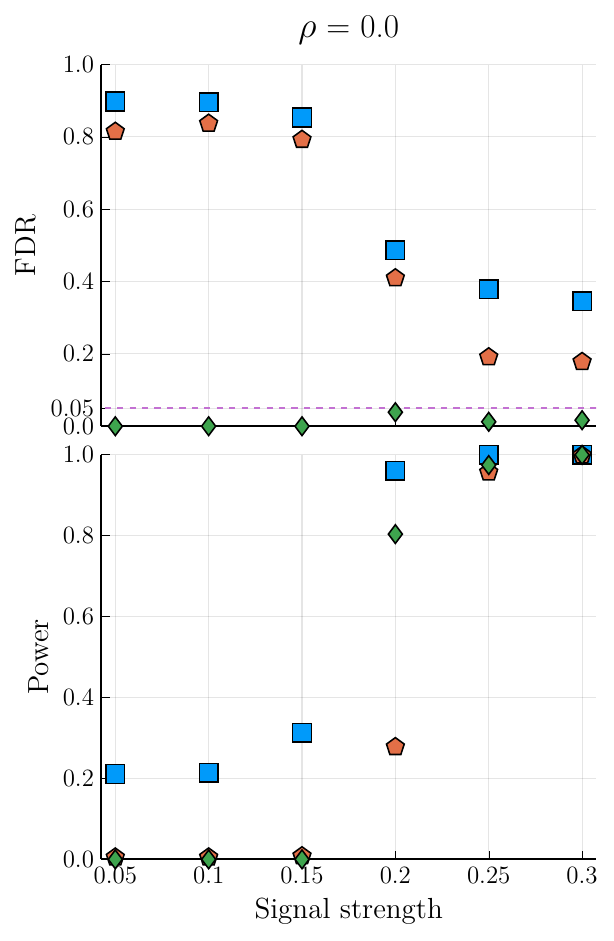}
    \end{subfigure}%
    \begin{subfigure}{0.33\textwidth}
        \centering
        \includegraphics[width=\textwidth]{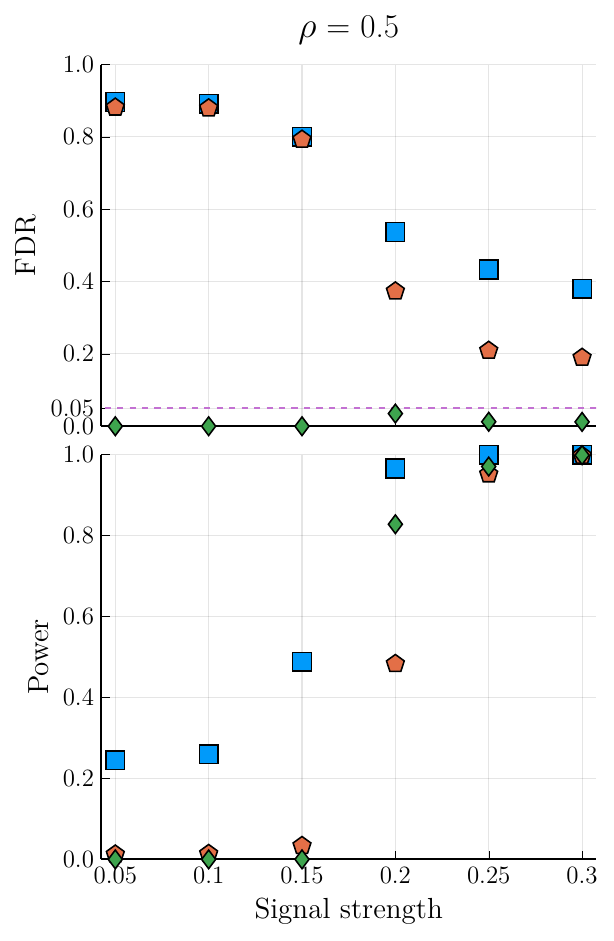}
    \end{subfigure}%
    \begin{subfigure}{0.33\textwidth}
        \centering
        \includegraphics[width=\textwidth]{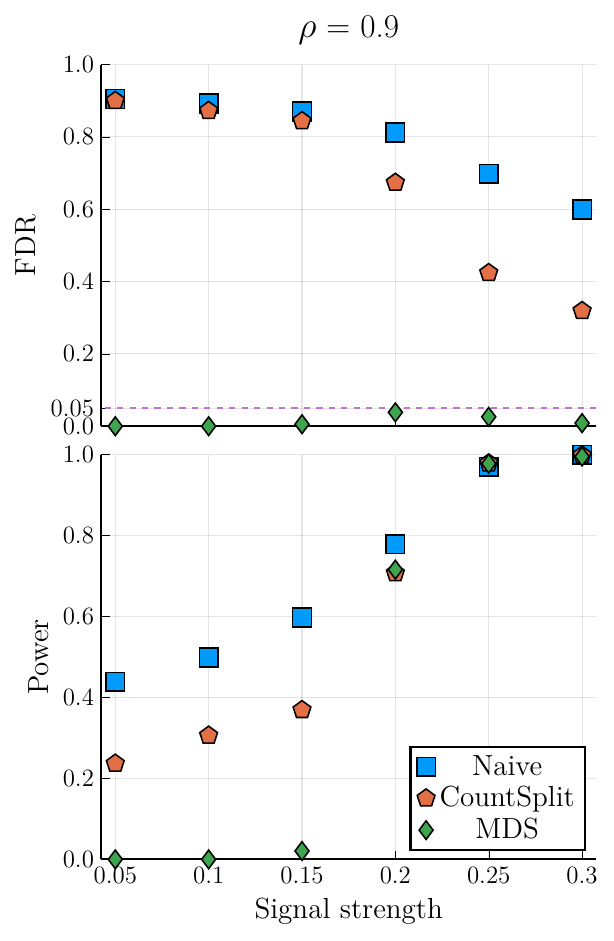}
    \end{subfigure}
    \caption{Average FDR and average power (with one standard deviation indicated by the error bar) versus the signal strength among 100 experiments under the Poisson setting with $n=1000$ samples, $p=2000$ features, $p_1=200$ relevant features and noise level $\sigma_\varepsilon = 0.5$. }
    \label{fig:pois-sigma0.5}
\end{figure}

\subsection{Trajectory setting with higher noise level}

\begin{figure}[H]
    \centering
    \begin{subfigure}{0.33\textwidth}
        \centering
        \includegraphics[width=\textwidth]{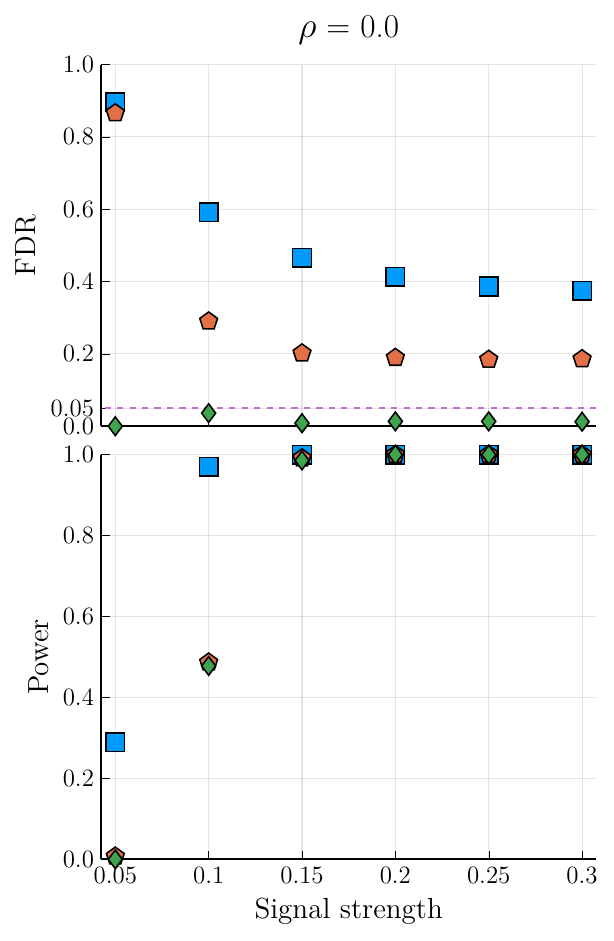}
    \end{subfigure}%
    \begin{subfigure}{0.33\textwidth}
        \centering
        \includegraphics[width=\textwidth]{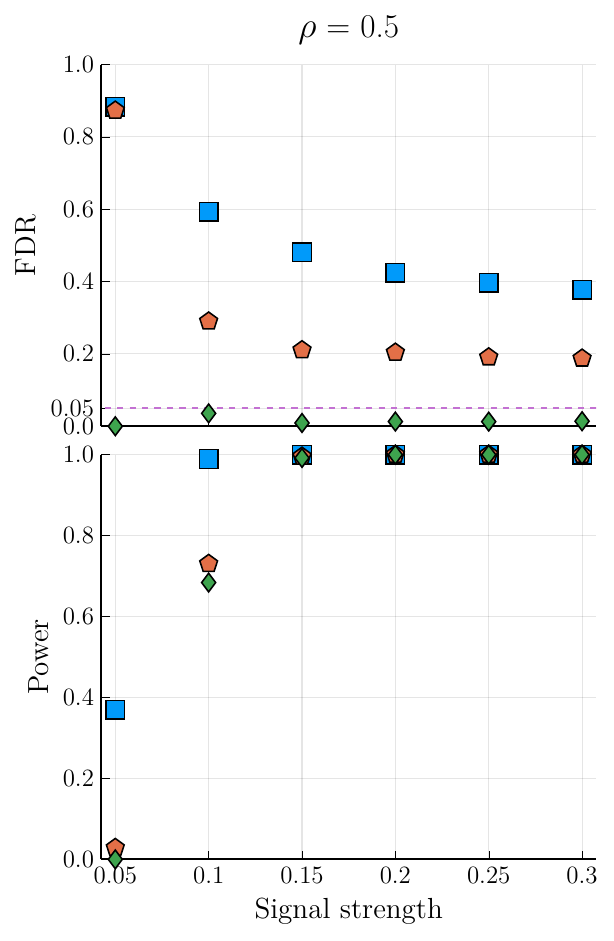}
    \end{subfigure}%
    \begin{subfigure}{0.33\textwidth}
        \centering
        \includegraphics[width=\textwidth]{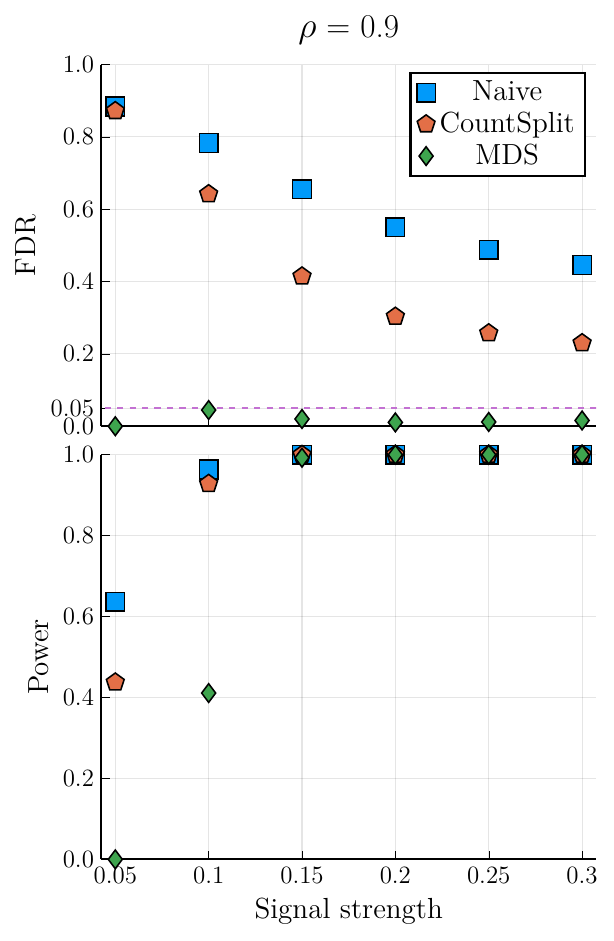}
    \end{subfigure}
    \caption{Average FDR and average power versus signal strength among 100 experiments under the linear trajectory setting with $n=1000$ samples, $p=2000$ features, $p_1=200$ relevant features and noise level $\sigma_\varepsilon=0.5$.}
    \label{fig:pois-traj-sigma0.5}
\end{figure}

\subsection{Synthetic scRNA-seq data for different numbers of DE genes}

Besides logFC, the number of DE genes can also reflect the signal strength. When the number of DE genes increases, the signal becomes stronger, and it is easy to separate them into two clusters. Figure~\ref{fig:synthetic_vary_nde} shows the actual FDR and power versus the target FDR when the number of DE genes is 200, 400, and 800 with logFC fixed to be 0.3. 
To illustrate results from different hypothesis tests,
here we show results from the Wilcox test, different from the t-test used in Figure 5. %~\ref{fig:synthetic_vary_logfc}. 
The simulation setting of the middle column of Figure 5 %~\ref{fig:synthetic_vary_logfc} 
as the setting of the first column of Figure~\ref{fig:synthetic_vary_nde} except that the used testing method (the former is t-test while the latter is Wilcox test). We find that different tests show quite similar performance for each method. Similar to the results of t-test, the naive double-dipping method again failed to control FDR when the signal is weak (nDE = 200). And the proposed MDS method can maintain a comparable power while controlling FDR. When the number of DE genes increases, all methods can control the power, and MDS can achieve higher power for a uniform range of target FDR.

% before SCT & SCT transform
%% different proportions of cell types

\begin{figure}[H]
    \centering
    \includegraphics[width=\textwidth]{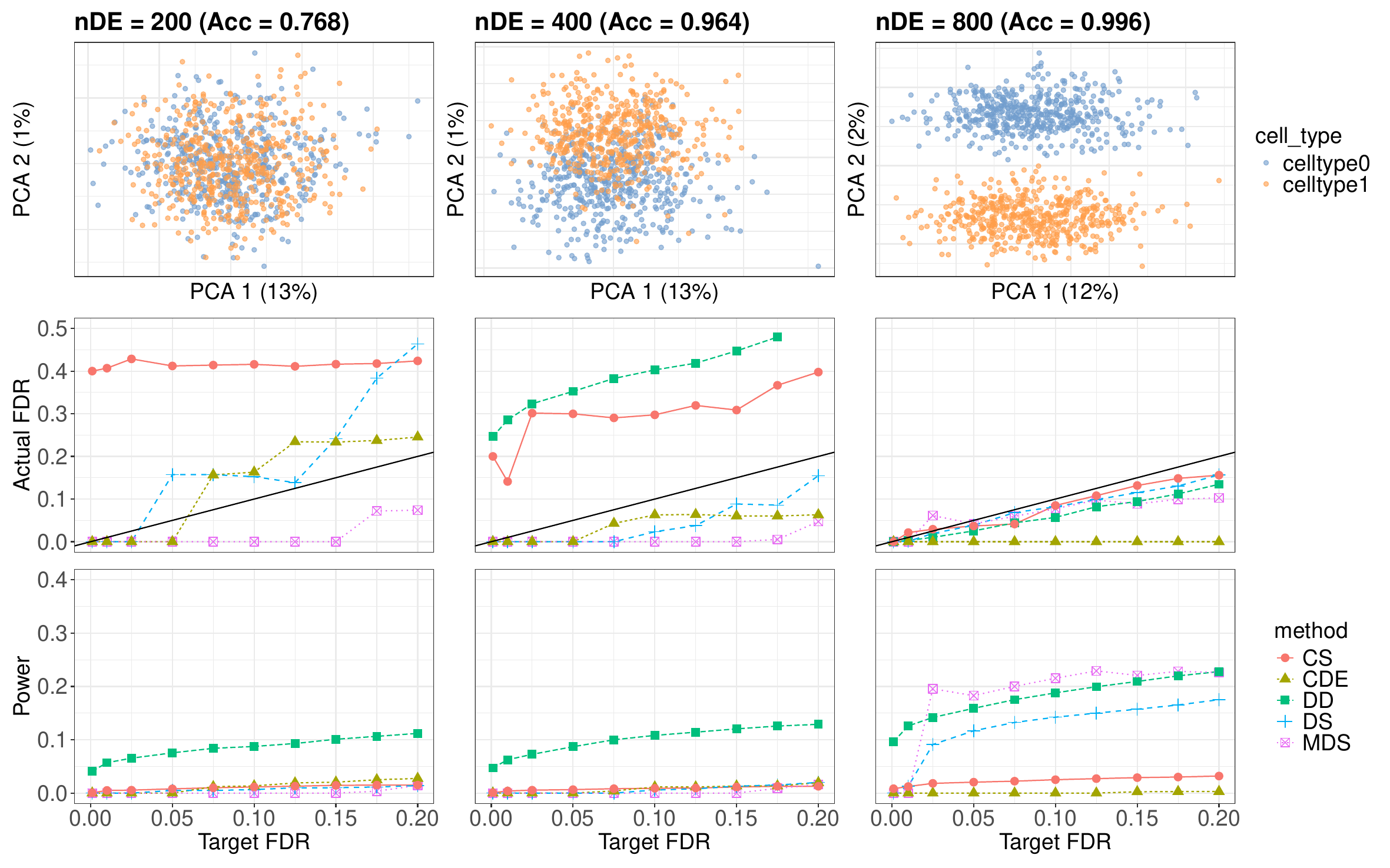}
    \caption{The actual FDR and power of five different approaches (CS: CountSplit; CDE: ClusterDE; DD: Double-dipping; DS: Data-splitting; MDS: Multiple DS) versus the target FDR given different \emph{numbers of DE genes} (nDE) when $\logFC = 0.3$ based on the Wilcox test.}
    \label{fig:synthetic_vary_nde}
\end{figure}

\subsection{Synthetic scRNA-seq data for different cell-type ratios}

In Figures 5 %~\ref{fig:synthetic_vary_logfc} 
and \ref{fig:synthetic_vary_nde}, the cell type ratio is 1, which means that the number of samples from cell type 1 is the same as the number of samplers from cell type 2. Now we consider the effect of different cell type ratios.
Specifically, if the ratio is $k$, then the proportion of the cell type 1 is $\frac{k}{k+1}$, while the proportion of cell type 2 is $\frac{1}{k+1}$. Figure~\ref{fig:synthetic_vary_ct} shows the FDR and power when the cell type ratio ranges from 1 to 4 given 800 DE genes. Note that the simulation setting of the right column of Figure~\ref{fig:synthetic_vary_nde} is the same as the left column of Figure~\ref{fig:synthetic_vary_ct} except that the used testing method: the former adopts the Wilcox test while the latter takes the Poisson test. We observe that in the unbalanced settings, the proposed DS and MDS can also outperform others while controlling FDR. In the most unbalanced case (cell type ratio = 4), there are slight inflations of FDR when the target FDR is small for the MDS and DS, the possible reason is that highly unbalanced data is more likely to produce extremely unbalanced splits.

\begin{figure}[H]
    \centering
    \includegraphics[width=\textwidth]{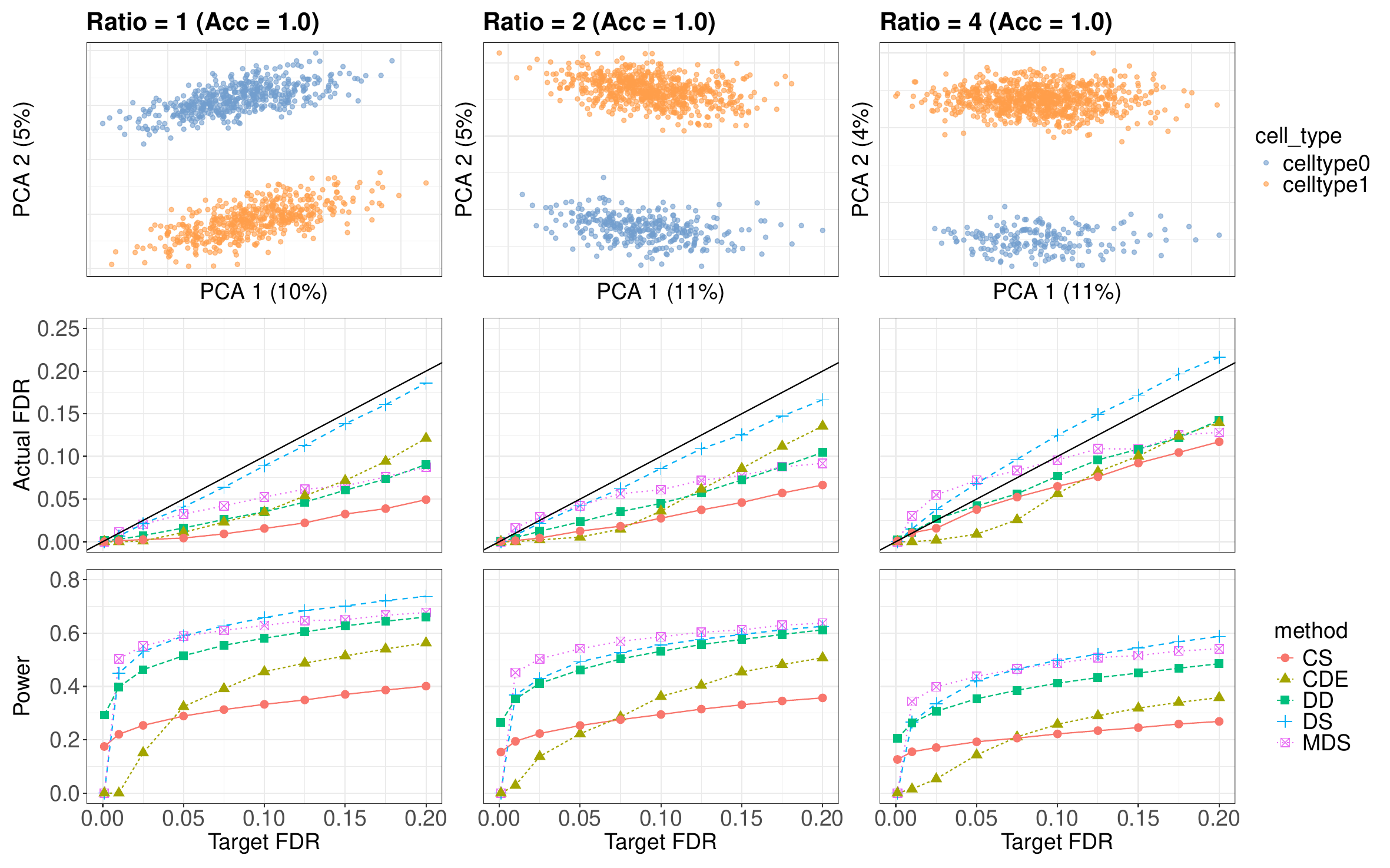}
    \caption{The actual FDR and power of five different approaches (CS: CountSplit; CDE: ClusterDE; DD: Double-dipping; DS: Data-splitting; MDS: Multiple DS) versus the target FDR for different cell type ratios based on the Poisson test when nDE = 800 and $\logFC = 0.5$.}
    \label{fig:synthetic_vary_ct}
\end{figure}

\subsection{Investigation of homogeneous cell type}

\begin{figure}[H]
    \centering
    \includegraphics[width=0.7\textwidth]{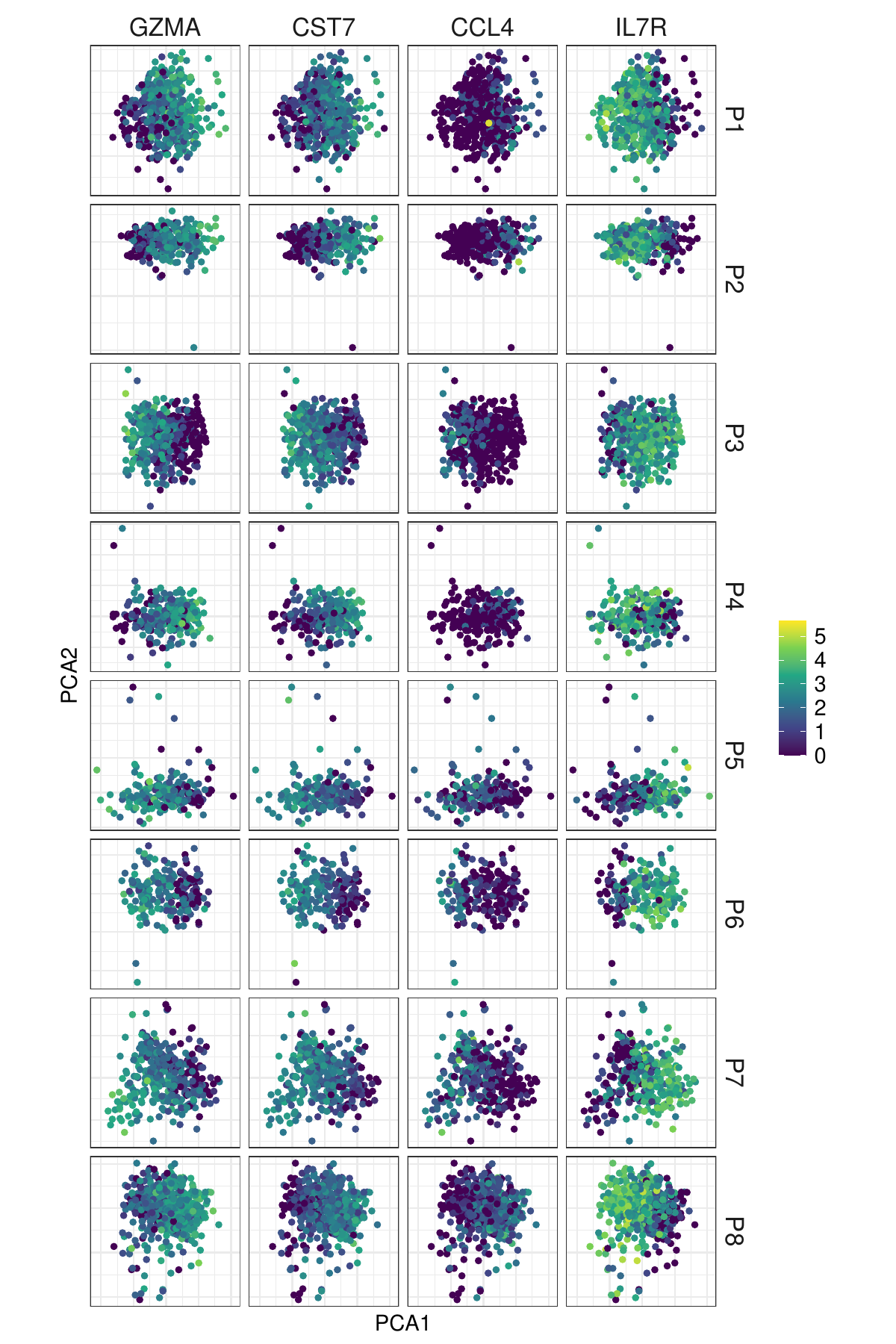}
    \caption{PCA plots of eight donors, colored by key CD8 T cell state markers: GZMA, CST7, CCL4 and IL7R.} 
    \label{suppfig:cd8tem1_de}
\end{figure}

\subsection{DE analysis across heterogeneous cell population - B cell subpopulations}

\begin{figure}[H]
    \centering
    \includegraphics[width=\textwidth]{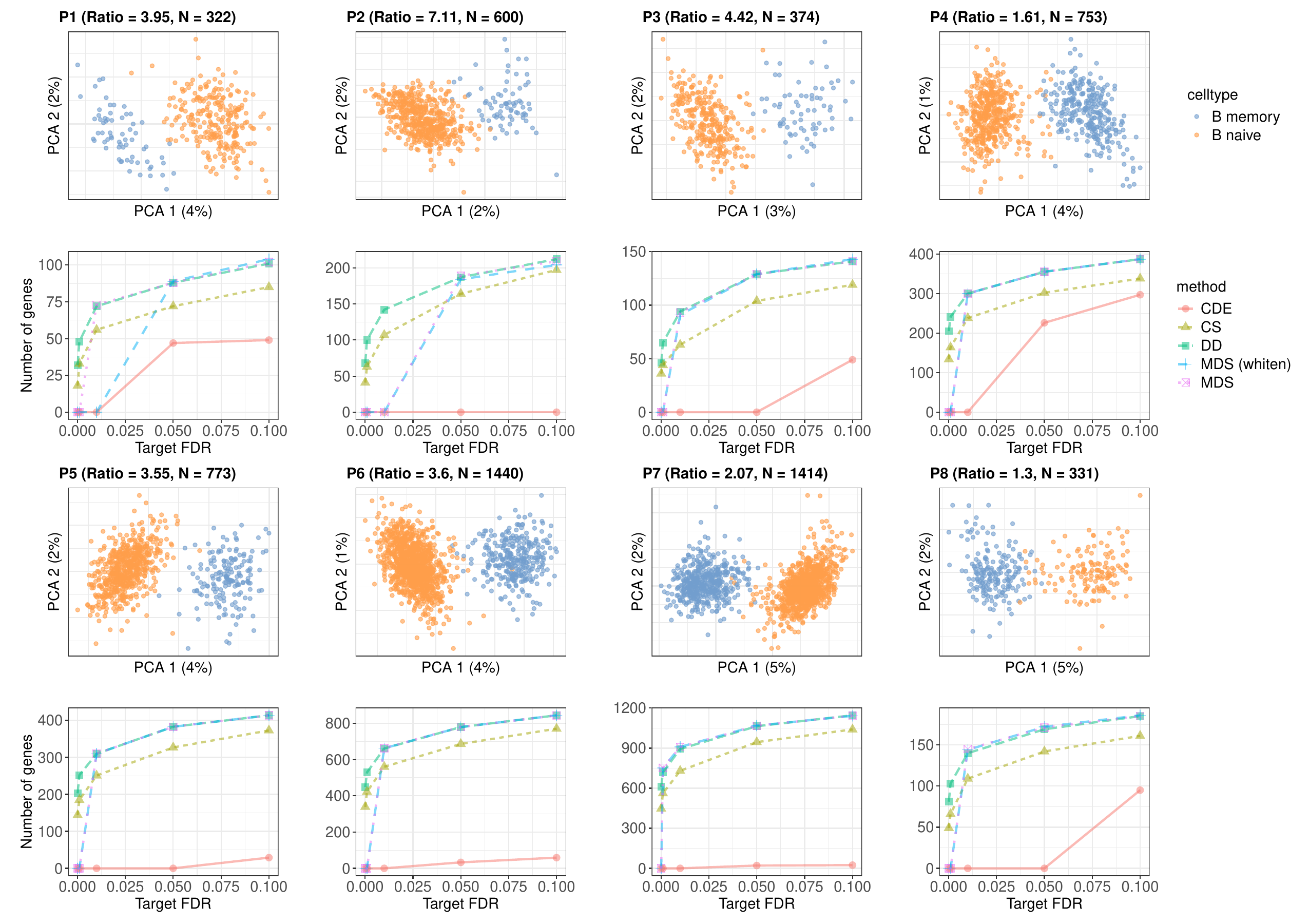}
    \caption{The number of DE genes overlapped with ground truth using five methods across eight donors, with varying target FDRs, comparing two B cell subtypes: B memory vs. B naive. The cell type ratio and the total number of cells for each donor are indicated in the subtitle of each PCA plot.} 
    \label{fig:twocelltype_B}
\end{figure}

\if1\isarxiv
\section{Proof of Proposition \ref{prop:label_switch}}
\else 
\section{Proof of Proposition 1} %% need for separate supp file
\fi

\begin{proof}
    \begin{align}
        \sum_{j=1}^p d_j^{(1)} d_j^{(2)}&=\sum_{j=1}^p (\delta_j + \varepsilon_j)(-\delta_j + e_j)\notag\\
        &=\sum_{j=1}^p (-\delta_j^2 + \delta_j(e_j-\varepsilon_j) + \varepsilon_j e_j)\notag\\
        &=-\sum_{j\in S_1} \delta_j^2 + \sum_{j\in S_1}\delta_j (e_j-\varepsilon_j) + \sum_{j=1}^p \varepsilon_j e_j\label{eq:dj1_dj2}
    \end{align}
    Note that $e_j - \varepsilon_j\sim N(0, 2\sigma^2)$, then 
    $$
    \sum_{j\in S_1}\delta_j(e_j-\varepsilon_j) \sim N(0, 2\sigma^2\sum_{j\in S_1}\delta_j^2)\,.
    $$
    By Chernoff bound, we have
    $$
    \Pr\left(\vert\sum_{j\in S_1}\delta_j (e_j - \varepsilon_j)\vert > t\right)\le 2\exp\left(-\frac{t^2}{4\sigma^2\sum_{j\in S_1}\delta_j^2}
    \right)\,.
    $$
    % By the bound on the Gaussian chaos variable \parencite{wainwrightHighdimensionalStatisticsNonasymptotic2019}, we have
    % $$
    % \Pr\left(\vert\sum_{j=1}^p \varepsilon_j e_j\vert \ge t \right)\le 2\exp\left(-\frac{t^2}{4p + 4t}\right) + 2\exp\left(-\frac{t}{2}\right)\,.
    % $$   
    Taking $t = \frac{1}{2}\sum_{j\in S_1}\delta_j^2$ yields
    \begin{equation}\label{eq:diff_noise}
    \Pr\left(\vert\sum_{j\in S_1}\delta_j (e_j - \varepsilon_j)\vert > \sum_{j\in S_1}\delta_j^2/2\right)\le 2\exp\left(-\frac{\sum_{j\in S_1}\delta_j^2}{4\sigma^2}
    \right)\,.
    \end{equation}
    Note that both $\varepsilon_j/\sigma, e_j/\sigma$ are standard Gaussian random variables, each with the sub-Gaussian norm 
    $$
    \Vert \varepsilon_j/\sigma\Vert_{\Psi_2} = \Vert e_j/\sigma\Vert_{\Psi_2} = 1\,,
    $$
    so
    $$
    \Vert \varepsilon_j e_j/\sigma^2\Vert_{\Psi_1} \le \Vert \varepsilon_j/\sigma\Vert_{\Psi_1}\Vert e_j/\sigma\Vert_{\Psi_1} = 1\,.
    $$
    Thus, we have
    $$
    \Pr\left(\vert\sum_{j=1}^p \varepsilon_j e_j\vert \ge pt\sigma^2 \right)\le 2\exp\left(-\frac p2 \cdot\min\{t^2, t\}\right)\,.
    $$
    Take $t = \frac{\sum_{j\in S_1}\delta_j^2}{2p\sigma^2}$. If $\sum_{j\in S_1}\delta_j^2 > 2p\sigma^2$, then
    \begin{equation}\label{eq:e2_large}
    \Pr\left(\vert\sum_{j=1}^p \varepsilon_j e_j\vert \ge \frac{\sum_{j\in S_1}\delta_j^2}{2} \right)\le 2\exp\left(-\frac{\sum_{j\in S_1}\delta_j^2 }{4\sigma^2}\right)\,.        
    \end{equation}
    If $\sum_{j\in S_1}\delta_j^2 > c_1\sigma^2p^{1/2+\varepsilon},\varepsilon > 0$, we have
    \begin{equation}\label{eq:e2_small}
    \Pr\left(\vert\sum_{j=1}^p \varepsilon_j e_j\vert \ge \frac{\sum_{j\in S_1}\delta_j^2}{2} \right)\le 2\exp\left(-\frac{(\sum_{j\in S_1}\delta_j^2)^2 }{8p\sigma^4}\right)\le 2\exp\left(-\frac{c_1}{8}p^\epsilon\right)\,.        
    \end{equation}
    % In both cases, with a high probability, we have
    % \begin{equation}\label{eq:sum_e_eps}
    % \vert\sum_{j=1}^p \varepsilon_j e_j\vert < \frac{\sum_{j\in S_1}\delta_j^2}{2}\,.        
    % \end{equation}
    Thus, if $\sum_{j\in S_1}\delta_j^2 > 2p\sigma^2$, combing \eqref{eq:dj1_dj2}, \eqref{eq:diff_noise} and \eqref{eq:e2_large} yields
    \begin{equation}\label{eq:bound1}
    \Pr\left(\sum_{j=1}^p d_j^{(1)}d_j^{(2)} > 0
    \right)\le 
    2\exp\left(-\frac{\sum_{j\in S_1}\delta_j^2 }{4\sigma^2}\right)\,.  
    \end{equation}
    If $c_1\sigma^2p^{1/2+\varepsilon}< \sum_{j\in S_1}\delta_j^2 < 2p\sigma^2$, combining \eqref{eq:dj1_dj2}, \eqref{eq:diff_noise} and \eqref{eq:e2_small} yields
    \begin{align}
    \Pr\left(\sum_{j=1}^p d_j^{(1)}d_j^{(2)} > 0
    \right)&\le \min\left\{
    2\exp\left(-\frac{\sum_{j\in S_1}\delta_j^2}{4\sigma^2}\right),
    2\exp\left(-\frac{(\sum_{j\in S_1}\delta_j^2)^2 }{8p\sigma^4}\right)
    \right\}\notag
    \\
    &=2\exp\left(-\frac{(\sum_{j\in S_1}\delta_j^2)^2 }{8p\sigma^4}\right)\,.\label{eq:bound2}
    \end{align}
    Combing \eqref{eq:bound1} and \eqref{eq:bound2} yields
    $$
    \Pr\left(\sum_{j=1}^p d_j^{(1)}d_j^{(2)} > 0
    \right) \le 2\exp\left(
    -\min\left\{
    \frac{\sum_{j\in S_1}\delta_j^2}{4\sigma^2},
    \frac{(\sum_{j\in S_1}\delta_j^2)^2 }{8p\sigma^4}
    \right\}
    \right)\,.
    $$
    Thus, with a high probability of at least
    $$
    1 - 2\exp\left(
    -\min\left\{
    \frac{\sum_{j\in S_1}\delta_j^2}{4\sigma^2},
    \frac{(\sum_{j\in S_1}\delta_j^2)^2 }{8p\sigma^4}
    \right\}
    \right)\,,
    $$
    we have $\sum_{j=1}^p d_j^{(1)}d_j^{(2)} < 0$.
\end{proof}

\if1\isarxiv
\section{Proof of Proposition~\ref{prop:unbalanced_split}}
\else
\section{Proof of Proposition 3}
\fi

\begin{proof}
Let $n$ be the sample size, and $n_i, i=1,2$ be the sample size for each class. Now randomly split the data into two equal parts. Without loss of generality, assume $n$ is even. Let $X$ be the number of class-1 samples in the first part, then
$$
\Pr(X=k) = \frac{\binom{n_1}{k}\binom{n_2}{n/2-k}}{\binom{n}{n/2}},\quad k \le \min(n/2, n_1)\,.
$$
It follows that the number of the minority class of the first part is
$$
Y = \min(X, n/2-X)\,.
$$
Let $Z = n/2-X$, then
$$
\Pr(Z=k) = \Pr(X=n/2-k) = \frac{\binom{n_1}{n/2-k}\binom{n_2}{k}}{\binom{n}{n/2}}\,.
$$
It follows that the CDF of $Y$ is
\begin{align*}
F(y) &= \Pr(Y\le y) = 1 - \Pr(Y > y)\\
&=1-\Pr(X > y, Z > y) = 1 - \Pr\left(y < X < n/2 - y\right) \\
% &= 1-\sum_{k=y+1}^{n/2-y-1} \frac{\binom{n_1}{n/2-k}\binom{n_2}{k}}{\binom{n}{n/2}}\\
&=\Pr(X\le y) + \Pr(X\ge n/2-y)\,.
\end{align*}
Thus, if $y > n/4 -1$, i.e., $y\ge n/4$, we have $F(y) = 1$.

Note that $X\sim \text{Hypergeometric}(n, n_1, n/2)$, let $\alpha = n_1/n$. Then by the Hoeffding's inequality \parencite{hoeffdingProbabilityInequalitiesSums1963}, for $0 < t < \alpha$, 
\begin{align*}
    \Pr[X\le (\alpha-t)n/2 ] &\le \exp(-t^2n)\,,\\
    \Pr[X\ge (\alpha+t)n/2 ] &\le \exp(-t^2n)\,.
\end{align*}
Then we have
$$
\Pr(X\le y) \le \exp\left[-\left(\alpha-\frac{2y}{n}\right)^2n\right]\,,
$$
and
$$
\Pr(X\ge n/2-y) \le \exp\left[-\left(1-\alpha-\frac{2y}{n}\right)^2n\right]\,.
$$
Thus,
$$
F(y) \le \exp\left[-\left(\alpha-\frac{2y}{n}\right)^2n\right] + \exp\left[-\left(1-\alpha-\frac{2y}{n}\right)^2n\right]\,.
$$
Let $W\triangleq 2Y/n$ be the proportion of the minority class, then
$$
F(w)\le \exp(-(\alpha-w)^2n) + \exp(-(1-\alpha-w)^2n)\,.
$$
Particularly, if $\alpha = \frac 12$, i.e., equal size of two classes, we have
$$
F(w)\le 2\exp\left[-\left(\frac 12-w\right)^2n\right]\,.
$$
\end{proof}

\if1\isarxiv
\section{Proof of Proposition \ref{prop:ip_cluster}}
\else
\section{Proof of Proposition 4}
\fi

\begin{lemma}
    If $X\sim N(\mu, I)$, then 
    \begin{align*}
        \bbE[X1(a^\top X > b)] &= \mu\left(1-\Phi(\frac{b-a^\top \mu}{\sqrt{a^\top a}})\right) + \frac{a}{\sqrt{a^\top a}}\phi\left(\frac{b-a^\top \mu}{\sqrt{a^\top a}}\right)\,.
    \end{align*}
\end{lemma}
\begin{proof}
    If $X\sim N(\mu, I)$, then $Z = a^\top X\sim N(a^\top\mu, a^\top a)$. The joint distribution of $(X, a^\top X)$ is
    $$
    \begin{bmatrix}
        X\\
        a^\top X
    \end{bmatrix}
    \sim N\left(
    0, \begin{bmatrix}
        \bI & a\\
        a^\top & a^\top a
    \end{bmatrix}
    \right)\,,
    $$
    then $X$ given $Z=z$ is normally distributed with mean
    $$
    \bbE[X\mid a^\top X=z] = \mu +\frac{a}{a^\top a}(z-a^\top \mu) = \frac{a}{a^\top a}z+\bA\mu\,,
    $$
    and covariance matrix
    $$
    \Cov[X\mid a^\top X =z] = \bfI - \frac{aa^\top}{a^\top a}\triangleq \bA\,.
    $$
Note that
\begin{align*}
\bbE [X1(a^\top X > b)] & =\bbE[\bbE[X1(Z>b)\mid Z]]\\
&=\bbE[1(Z > b)\bbE[X\mid Z]]\\
&=\bbE\left[1(Z > b)\left(\frac{a}{a^\top a}Z + \left(\bfI-\frac{a a^\top}{a^\top a}\right)\mu\right) \right]\\
&=\frac{a}{a^\top a}\bbE[Z1(Z > b)] + \left(\bfI-\frac{ a a^\top}{a^\top a}\right)\mu\bbE[1(Z > b)]\,.
\end{align*}
Let $U = \dfrac{Z-a^\top\mu}{\sqrt{a^\top a}}$ and $u = \dfrac{b-a^\top\mu}{\sqrt{a^\top a}}$, then
\begin{align*}
    \bbE[1(Z > b)] &= \bbE[1(U > u)]= 1-\Phi(u)\,,\\
    \bbE [Z1 (Z > b)] &= \sqrt{a^\top a}\bbE\left[U 1\left( U > u\right) \right]+a^\top\mu\bbE\left[1\left( U > u\right)\right] \\
    &\triangleq\sqrt{a^\top  a} \Psi\left(u\right) +  a^\top\mu\left(1-\Phi\left( u\right)\right)\,,
\end{align*}
where $\Psi(x) = \int_x^\infty t\phi(t)dt$. Note that
\begin{align*}
\Psi(x) &= \int_x^\infty t\phi(t)dt= \frac{1}{\sqrt{2\pi}}\int_x^\infty t\exp(-t^2/2)dt \\
&= \frac{1}{\sqrt{2\pi}}\int\exp(-t^2/2)dt^2/2 = \frac{1}{\sqrt{2\pi}}(-\exp(-t^2/2))\mid_{x}^\infty \\
&= \frac{1}{\sqrt{2\pi}}\exp(-x^2/2) = \phi(x)\,.    
\end{align*}
Thus,
$$
\bbE[X1(a^\top X > b)] = \mu(1-\Phi(u)) + \frac{a}{\sqrt{a^\top a}}\phi(u)\,,
$$
similarly,
$$
\bbE[X1(a^\top X < b)] = \mu - \frac{a}{\sqrt{a^\top a}}\phi(u)\,.
$$

\end{proof}

\begin{proof}
% First consider $X\sim 0.5N(\mu_1,\bfI) + 0.5N(\mu_2, \bfI)\triangleq 0.5 X_1 + 0.5 X_2$ and let $S = \{x:  a^\top x> b\}$. Note that
% \begin{align*}
%     \bbE X_C & = \bbE [X1(a^\top X > b)]\\
%     &= 0.5 \bbE[X_11(a^\top X_1 > b)] + 0.5 \bbE[X_21(a^\top X_2 > b)]\,,\\
%     \bbE X_C^\top X_C &= 0.5\bbE[X_1^\top X_11(a^\top X_1 > b)] + 0.5\bbE[X_2^\top X_21(a^\top X_2 > b)]\,.
% \end{align*}

Note that
\begin{align*}
    &\bbE \Vert X_C - Y_C\Vert^2 +\bbE \Vert X_{-C} - Y_{-C}\Vert^2\\
    =&\bbE\left[\Vert X_C\Vert^2 + \Vert Y_C\Vert^2 - 2X_C^\top Y_C + \Vert X_{-C}\Vert^2 + \Vert Y_{-C}\Vert^2 - 2X_{-C}^\top Y_{-C}\right]\\
    =&\bbE\left[\Vert X\Vert^2 + \Vert Y\Vert^2 - 2X_C^\top Y_C - 2X_{-C}^\top Y_{-C} \right]\,,
\end{align*}
and since $X_C$ and $Y_C$ are independent, then the target function becomes
$$
\argmin_C \bbE[X_C^\top Y_C + X_{-C}^\top Y_{-C}] =\argmin \Vert \bbE X_C\Vert^2 + \Vert \bbE X_{-C}\Vert^2 \,.
$$
Note that
$$
\bbE X_{-C} = \bbE X1(X\not\in C) = \bbE X(1 - 1(X\in C)) = -\bbE X1(X\in C) = -\bbE X_{-C},,
$$
then the target function simplifies to
$$
\argmin_C \Vert\bbE X_C\Vert^2\,.
$$
% \subsection{$\Sigma = \bfI$}
\subsection{(i)}
If $C$ is determined by hyperplanes $\sum_{j=1}^p a_jX_j > 0$, then when $\Sigma = \bfI_p$,
$$
X_{C_1}\overset{d}{=} X_{C_2}\,,
$$
and hence the optimal hyperplane is not unique.

% \subsection{$\Sigma \neq \bfI$}
\subsection{(ii)}
On the other hand, when $\Sigma\neq \bfI_p$. 
% Split along the direction perpendicular the first eigenvector can achieve the largest gap. 
Write $Z = a^\top X$, then $Z \sim N(0, a^\top \Sigma a)$. Rewrite
$$
\bbE X_C = \bbE [X\mid C] = \bbE [X\mid Y > 0]\,.
$$
Note that the joint distribution of $(X, Z)$ is 
$$
\begin{bmatrix}
    X\\
    a^\top X
\end{bmatrix}
\sim 
N\left(
0,
\begin{bmatrix}
    \Sigma & \Sigma a\\
    a^\top \Sigma & a^\top \Sigma a
\end{bmatrix}
\right)\,,
$$
then $X$ given $a^TX=z$ is normally distributed with mean 
$$
\bbE [X\mid a^\top X = z] = \frac{\Sigma a}{a^\top\Sigma a}z\,.
$$
It follows that the conditional expectation given $Z> 0$ is
$$
\bbE[X\mid a^\top X>0] = \bbE[\bbE[X\mid a^\top X]\mid a^\top X > 0] = \frac{\Sigma a}{a^\top\Sigma a}\bbE [Z\mid Z > 0]\,.
$$
Note that 
$$
\bbE[Z\mid Z > 0] = \sqrt{a^\top\Sigma a}\bbE\left[
\frac{Z}{\sqrt{a^\top \Sigma a}}\mid \frac{Z}{\sqrt{a^\top \Sigma a}} > 0
\right] =\frac{\sqrt{a^\top\Sigma a}}{\sqrt{2\pi}}\,.
$$
Therefore,
$$
\bbE[X\mid a^\top X > 0] = \frac{\Sigma a}{\sqrt{2\pi a^\top\Sigma a}}\,.
$$
Then the target function can be written as
$$
\argmax_{a, \Vert a\Vert = 1}\frac{\Vert \Sigma a\Vert^2}{2\pi a^\top\Sigma a} = \argmax_{a, \Vert a\Vert = 1} \frac{a^\top\Sigma^2a}{a^\top\Sigma a}\,,
$$
which is a generalized Rayleigh quotient. The maximum is attained when $a$ is proportional to the first eigenvector.

\end{proof}

\if1\isarxiv
\section{Proof of Proposition \ref{prop:2normal}}
\else
\section{Proof of Proposition 5}
\fi

\begin{proof}
    Note that
\begin{align*}
    &\bbE \Vert X_C - Y_C\Vert^2 +\bbE \Vert X_{-C} - Y_{-C}\Vert^2\\
    =&\bbE\left[\Vert X_C\Vert^2 + \Vert Y_C\Vert^2 - 2X_C^\top Y_C + \Vert X_{-C}\Vert^2 + \Vert Y_{-C}\Vert^2 - 2X_{-C}^\top Y_{-C}\right]\\
    =&\bbE\left[\Vert X\Vert^2 + \Vert Y\Vert^2 - 2X_C^\top Y_C - 2X_{-C}^\top Y_{-C} \right]\,,
\end{align*}
and since $X_C$ and $Y_C$ are independent, then the target function becomes
$$
\argmin_C \bbE[X_C^\top Y_C + X_{-C}^\top Y_{-C}] =\argmin \Vert \bbE X_C\Vert^2 + \Vert \bbE X_{-C}\Vert^2 \,.
$$
Note that $\bbE X_{-C} = \bbE X - \bbE X_{C}$, then we have
$$
\Vert \bbE X_C\Vert^2 + \Vert\bbE X_{-C}\Vert^2 = \Vert\bbE X_C + \bbE X_{-C}\Vert^2 - 2\bbE X_C^\top\bbE X_{-C}\,.
$$
Thus the goal is 
\begin{equation}
\argmax_C \bbE X_C^\top \bbE X_{-C}\,.    
\label{eq:prop5_goal}
\end{equation}
Now if $X \sim 0.5 N(\mu_1, \bfI) + 0.5 N(\mu_2, \bfI)$, then
\begin{align*}
    \bbE X_{-C} &= \pi(\mu_1\Phi(u_1) - a\phi(u_1)) + (1 - \pi)(\mu_2\Phi(u_2) - a\phi(u_2))\\
    &=\frac{1}{2}(\mu_1\Phi(u_1) + \mu_2\Phi(u_2) - a\phi(u_1) - a\phi(u_2))\triangleq \frac{1}{2}A\,,\\
    \bbE X_C &= \pi(\mu_1 - \mu_1\Phi(u_1)) + a\phi(u_1)) + (1-\pi)(\mu_2 - \mu_2\Phi(u_2) + a\phi(u_2))\\
    &=\frac{1}{2}(\mu_1 + \mu_2 - A)\,.
\end{align*} 
It follows that the goal \eqref{eq:prop5_goal} can be written as
\begin{equation}\label{eq:goalA}
\argmax_{C}\, (\mu_1+\mu_2-A)^\top A\triangleq \argmax_C\, f(A)\,,    
\end{equation}
where
$$
A = \mu_1\Phi(u_1) + \mu_2\Phi(u_2) - a\phi(u_1) - a\phi(u_2)\,.
$$
Note that the set is defined as $C\triangleq \{x: a^\top X > b, \Vert a\Vert^2 = 1\}$.
The hyperplane $a^\top X > b$ should pass the center of their mean, thus $b = a^\top(\mu_1 + \mu_2)$. It follows that
$$
u_1 = -u_2 = \frac{a^\top (\mu_2 - \mu_1)}{2a^\top a} \triangleq a^\top d\,,
$$
where $d = (\mu_2-\mu_1)/2$. Then we have $\phi(u_1) = \phi(u_2)$ and $\Phi(u_2) = 1-\Phi(u_1)$. It follows that
\begin{align*}
A &= \mu_1 \Phi(a^\top d) + \mu_2 (1-\Phi(a^\top d)) - 2\phi(a^\top d)a\\
&=-2d\Phi(a^\top d) + \mu_2 - 2a\phi(a^\top d)\,.
\end{align*}
Note that
\begin{align*}
\left(\frac{df}{dA}\right)^\top &= \mu_1 + \mu_2 - 2A  = \mu_1 - \mu_2 + 4d\Phi(a^\top d) + 4a\phi(a^\top d) \\
&=-2d + 4d\Phi(a^\top d) + 4a\phi(a^\top d)\\
&=-2(1-2\Phi(a^\top d))d + 4\phi(a^\top d))a\,,
\end{align*}
and
\begin{align*}
\frac{dA}{da} &= -2dd^\top\phi(a^\top d) -2( \phi(a^\top d) - aa^\top dd^\top\phi(a^\top d) )  \\
&= -2\phi(a^\top d)\left[dd^\top +\bI -aa^\top dd^\top \right]\,.
\end{align*}
By the chain rule, we have
\begin{align*}
    \frac{df}{da} &= \frac{df}{dA}\frac{dA}{da}\\
    &=4\phi(a^\top d)\{
    (1-2\Phi(a^\top d))d^\top  - 2\phi(a^\top d)a^\top +\\
    &\qquad 
    \left[
    (1-2\Phi(a^\top d)d^\top(\bfI-aa^\top)d-2\phi(a^\top d) a^\top (\bfI-aa^\top)d)
    \right]d^\top
    \}\\
    &\triangleq 4\phi(a^\top d)\left\{
    cd^\top -2\phi(a^\top d)a^\top
    \right\}\,.
\end{align*}
To have $\frac{df}{da} = 0$, then $a \propto d$. Since $\Vert a\Vert_2^2 = 1$, we have $a = \frac{d}{\Vert 
 d\Vert^2}$.
Thus, the optimal hyperplane is 
$$
\left(x - \frac{\mu_1+\mu_2}{2}\right)^\top (\mu_2 - \mu_1) > 0\,.
$$

Now for general case, if $X_1\sim N(\mu_1,\Sigma)$ and $X_2\sim N(\mu_2, \Sigma)$, then we have $\Sigma^{-1/2}X_1\sim N(\Sigma^{-1/2}\mu_1, \bfI)$ and $\Sigma^{-1/2}X_2\sim N(\Sigma^{-1/2}\mu_2, \bfI)$. For $X\sim 0.5 X_1 + 0.5X_2$, the optimal hyperplane is 
$$
\left(\Sigma^{-1/2}x - \Sigma^{-1/2}\frac{\mu_1+\mu_2}{2}\right)^\top (\Sigma^{-1/2}\mu_2 - \Sigma^{-1/2}\mu_1) > 0\,,
$$
that is
$$
\left(x-\frac{\mu_1+\mu_2}{2} \right)^\top \Sigma^{-1}(\mu_2-\mu_1) > 0\,.
$$

\end{proof}

\if1\isarxiv
\section{Proof of Proposition \ref{prop:power}}
\else
\section{Proof of Proposition 6}
\fi

% \begin{align*}
% \Pr(\hat G(Z) = 2\mid G(Z)=1) &= \Phi\left(
%     -\frac{\Delta}{2} +\frac{\log\frac{1-\pi}{\pi}}{\Delta}
%     \right)\,,\\
% \Pr(\hat G(Z) = 1\mid G(Z)=2) &= \Phi\left(
%     -\frac{\Delta}{2} -\frac{\log\frac{1-\pi}{\pi}}{\Delta}
%     \right)\,,
% \end{align*}
% where $\Delta = \delta^\top\Sigma^{-1}\delta$.
\subsection{(i)}
\begin{proof}
Note that
\begin{align*}
    &\Pr(\hat G(Z) = 2\mid G(Z)=1)\\
    & = \Pr\left(\left(Z - \frac{\xi+\eta}{2}\right)^\top\Sigma^{-1}\delta > 0 \mid G(Z) = 1\right)\\
    &= \Pr \left(
    Z^\top \Sigma^{-1}\delta > (\frac{\xi+\eta}{2})^\top\Sigma^{-1}\delta\mid G(Z)=1
    \right)\\
    &=\Pr\left(
    \frac{Z^\top\Sigma^{-1}\delta-\xi^\top\Sigma^{-1}\delta}{\sqrt{\delta^\top\Sigma^{-1}\delta}} > \frac{
    (\frac{\xi+\eta}{2})^\top\Sigma^{-1}\delta
    % - \log\frac{1-\pi}{\pi} 
    - \xi^\top\Sigma^{-1}\delta
    }{\sqrt{\delta^\top\Sigma^{-1}\delta}}\mid G(Z)=1
    \right)\\
    &= 1-\Phi\left(
    \frac{
    (\frac{\xi+\eta}{2})^\top\Sigma^{-1}\delta
    % - \log\frac{1-\pi}{\pi} 
    - \xi^\top\Sigma^{-1}\delta
    }{\sqrt{\delta^\top\Sigma^{-1}\delta}}
    \right)\\
    &=1-\Phi\left(
    \frac{1}{2}\sqrt{\delta^\top\Sigma^{-1}\delta} 
    % -\frac{\log\frac{1-\pi}{\pi}}{\sqrt{\delta^\top\Sigma^{-1}\delta}}
    \right)\\
    &=\Phi(-\Delta/2)\,.
    % &=\Phi\left(
    % -\frac{1}{2}\sqrt{\delta^\top\Sigma^{-1}\delta} +\frac{\log\frac{1-\pi}{\pi}}{\sqrt{\delta^\top\Sigma^{-1}\delta}}
    % \right)\triangleq 
    % \Phi\left(
    % -\frac{\Delta}{2} +\frac{\log\frac{1-\pi}{\pi}}{\Delta}
    % \right)\,,
\end{align*}
where $\Delta\triangleq \sqrt{\delta^\top\Sigma^{-1}\delta}$ and $\Phi(\cdot)$ is the CDF of the standard Normal distribution. Similarly,
$$
\Pr(\hat G(Z) = 1\mid G(Z)=2) = \Phi\left(
    -\frac{\Delta}{2}
    % -\frac{\log\frac{1-\pi}{\pi}}{\Delta}
    \right)\,.
$$
    
\end{proof}

\subsection{(ii)}
\begin{lemma}
Let
\begin{align*}
    X_1, \ldots, X_m \sim N(\xi, \sigma^2)\\
    Y_1, \ldots, Y_n \sim N(\eta, \sigma^2)\,.
\end{align*}
Suppose 
\begin{itemize}
    \item the first cluster consists of $m-k$ observations from $X$ and $k$ observations from $Y$. 
    \item the second cluster consists of $n-k$ observations from $Y$ and $k$ observations from $X$
\end{itemize}
The power function for the Z-test statistic is
$$
\beta(k) = \Phi\left(-k\frac{\delta}{\sigma}r^{1/2} +\frac{\delta}{\sigma}r^{-1/2}-c\right) + \Phi\left(k\frac{\delta}{\sigma}r^{1/2} -\frac{\delta}{\sigma}r^{-1/2}-c\right)\,,
$$
where $r = (m+n)/mn$ and $c = z_{1-\alpha/2}$.
\end{lemma}

\begin{proof}
Without loss  of generality, write 
\begin{align*}
\check{X} &= \{Y_1,\ldots, Y_k, X_{k+1},\ldots, X_{m}\}    \\
\check{Y} &= \{X_1,\ldots, X_{k}, Y_{k+1},\ldots, Y_{n}\} 
\end{align*}
Then 
\begin{align*}
\bar{\check X} = \frac{1}{m}\sum_{i=1}^m \check X_i &=\frac{\sum_{i=1}^k Y_i + \sum_{j=k+1}^m X_j}{m}\sim N\left(\frac{k\eta + (m-k)\xi}{m}, \frac{\sigma^2}{m} \right)\\
\bar{\check Y} = \frac{1}{n}\sum_{i=1}^n \check Y_i &=\frac{\sum_{i=1}^k X_i + \sum_{j=k+1}^n Y_j}{n}\sim N\left(\frac{k\xi + (n-k)\eta}{n}, \frac{\sigma^2}{n} \right)\,.
\end{align*}
Thus the test statistic
$$
Z = \frac{\bar{\check Y}-\bar{\check{X}}}{\sigma \sqrt{1/m + 1/n}} \sim  N\left(\frac{
(\eta-\xi)\left[1 -\frac{k(m+n)}{mn}\right]}{\sigma\sqrt{\frac{m+n}{mn}}}, 1
\right)
$$

Let $r = (m+n)/mn$ and $c = z_{1-\alpha/2}$. Then the power function is
$$
\beta(k) = \Phi\left(-k\frac{\delta}{\sigma}r^{1/2} +\frac{\delta}{\sigma}r^{-1/2}-c\right) + \Phi\left(k\frac{\delta}{\sigma}r^{1/2} -\frac{\delta}{\sigma}r^{-1/2}-c\right)\,.
$$
    
\end{proof}
\begin{proof}
Without loss of generality, assume $m < n$. 
Now $k$ follows $\bern(n, p_{e})$ with $p_e = \Phi(-\Delta/2)$, then
\begin{align*}
    \beta  = \bbE \beta(k) = \sum_{k=0}^m \beta(k)\binom{m}{k}p_{e}^k (1-p_{e})^{m-k}\,.
\end{align*}    
\end{proof}

\subsection{(iii)}
\begin{proof}
By the Taylor expansion,
\begin{align}
    \beta(k) & = \beta(\bbE k) + \beta'(\bbE k)(k - \bbE k) + \frac{1}{2}\beta''(z)(k-\bbE k)^2\,,\label{eq:betak_taylor}
\end{align}
where $z$ is some point between $k$ and $\bbE k$.
Note that
\begin{align*}
    \beta(k) &= \Phi(-rkb + b - c) + \Phi(rkb - b-c)\\
    \beta'(k) &= -rb\phi(-rkb+b-c) + rb\phi(rkb-b-c)\\
    &=rb\left[\phi(rkb-b-c) - \phi(-rkb+b-c)\right]\\
    \beta''(k) &= rb\left[-rb(rkb-b-c)\phi(rkb-b-c) + rb(-rkb+b-c)\phi(-rkb+b-c)\right]\\
    &=r^2b^2\left[(-rkb+b-c)\phi(-rkb+b-c) - (rkb-b-c)\phi(rkb-b-c)\right]\,.
\end{align*}
By Hoeffding's inequality, we have
$$
\Pr(\vert k - \bbE k\vert \ge m\epsilon) \le 2 \exp\left(-2m\epsilon^2\right)\,.
$$
Take $\epsilon = \sqrt{\frac{2\log m}{m}}$, then with a high probability of at least $1-2e^{-2m\epsilon^2} = 1-\frac{2}{m^4}$, we have 
$$
\vert k - \bbE k\vert \le m\epsilon\,,
$$
where $\bbE k = mp_e$, then
$$
p_e - \epsilon\le \frac{k}{m} \le p_e +\epsilon\,.
$$
It follows that
$$
rm(p_e - \epsilon)\le rk \le rm(p_e+\varepsilon)\,,  
$$
where
$$
rm = 1 +\frac{m}{n} \rightarrow 1 + \kappa\,.
$$
Hence $1-rk\rightarrow 1-(1+\kappa)p_e$ is a constant when $m\rightarrow\infty$. It follows that
\begin{align*}
    (rkb-b-c)\phi(rkb - b -c) &= -((1-rk)b+c)\phi( (1-rk)b +c)\\
    &=-((1-rk)b-c)\cdot\frac{(1-rk)b+c}{(1-rk)b-c}\cdot \phi( (1-rk)b -c)\exp(-2(1-rk)bc)\\
    &=((1-rk)b-c)\phi( (1-rk)b -c) \cdot O(\exp(-2(1-rk)bc))\,.
\end{align*}
Thus
\begin{align*}
    \beta''(k) = r^2b^2((1-rk)b-c)\phi( (1-rk)b -c) \cdot \left(1 + O(\exp(-2(1-rk)bc))\right)\,.
\end{align*}
% \begin{align*}
%     \phi(rkb - b -c) &= \phi((1-rk)b + c)\\
%     &=O(e^{-b^2})\\
%     &=O(e^{-r^{-1}}) = O(e^{-m})\,,
% \end{align*}
% and 
% \begin{align*}
%     (rkb-b-c)\phi(rkb - b -c) &= -((1-rk)b+c)\phi((1-rk)b + c) = O(b)\cdot O(e^{-m}) \\
%     &= O(\sqrt{m}e^{-m})\,.
% \end{align*}
% Thus,
% $$
% \beta''(k) = r^2b^2[(-rkb+b-c)\phi(-rkb+b-c) + O(\sqrt{m}e^{-m})]\,.
% $$
% For $f(x) = x\phi(x)$, the first derivative is
% $$
% f'(x) = \phi(x) -x^2\phi(x)\,,
% $$
% then when $x=\pm 1$, we have $f'(x) = 0$. It follows that
% $$
% \vert f(x)\vert \le f(1) = \frac{1}{\sqrt{2\pi}}\,,
% $$
% thus,
% $$
% \vert \beta''(z)\vert \le \frac{2}{\sqrt{2\pi}}\,.
% $$
Take expectation on \eqref{eq:betak_taylor},
\begin{align}
    \bbE \beta(k) &= \beta(\bbE k) +\frac{1}{2}\bbE \beta''(z) (k-\bbE k)^2\,.\label{eq:betak_taylor_expect}
    % &\le \beta(\bbE k) +\frac{1}{\sqrt{2\pi}}\bbE (k-\bbE k)^2\\
    % &= \beta(\bbE k) +\frac{1}{\sqrt{2\pi}} \Var(k)
\end{align}
% Note that
% \begin{align*}
%     \bbE \beta''(z) (k-\bbE k)^2 &= \sup_z \beta''(z) \bbE (k-\bbE k)^2\\
%     &= \sup_z\beta''(z)\Var(k)\\
%     &=\sup_z\beta''(z)mp_e(1-p_e)\,.
% \end{align*}
% And 
For the second term,
\begin{align*}
    \bbE \beta''(z) (k-\bbE k)^2 &= \bbE[\beta''(z) (k-\bbE k)^2\mid \vert k-\bbE k\vert \le m\epsilon]P(\vert k-\bbE k\vert \le m\epsilon) +\\
    &\qquad + \bbE[\beta''(z) (k-\bbE k)^2\mid \vert k-\bbE k\vert > m\epsilon]P(\vert k-\bbE k\vert > m\epsilon)\\
    % &= m^2\epsilon^2 [\bbE \beta''(z)\mid \vert k-\bbE k\vert \le m\epsilon](1-O(e^{-2m\epsilon^2})) + O(e^{-2m\epsilon^2})
    &=E_1+E_2\,.
    % &=\sup_{z: \vert k-\bbE k\vert \le m\epsilon}\beta''(z) \bbE(k-\bbE k)^2 + O(e^{-2m\epsilon^2})\,.
\end{align*}
For $E_2$, we have
\begin{align*}
    E_2 &\le m^2O(e^{-2m\epsilon^2}) = O(m^{-2})\,,
\end{align*}
and for $E_1$, we have
\begin{align*}
E_1 &\le \sup_{z: \vert k-\bbE k\vert \le m\epsilon}\beta''(z) \bbE(k-\bbE k)^2\\
& = mp_e(1-p_e) \cdot r^2b^2((1-rz)b-c)\phi( (1-rz)b -c) \cdot \left(1 + O(\exp(-2(1-rz)bc))\right)
% &=\frac{\delta^2}{\sigma^2}(1+\frac{m}{n}) \cdot((1-rz)b-c)\phi( (1-rz)b -c) \cdot \left(1 + O(\exp(-2(1-rz)bc))\right)\,,
\end{align*}
and
\begin{align*}
    E_1&\ge \inf_{z: \vert k-\bbE k\vert \le m\epsilon}\beta''(z) \bbE(k-\bbE k)^2\\
    &=mp_e(1-p_e) \cdot r^2b^2((1-rz)b-c)\phi( (1-rz)b -c) \cdot \left(1 + O(\exp(-2(1-rz)bc))\right)\,.
\end{align*}
Note that when $\vert k-Ek\vert \le m\epsilon$ and $z$ is a point between $k$ and $\bbE k$, then
$$
rm(p_e - \epsilon)\le rz \le rm(p_e+\epsilon)\,,
$$
then $rz \rightarrow (1+\kappa)p_e$.
% Consider $g(z) = ((1-rz)b-c)\phi((1-rz)b-c)$, then 
% \begin{align*}
%     g(z) &= g(mp_e) + g'(z_\xi) (z-mp_e)\\
%     &=g(mp_e)\left[1 + \frac{g'(z_\xi)}{g(mp_e)}(z-mp_e)\right]\\
%     &=g(mp_e)\left[1 - rb(((1-rz_\xi)b-c)^2-1) \frac{\phi((1-rz_\xi)b-c)}{\phi((1-mp_e)b-c)} (z-mp_e)\right]
% \end{align*}
% Let $\rho = rz$ and $\rho_0 = rmp_e$. 
% Consider $g(x) = x\phi(x)$, then 
% $$
% g(x) = g(x_0) + g'(x_\xi) (x-x_0) = g(x_0) + O(g'(x_\xi) (x-x_0))\,.
% $$
% That is
% $$
% g(x_0)
% $$
% then
% $$
% g(\rho) = g(\rho_0) + g'(\rho_\xi) (\rho -\rho_0) = g(\rho_0) + O(\sqrt{\frac{2\log m}{m}})\,.
% $$
% and
% Note that $1-rz\rightarrow 1-(1+\kappa)p_e$ when $\vert k-\bbE k\vert\le m\epsilon$, then
% $$
%  \beta''(z) = r^2b^2((1-rz)b-c)\phi( (1-rz)b -c) \cdot \left(1 + O(\exp(-2(1-rz)bc))\right)\,.
% $$
And the pdf $\phi(\cdot)$ exibits an expoential decay,
\begin{align*}
((1-rz)b -c)\phi((1-rz)b -c)   = O(\sqrt{m}e^{-m})\,.
\end{align*}
% $$
% \lim_{m\rightarrow\infty} \sup_{z: \vert k-\bbE k\vert \le m\epsilon}\beta''(z) = \frac{\delta}{m\sigma}[((1-p_e)b-c) + \phi((1-p_e)b-c)]
% $$
% And since $\bbE(k-\bbE k)^2=mp_e(1-p_e)$, then
Thus
\begin{align*}
    \bbE \beta''(z) (k-\bbE k)^2 &= O(r^2b^2)\cdot mp_e(1-p_e) \cdot O(\sqrt{m}e^{-m}) + O(m^{-2})\\
    &=O(m^{-1})\cdot mp_e(1-p_e)\cdot O(\sqrt{m}e^{-m})  + O(m^{-2})\\
    &=p_e(1-p_e)O(\sqrt{m}e^{-m}) + O(m^{-2})\\
    &=O(m^{-2})\,.
\end{align*}
Thus, \eqref{eq:betak_taylor_expect} becomes
$$
\bbE \beta(k) = \beta(\bbE k) + O(m^{-2})\,.
$$
Thus, with a high probability at least $1-\frac{2}{m^4}$, we have
\begin{equation}\label{eq:beta_highprob}
\beta
% = \Phi\left(\delta \left[1-\frac{[mp_{2\mid 1}](m+n)}{mn}\right] - c\right) + \Phi\left(-\delta \left[1-\frac{[mp_{2\mid 1}](m+n)}{mn}\right] - c\right)\\
=\Phi\left(\frac{\delta}{\sigma} \frac{1-mp_{e}r }{\sqrt{r} } - c\right) + \Phi\left(-\frac{\delta}{\sigma} \frac{1-mp_{e}r }{\sqrt{r} } - c\right) + O(m^{-2})\,.
\end{equation}
Note that $g(x) = \Phi(x+a) + \Phi(-x-a), x > 0$ is an increasing function for any fixed $a > 0$ because that
\begin{align*}
\frac{d}{dx}[\Phi(x-a)+\Phi(-x-a)] &= \frac{\exp(-(x-a)^2/2) - \exp(-(x+a)^2/2) }{\sqrt{2\pi}}\\
&= \frac{\exp(-(x^2+a^2)/2)(\exp(ax)-\exp(-ax) )}{\sqrt{2\pi}}\\
&> 0\,.
\end{align*}
Then if $\delta$ increases, $\Delta$ also increases, and hence $p_{e}$ decreases, which means the clustering accuracy increases, and finally the power $\beta$ increases.
    
\end{proof}

\subsection{(iv)}

\begin{proof}
    The power of the oracle case that there is no classification error is $\beta(0)$. Then the power loss for the case with $k$ errors is
    \begin{align*}
        & \beta(0) - \beta(k) \\
        &= \left[\Phi\left(\frac{\delta}{\sigma}r^{-1/2}-c\right) + \Phi\left(-\frac{\delta}{\sigma}r^{-1/2}-c\right)\right] - \left[\Phi\left(-k\frac{\delta}{\sigma}r^{1/2} +\frac{\delta}{\sigma}r^{-1/2}-c\right) + \Phi\left(k\frac{\delta}{\sigma}r^{1/2} -\frac{\delta}{\sigma}r^{-1/2}-c\right)\right]\\
        &=\left[
        \Phi\left(\frac{\delta}{\sigma}r^{-1/2}-c\right) - \Phi\left(-k\frac{\delta}{\sigma}r^{1/2} +\frac{\delta}{\sigma}r^{-1/2}-c\right)
        \right] +
        \left[
        \Phi\left(-\frac{\delta}{\sigma}r^{-1/2}-c\right) - \Phi\left(k\frac{\delta}{\sigma}r^{1/2} -\frac{\delta}{\sigma}r^{-1/2}-c\right)
        \right]\\
        &\triangleq \Delta_1 + \Delta_2\,.
    \end{align*}
    Let $b \triangleq \frac{\delta}{\sigma}r^{-1/2}$.
    For the first term $\Delta_1$, by the mean value theorem, there exists $u\in (-rkb+b-c, b-c)$
    % $$
    % u \in \left(-k\frac{\delta}{\sigma}r^{1/2} +\frac{\delta}{\sigma}r^{-1/2}-c, \frac{\delta}{\sigma}r^{-1/2}-c\right)
    % $$
    such that
    \begin{align*}
        % \Delta_1 = \phi(u) \cdot k\frac{\delta}{\sigma}r^{1/2}\,.
        \Delta_1 = \phi(u)\cdot rkb\,.
    \end{align*}
    Similarly, for $\Delta_2$, there exists $v\in (-b-c, rkb-b-c)$
    % $$
    % v \in \left(-\frac{\delta}{\sigma}r^{-1/2}-c, k\frac{\delta}{\sigma}r^{1/2} -\frac{\delta}{\sigma}r^{-1/2}-c\right)
    % $$
    such that
    $$
    % \Delta_2 = -\phi(v)\cdot k\frac{\delta}{\sigma}r^{1/2}\,.
    \Delta_2 = -\phi(v)\cdot rkb\,.
    $$
    Since $\phi(\cdot)$ is symmetric, let $\bar v \triangleq -v \in (-rkb+b+c, b+c)$, then 
    $$
    \Delta_2 = - \phi(\bar v)\cdot rkb\,.
    $$
    It follows that
    \begin{align*}
        \Delta_1 +\Delta_2 &= [\phi(u) - \phi(\bar v)]\cdot rkb\,.        
    \end{align*}
    Thus,
    \begin{align}
        \left[\phi(b-c)-\phi(-rkb+b+c)\right]\cdot rkb \le \Delta_1 + \Delta_2\le \left[\phi(-rkb+b-c)-\phi(b+c)\right]\cdot rkb \,.\label{eq:delta12_bounds}
    \end{align}
    In the lower bound, note that
    \begin{align*}
        \phi(-rkb+b+c) &= \frac{1}{\sqrt{2\pi}}\exp\left(-\frac{((1-rk)b+c)^2}{2}\right)\\
        % &=\frac{1}{\sqrt{2\pi}}\exp\left(-\frac{((1-rk)\delta/\sigma r^{-1/2}+c)^2}{2}\right)\\
        % &=O\left(\frac{1}{((1-rk)b+c)^2}\right)\\
        &=O(\exp(-b^2))\\
        % &=O\left(\frac{1}{b^2}\right)\\
        % &=O(r) = O\left(\frac{1}{m}\right)\,,
        &=O(\exp(-r^{-1})) =O(e^{-m})\,.
    \end{align*}
    % where we use the fact that $\exp(-x) = O(x^{-p})$. 
    When $k=mp_e$, denote
    $$
    \rho \triangleq rk = rmp_e = \left(1+\frac{m}{n}\right)p_e\,,
    $$
    then
    \begin{align*}
    \phi(-rkb+b+c) \cdot rkb &= O\left(e^{-m}\right) \cdot \rho \cdot O(b)\\ 
    &=O\left(\sqrt{m}e^{-m}\right)\,.
    \end{align*}
    Thus the bounds \eqref{eq:delta12_bounds} becomes
    $$
    \phi(b-c)\cdot \rho b + O\left(\sqrt{m}e^{-m}\right)\le \Delta_1 +\Delta_2 \le \phi((1-\rho)b-c)\cdot \rho b+ O\left(\sqrt{m}e^{-m}\right)\,.
    $$
    Incorporating \eqref{eq:beta_highprob}, we have
    $$
    \phi(b-c)\cdot \rho b + O\left(m^{-2}\right)\le \beta(0) - \beta \le \phi((1-\rho)b-c)\cdot \rho b+ O\left(m^{-2}\right)\,.
    $$
    % It implies that when the signal increases, $\rho$ decreases, and the lower bound decreases.
    % By the mean value theorem again, there exists $w\in (u, \bar v)$ such that
    % $$
    % \phi(u) - \phi(\bar v) = \phi'(w) (u - \bar v) = -w\phi(w)(u-\bar v) = w\phi(w)(\bar v-u)\,.
    % $$
\end{proof}

%%% TODO: clarify? discuss? (seems not)
% $mp_{2\mid 1}$ is not necessarily equal to $np_{1\mid 2}$, where $p_{2\mid 1} = p_{1\mid 2} = p_e$. Note that if $m < n$ and we adopt $p_{1\mid 2}$, we need to make sure $k \le m$. It will be a truncated Bernoulli. %% would it be true that
% $$
% \frac{\sum_{k=0}^m k\binom{n}{k}p_{1\mid 2}^k (1-p_{1\mid 2})^{n-k}}{\sum_{k=0}^m \binom{n}{k}p_{1\mid 2}^k (1-p_{1\mid 2})^{n-k}} \overset{?}{=}
% \sum_{k=0}^m k\binom{m}{k}p_{2\mid 1}^k (1-p_{2\mid 1})^{m-k}
% $$
% where $p_{2\mid 1} = p_{1\mid 2} = p_e$.

\if1\isarxiv
\section{Proof of Proposition \ref{prop:fdr_normal}}
\else
\section{Proof of Proposition 7}
\fi

\begin{proof}
Consider the test statistic
$$
T_j = \frac{(\bar X_j - \bar Y_j) - (\xi - \mu)}{\sqrt{2\sigma_j^2/n}}\sim N(0, 1)\,.
$$
As in \textcite{daiFalseDiscoveryRate2023}, we can decompose the variance of the number of false positives as follows:
\begin{align*}
    \Var(\sum_{j\in S_0}1(M_j > t)) = \sum_{j\in S_0}\Var(1(M_j > t)) + \sum_{i\neq j\in S_0}\Cov(1(M_i > t), 1(M_j > t))\,.
\end{align*}
Note that $1(M_j > t)$ can be viewed as a Bernoulli random variable, and hence its variance is not larger than 1, then the first term on the right-hand side is bounded by $p_0$. For the second term, 
\begin{align*}
\Cov(1(M_i > t), 1(M_j > t)) &= \bbE[(1(M_i > t)-\bbE 1(M_i > t))(1(M_j > t)-\bbE 1(M_j > t))]\\
&= \Pr(M_i > t, M_j > t) - P(M_i > t) P(M_j > t)\,.
\end{align*}

Consider the general form of the mirror statistic, in which function $f(u, v)$ is non-negative, symmetric about $u$ and $v$, and monotonically increasing in both $u$ and $v$. For any $t$ and $u\ge 0$, let
$$
I_t(u) = \inf\{v\ge 0: f(u, v) > t\}\,.
$$
Then
\begin{align}
P(M_i > t, M_j > t) &= P(T_i^{(2)} > I_t(T_i^{(1)}), T_j^{(2)} > I_t(T_j^{(1)}))\,.    
\end{align}
Note that $(T_i^{(2)}, T_j^{(2)})$ follows a bivariate Normal distribution with correlation $R_{ij}^0$.

Now consider the correlation for the bivariate distribution $(T_i^{(2)}, T_j^{(2)})$. For simplicity, we omit the superscript since we just need to focus on one part of the data without loss of generality. 
The covariance between $T_i$ and $T_j$ is
\begin{align}
    \Cov(T_i, T_j) &= \bbE T_iT_j - \bbE T_i\bbE T_j = \bbE T_iT_j\\
    &= \bbE\frac{\bar X_i\bar X_j - \bar Y_i\bar X_j - \bar X_i\bar Y_j +\bar Y_i \bar Y_j}{\sigma_i\sigma_j\left(\frac{1}{n_{21}}+ \frac{1}{n_{22}}\right)}\,.
\end{align}

Let $\vert \hat I^{(2)}_1\vert = n_{21}, \vert \hat I^{(2)}_2\vert = n_{22}$.
Note that
\begin{align*}
    \bbE \bar X_i\bar X_j = \frac{1}{n_{21}^2}\bbE\sum_{r=1}^{n_{21}} X_{ri} \sum_{s=1}^{n_{21}}X_{sj} = \frac{1}{n_{21}^2}\sum_{r=1}^{n_{21}}\bbE X_{ri}X_{rj} = \frac{1}{n_{21}} \Sigma_{ij} + \mu_i\mu_j
\end{align*}
and
$$
% \bbE\bar Y_i\bar X_j = \bbE\bar X_i\bar Y_j = 0\,,
\bbE\bar Y_i\bar X_j = \mu_i\mu_j\,,\quad 
\bbE \bar Y_i\bar Y_j = \frac{1}{n_{22}} \Sigma_{ij} + \mu_i\mu_j\,,
$$
It follows that
$$
\Cov(T_i, T_j) = \frac{R_{ij}}{\sigma_i\sigma_j} \triangleq R_{ij}^0\,.
$$

Under the regularization condition, for some $c > 0$, 
$$
1/c < \lambda_{\min}(\Sigma) \le \lambda_{\max}(\Sigma) < c\,,
$$
Let $\Vert R_{S_0}\Vert_1 = \sum_{i, j\in S_0}\vert R_{ij}\vert$ and $\Vert R_{S_0}\Vert_2 = (\sum_{i, j\in S_0}\vert R_{ij}\vert^2)^{1/2}$. Note that for any positive definite matrix $A\in \IR^{m\times n}, \lambda_{\min}(A)\le A_{ii} \le \lambda_{\max}(A)$ for $i\in \{1,\ldots, m\}$, then
$$
\Vert R_{S_0}^0\Vert_1\le \frac{1}{\lambda_{\min}(R_{S_0})}\Vert R_{S_0}\Vert_1\,.
$$
By Cauchy-Schwarz inequality,
$$
\Vert R_{S_0}\Vert_1 \le p_0 \Vert R_{S_0}\Vert_2\,.
$$
Note that the fact
$$
\sum_{i, j}A_{ij}^2 = \tr(A^\top A) = \sum_{i=1}^m \lambda_i^2(A)\,,
$$
then
$$
\Vert R_{S_0}\Vert_2 \le p_0^{1/2}\lambda_{\max}(R_{S_0})\,,
$$
Combine them together, we have
$$
\Vert R_{S_0}^0\Vert_1 \le p_0^{3/2}\lambda_{\max}(R_{S_0}) / \lambda_{\min}(R_{S_0}) = O_p(p_0^{3/2})\,.
$$
Thus, the second weak dependence assumption is satisfied. We can conclude that Proposition 2.1 of \textcite{daiFalseDiscoveryRate2023} also works in the clustering setting.
% As a counterexample, for the equi-correlation matrix,
% $$
% \Sigma = \rho \one\one^T + (1-\rho) \bfI\,,
% $$
% the eigenvalues are $\lambda_{\max}= (p-1)\rho+1$ and $\lambda_{\min} = 1-\rho$.

% When the sample size is large, the $t$-distribution is very close to the Normal distribution, so we can directly apply the $z$-test instead of the $t$-test.

\end{proof}

%%% omit draft and random thoughts
% \input{supp_draft}

\end{document}